\documentclass[reqno,10pt,centertags, draft]{amsart}
\usepackage{amsmath,amsthm,amscd,amssymb,latexsym,upref,mathrsfs}
\usepackage{color}
\usepackage{amsfonts,amssymb}
\usepackage{mathrsfs}
\usepackage{color}
\newcommand{\be}{\begin{equation}}
\newcommand{\ee}{\end{equation}}
\newcommand{\bed}{\begin{displaymath}}
\newcommand{\eed}{\end{displaymath}}

\newcommand{\dC}{{\mathbb{C}}}
\newcommand{\dR}{{\mathbb{R}}}

\newcommand{\cB}{{\mathcal B}}
\newcommand{\cC}{{\mathcal C}}
\newcommand{\cD}{{\mathcal D}}
\newcommand{\cE}{{\mathcal E}}

\newcommand{\cG}{{\mathcal G}}
\newcommand{\cH}{{\mathcal H}}

\newcommand{\cK}{{\mathcal K}}
\newcommand{\cL}{{\mathcal L}}

\newcommand{\cN}{{\mathcal N}}

\newcommand{\cP}{{\mathcal P}}

\newcommand{\cS}{{\mathcal S}}

\newcommand{\sH}{{\mathfrak H}}
\newcommand{\sS}{{\mathfrak S}}

\newcommand{\f}{\frac}

\def\senki{{\lbrack\negthinspace [\bot ]\negthinspace\rbrack}}
\def\senki+{{\lbrack\negthinspace [+] \negthinspace\rbrack}}

\renewcommand{\Re}{\mathop\mathrm{Re}}
\renewcommand{\Im}{\mathop\mathrm{Im}}
\renewcommand{\ge}{\geqslant}

\DeclareMathOperator{\dom}{dom}
\DeclareMathOperator{\ran}{ran}

\DeclareMathOperator{\sgn}{sgn}

\DeclareMathOperator*{\slim}{s-lim}
\DeclareMathOperator*{\wlim}{w-lim}

\DeclareMathOperator{\essinf}{ess\,inf}

\newcommand{\imag}{\operatorname{Im}}

\numberwithin{equation}{section}

\newtheorem{theorem}{Theorem}[section]
\newtheorem{proposition}[theorem]{Proposition}
\newtheorem{lemma}[theorem]{Lemma}
\newtheorem{corollary}[theorem]{Corollary}
\newtheorem{definition}[theorem]{Definition}

\theoremstyle{remark}
\newtheorem{remark}[theorem]{Remark}

\begin{document}

\title{{Scattering matrices and Dirichlet-to-Neumann maps}}
%\titlerunning{Scattering matrices and Dirichlet-to-Neumann maps}

\author[J.\ Behrndt]{Jussi Behrndt}
\address{Institut f\"ur Numerische Mathematik, Technische Universit\"at
Graz, Steyrergasse 30, 8010 Graz, Austria}
\email{behrndt@tugraz.at}
\urladdr{www.math.tugraz.at/~behrndt/}

\author[M.\,M. Malamud]{Mark M. Malamud}
\address{Institute of Applied Mathematics and Mechanics,
National Academy of Science of Ukraine,
Dobrovol's'kogo Str. 1,
84100 Slavyansk, 
Donetsk region,
Ukraine}
\email{mmm@telenet.dn.ua}

\author[H.\ Neidhardt]{Hagen Neidhardt}
\address{Weierstrass Institute for Applied Analysis and Stochastics,
Mohrenstr. 39, 10117 Berlin, Germany}
\email{neidhard@wias-berlin.de}
\urladdr{www.wias-berlin.de/~neidhard/}

\maketitle
\begin{abstract}
 A general representation formula for the scattering matrix of a scattering system consisting of two self-adjoint operators in terms of an abstract operator valued Titchmarsh-Weyl $m$-function
is proved. This result is applied to scattering problems for different self-adjoint realizations of
Schr\"{o}din\-ger operators on unbounded domains, Schr\"{o}d\-in\-ger operators with singular potentials supported on hypersurfaces, and orthogonal
couplings of Schr\"{o}dinger operators. In these applications
the scattering matrix is expressed in an explicit form with the help of Dirichlet-to-Neumann maps.
\end{abstract}

\section{Introduction}

Let $A$ and $B$ be self-adjoint operators in a Hilbert space $\sH$ and assume that the resolvent difference
\begin{equation}\label{rdintro}
 (B-\lambda)^{-1}-(A-\lambda)^{-1},\qquad\lambda\in\rho(A)\cap\rho(B),
\end{equation}
belongs to the ideal $\sS_1(\sH)$ of trace class operators. It is well known that in this situation the wave
operators $W_\pm(A,B)$ of the pair $\{A,B\}$
exist and are complete, and the scattering operator $S(A,B)=W_+(A,B)^*W_-(A,B)$ is unitarily equivalent to a multiplication operator induced by a
family $\{S(A,B;\lambda)\}_{\lambda\in\dR}$ of unitary operators $S(A,B;\lambda)$ in the spectral representation of the absolutely continuous part of $A$.
This family is called the scattering matrix of the scattering
system $\{A,B\}$ and is one of the most important quantities in the analysis of scattering
processes; we refer the reader to the monographs \cite{BW83,K76,RS79,W03,Y92} for more details.

The main objective of this paper is to express the scattering matrix of $\{A,B\}$ in terms of an abstract operator valued Titchmarsh-Weyl $m$-function, and to apply
this result to scattering problems for Schr\"{o}dinger operators. In order to explain our main abstract result Theorem~\ref{th:3.1}
consider the closed symmetric operator $S=A\cap B$ and note that $S$ has infinite defect
numbers whenever the resolvent difference of $A$ and $B$ in \eqref{rdintro} is infinite dimensional. The closure of the
operator $T=A\,\widehat + B$, where $\widehat +$ denotes the sum of subspaces in
$\sH\times\sH$, coincides with $S^*$ and clearly $A$ and $B$ are self-adjoint restrictions of $T$. This setting can be fitted in  the framework
of ($B$-)generalized boundary triples and their Weyl functions from \cite{DM95} and allows to introduce boundary maps $\Gamma_0$ and $\Gamma_1$ on $\dom (T)$,
which can be viewed as abstract analogs of the Dirichlet and Neumann trace operators (see also \cite{BL07,BL12,DHMS06,DHMS12}). For $\lambda\in\dC\setminus\dR$ one
defines the Weyl function $M$ via
\begin{equation*}
 M(\lambda)\Gamma_0 f_\lambda=\Gamma_1 f_\lambda,\qquad f_\lambda\in\ker(T-\lambda),
\end{equation*}
see Section~\ref{sec1} for the details.
In PDE applications $M(\lambda)$ is usually the Dirichlet-to-Neumann map (or its inverse, the Neumann-to-Dirichlet map) acting in some boundary space. Roughly speaking our main abstract
result states that the scattering matrix of $\{A,B\}$ is of the form
\begin{equation*}
 S(A,B;\lambda)=I-2i\sqrt{\Im M(\lambda+i0)}\,M(\lambda+i0)^{-1}\sqrt{\Im M(\lambda+i0)}
\end{equation*}
for a.e. $\lambda\in\dR$. This representation is a highly nontrivial generalization of a similar result from \cite{BMN08}, where the special case that the resolvent difference
in \eqref{rdintro} is a finite rank operator was treated in the context of ordinary boundary triples and their Weyl functions from \cite{DM91,DM95}, see also \cite{AP86},
\cite[Chapter 4]{AK99}, \cite[Chapter~3,~$\S$1]{Y92}, and \cite{BMN09} for related results and simple examples.
In contrast to the earlier results in the finite rank case the present representation formula is applicable
to scattering problems for Schr\"{o}dinger operators (or more general elliptic second order differential operators) on unbounded domains,
which we shall explain in more detail next.

In fact, our main motivation for establishing the general representation formula for the scattering matrix in Section~\ref{scatsec} in an abstract extension theory framework
is the applicability to scattering problems for Schr\"{o}dinger operators
with Dirichlet, Neumann, and Robin boundary conditions on exterior domains in $\dR^2$ and $\dR^3$ in Section~\ref{nrsec},
and orthogonal couplings of Schr\"{o}dinger operators, and Schr\"{o}dinger operators with singular potentials supported on curves and
hypersurfaces in $\dR^2$ and $\dR^3$ in Section~\ref{deltasec}. Let us first explain the situation for a scattering system consisting of a Neumann and a Robin realization;
for more details and a slightly more general situation see Section~\ref{nrsec_N-Roben}.
Denote the Dirichlet and Neumann trace operators by $\gamma_D$ and $\gamma_N$, respectively, and consider the self-adjoint operators
\begin{equation*}
 A f=-\Delta f +Vf,\qquad \dom (A) =\bigl\{f\in H^2(\Omega):\gamma_N f=0\bigr\},
\end{equation*}
and
\begin{equation*}
 B f=-\Delta f +Vf,\qquad \dom (B)=\bigl\{f\in H^2(\Omega):\alpha\gamma_D f=\gamma_N f\bigr\},
\end{equation*}
where $\alpha\in C^2(\partial\Omega)$ is real, the potential $V$ is real and bounded, and the domain  $\Omega$ is the complement of
a bounded set with a $C^\infty$-smooth boundary in $\dR^2$ or $\dR^3$. In this situation it is known from \cite{BLLLP10,G11}
that the resolvent difference of $A$ and $B$ satisfies the trace class condition \eqref{rdintro}.
If $\cN(\lambda)$, $\lambda\in\dC\setminus\dR$, denotes the Neumann-to-Dirichlet map, that is,
\begin{equation*}
 \cN(\lambda)\gamma_N f_\lambda=\gamma_D f_\lambda,\qquad -\Delta f_\lambda+Vf_\lambda=\lambda f_\lambda,
\end{equation*}
we obtain in Theorem~\ref{nrthm} that the scattering matrix of the scattering system $\{A,B\}$ admits the form
\begin{equation*}
 S(A,B;\lambda)=I_{\cG_\lambda}+2i\sqrt{\Im \cN(\lambda+i0)}\bigl(I-\alpha\cN(\lambda+i0)\bigr)^{-1}\alpha  \sqrt{\Im \cN(\lambda+i0)}
 \end{equation*}
 for a.e. $\lambda\in\dR$.
Here the space $L^2(\dR,d\lambda,\cG_\lambda)$, where $\cG_\lambda=\overline{\ran(\Im \cN(\lambda+i0))}$ for a.e. $\lambda\in\dR$,
forms a spectral representation of the absolutely continuous part of the Neumann operator $A_N$ and the limits $\Im \cN(\lambda+i0)$ and
$(I-\alpha\cN(\lambda+i0))^{-1}$ have to be interpreted in suitable operator topologies; cf. Theorem~\ref{nrthm} for details.
A similar result is proved in Theorem~\ref{th:4.1} for the pair consisting of the Dirichlet realization of $-\Delta+V$ and the Robin operator $B$ in $L^2(\dR^2)$;
here the trace class property \eqref{rdintro} for $n=2$ is due to Birman \cite{Bir62}.  
For some recent work on related spectral problems for Schr\"{o}dinger operators we refer the reader to \cite{AP04,BR16,BMNW08,GM08,GM09a,GM09,GM11,GMZ07,MalNei14,M04,P08,P15} and for more
general partial elliptic differential operators to \cite{AGW14,BL07,BL12,BLL13-2,BLLR16,BR15,BGW09,G08,G11a,G11b,G11,M10,MPS15,MPS16,P16,PR09}.

Our second set of examples in Section~\ref{deltasec} is a bit more involved. Here scattering systems consisting of the free Schr\"{o}dinger operator
\begin{equation}\label{aq1}
 A f=-\Delta f +Vf,\qquad \dom (A)=H^2(\dR^n),
\end{equation}
and orthogonal couplings of Schr\"{o}dinger operators with Dirichlet and Neumann boundary conditions, or
Schr\"{o}dinger operators with singular $\delta$-potentials of strength $\alpha\in L^\infty(\cC)$ supported on hypersurfaces $\cC$ which split $\dR^2$ or $\dR^3$
into a bounded smooth domain $\Omega_+$ and a smooth exterior domain $\Omega_-$ are studied. The latter operator is of the form
\begin{equation}\label{bq1}
\begin{split}
 B f&=-\Delta f +Vf,\\
 \dom (B)&=\left\{f=\begin{pmatrix}f_+\\f_-\end{pmatrix}\in H^{3/2}_\Delta(\dR^n\setminus\cC):
 \begin{matrix}\gamma_D^+f_+=\gamma_D^-f_-,\qquad\quad\!\\
\alpha\gamma_D^\pm f_\pm=\gamma_N^+ f_++\gamma_N^-f_-\end{matrix}\right\};
\end{split}
 \end{equation}
here  $H^{3/2}_\Delta(\dR^n\setminus\cC)$ is a subspace of $H^{3/2}(\Omega_+)\times H^{3/2}(\Omega_-)$ and
$\gamma_D^\pm$ and $\gamma_N^\pm$ denote the Dirichlet and Neumann trace operators on the interior and exterior domain; cf. Section~\ref{deltasec3} for the details.
Schr\"{o}dinger operators with  $\delta$-potentials play an important role in various physically relevant problems and have therefore attracted
a lot of attention. We refer the interested reader to the review paper \cite{E08}, to e.g. \cite{AKMN13,AGS87,BLL13,BEKS94,EI01,EK03,EK05,EY02} and the monographs \cite{AGHH05,AK99}
for more details and further references. We shall briefly discuss the scattering matrix for the pair of operators in \eqref{aq1}--\eqref{bq1};
for the pairs consisting of $A$ in \eqref{aq1} and the orthogonal sum of the Dirichlet or the Neumann realizations of $-\Delta+V$ on $\Omega_+$ and $\Omega_-$
see Theorem~\ref{Th_Scat_Mat_for_Dir-Free} and Theorem~\ref{nfthm}, respectively.
It follows from \cite{BLL13} that the above choice of $A$ and $B$ satisfies the trace class condition \eqref{rdintro} in dimensions $n=2$ and $n=3$ and we show in this situation in Theorem~\ref{dddthm}
that the scattering matrix is given by
\begin{equation*}
 S(A,B;\lambda)=I_{\cG_\lambda}+2i\sqrt{\Im \cE(\lambda+i0)}\bigl(I-\alpha\cE(\lambda+i0)\bigr)^{-1}\alpha  \sqrt{\Im \cE(\lambda+i0)},
 \end{equation*}
where the function $\cE$ is defined as
\begin{equation*}
 \cE(\lambda)=\bigl(\cD_+(\lambda)^{-1}+\cD_-(\lambda)^{-1}\bigr)^{-1},\qquad\lambda\in\dC\setminus\dR,
\end{equation*}
and $\cD_\pm(\lambda)$ denote the Dirichlet-to-Neumann maps corresponding to $-\Delta+V$ on the domains $\Omega_\pm$.
In this context we also refer the reader to related work by B.S. Pavlov and coauthors in \cite{BMPPY05,MPP07,PA05}, where scattering problems for certain couplings of
Schr\"{o}dinger operators were considered.

\subsection*{Notation} Throughout the paper $\sH$ and $\cG$ denote separable Hilbert spaces with scalar
product $(\cdot,\cdot)$. The linear space of bounded linear operators defined from $\sH$ to $\cG$ is denoted by
$\cB(\sH,\cG)$. For brevity we write $\cB(\sH)$ instead of $\cB(\sH,\sH)$. The ideal of compact operators is denoted by $\sS_\infty(\sH,\cG)$ and $\sS_\infty(\sH)$.
For $p>0$ the Schatten-von Neumann ideals are denoted by $\sS_p(\sH,\cG)$ and $\sS_p(\sH)$; they consist of all
compact operators $T$ with $p$-summable singular values $s_j(T)$ (i.e. eigenvalues of $(T^*T)^{1/2}$).
We shall also work with the operator ideals
\begin{equation*}
\cS_p(\sH,\cG)=\bigl\{T\in \sS_\infty(\sH,\cG) \,|\,  s_j(T) =O(j^{-1/p}) \, \text{as}\,
j\to \infty\bigr\},\quad p>0,
\end{equation*}
and we recall that
    \begin{equation}\label{idealipp}
\cS_p(\sH,\cG)\cdot \cS_q(\sH,\cG) = \cS_r(\sH,\cG), \quad \text{where} \quad \frac{1}{p}+\frac{1}{q}=\frac{1}{r}.
        \end{equation}

The resolvent set and the spectrum of a linear operator $A$ is denoted by $\rho(A)$ and
$\sigma(A)$, respectively. The domain, kernel and range of a linear operator $A$ are denoted by
$\dom(A)$, $\ker(A)$, and $\ran (A)$, respectively. By $\mathfrak B(\dR)$ we denote the Borel sets of $\dR$. The Lebesgue
measure on $\mathfrak B(\dR)$ is denoted by $d\lambda$.

A holomorphic function $M(\cdot): \dC_+ \longrightarrow \cB(\cH)$ is a Nevanlinna (or Herglotz or $R$-function) if its imaginary part
$\imag(M(z)) := \frac{1}{2i}(M(z) - M(z)^*)$, $z\in\dC_+$,  is a non-negative operator. Nevanlinna functions are extended to $\dC_-$ by $M(z):=M(\bar z)^*$, $z\in\dC_-$.
The class of $\cB(\cH)$-valued Nevanlinna functions is denoted by $R[\cH]$. A Nevanlinna function satisfying $\ker(\imag(M(z)) = \{0\}$ ($0\in\rho(\imag(M(z))$) for some,
and hence for all, $z\in\dC_+$, is said to be strict (uniformly strict, respectively). These subclasses are denoted by
$R^s[\cH]$ and $R^u[\cH]$, respectively.

\section{Self-adjoint extensions of symmetric operators and abstract Titchmarsh-Weyl $m$-functions}\label{sec1}

In the preparatory Section~\ref{sec1.1} we recall the notion of boundary triples and
their Weyl functions from extension theory of symmetric operators, and we introduce the
concept of $\sS_p$-regular Weyl functions in Section~\ref{sec2.2}. This notion is
important and useful for our purposes since it is directly related (and
in some situations equivalent) to the $\sS_p$-property of the resolvent difference of certain
self-adjoint extensions.

\subsection{$B$-generalized  boundary triples and their Weyl functions}\label{sec1.1}

In this subsection we review the notion of generalized (or $B$-generalized) and ordinary boundary triples
from extension theory of symmetric operators, and we introduce a new concept, the so-called double $B$-generalized boundary
triples in Definition~\ref{B-Generalized_BT} below. We refer the reader to \cite{BMN02,BGP08,DHMS06,DM91,DM95,GG,Sch12}
for more details on ordinary and $B$-generalized boundary triples, see also \cite{BL07,BL12,Cal39} for related notions.

In the following  $S$  denotes a densely defined, closed, symmetric operator in a
separable Hilbert space $\sH$.

  \begin{definition}[\cite{DM95}]\label{B-Generalized_BT}
A triple $\Pi = \{\cH,\Gamma_0,\Gamma_1\}$ is called a {\em $B$-generalized boundary
triple} for $S^*$ if $\cH$ is a Hilbert space and for some operator $T$ in $\sH$ such
that $\overline{T} = S^*$, the linear mappings $\Gamma_0,\Gamma_1: \dom(T)
\longrightarrow \cH$ satisfy the abstract Green's identity
\begin{equation}\label{Green_formula_for_BGBtrip}
(Tf,g)-(f,Tg)=(\Gamma_1 f,\Gamma_0 g)-(\Gamma_0 f,\Gamma_1 g), \quad f,g\in\dom (T),
\end{equation}
the operator $A_0:=T\upharpoonright\ker(\Gamma_0)$ is  self-adjoint  in $\mathfrak H$, and $\ran(\Gamma_0)=\cH$ holds.

If, in addition, the  operator $A_1:=T\upharpoonright\ker(\Gamma_1)$ is  self-adjoint in
$\mathfrak H$ and $\ran(\Gamma_1)=\cH$, then the triple $\Pi =
\{\cH,\Gamma_0,\Gamma_1\}$ is called a {\em double $B$-generalized boundary triple} for
$S^*$.
\end{definition}
We note that a $B$-generalized boundary triple for $S^*$ exists if and only if $S$ admits self-adjoint extensions in $\sH$, that is, the deficiency
indices of $S$ coincide. Furthermore, if $\Pi = \{\cH,\Gamma_0,\Gamma_1\}$ is a $B$-generalized boundary triple for $S^*$ then
\begin{equation*}
 \dom (S)=\ker(\Gamma_0)\cap\ker(\Gamma_1)
\end{equation*}
holds, the mappings $\Gamma_0,\Gamma_1: \dom(T) \longrightarrow \cH$ are closable when viewed as linear operators from $\dom S^*$ equipped with the graph norm
to $\cH$, and $\ran(\Gamma_1)$ turns out to be dense in $\cH$; cf. \cite[Section 6]{DM95}

The notion of double $B$-generalized boundary triples is inspired by the fact that the mappings in the so-called transposed
triple $\Pi^\top:=\{\cH,\Gamma_1,-\Gamma_0\}$
satisfy the abstract Green's identity but since in general neither $A_1=T\upharpoonright\ker(\Gamma_1)$ is self-adjoint nor $\ran(\Gamma_1)=\cH$ holds the transposed triple $\Pi^\top$
is not a $B$-generalized boundary triple in general. In fact, a $B$-generalized boundary triple $\Pi= \{\cH,\Gamma_0,\Gamma_1\}$ for $S^*$ 
is a double $B$-generalized boundary triple for $S^*$  if and only if
the transposed triple $\Pi^\top=\{\cH,\Gamma_1,-\Gamma_0\}$ is also a $B$-generalized boundary triple for $S^*$.

In some of the proofs of the results in Section~\ref{sec2.2} we shall also make use of the notion of ordinary boundary triples, which
we recall here for the convenience of the reader.

\begin{definition}
 A triple $\Pi = \{\cH,\Gamma_0,\Gamma_1\}$ is called an {\em ordinary boundary triple} for $S^*$
if $\cH$ is a Hilbert space, the linear mappings
$\Gamma_0,\Gamma_1: \dom(S^*) \longrightarrow \cH$ satisfy the abstract Green's identity
\begin{equation}\label{Green_formula_for_BGBtripOBT}
(S^*f,g)-(f,S^*g)=(\Gamma_1 f,\Gamma_0 g)-(\Gamma_0 f,\Gamma_1 g), \quad f,g\in\dom (S^*),
\end{equation}
and the mapping $\Gamma=(\Gamma_0,\Gamma_1)^\top:\dom(S^*)\rightarrow\cH\times\cH$ is surjective.
\end{definition}

Observe that any ordinary boundary triple is automatically a double $B$-genera\-lized
boundary triple;  the converse is not true in general. Ordinary boundary
triples are an efficient tool in extension theory of symmetric operators. In particular,
if $\Pi = \{\cH,\Gamma_0,\Gamma_1\}$ is an ordinary boundary triple for $S^*$, then all
closed proper extensions $\widetilde S\subset S^*$ of $S$ in $\sH$ can be
parametrized by means of the set  of closed linear relations in
$\cH$ via
\begin{equation}\label{para}
 \widetilde S \mapsto \Theta :=\bigl\{\{\Gamma_0 f,\Gamma_1 f\}:f\in\dom(\widetilde S)\bigr\}
\subset \cH\times \cH 
\end{equation}
We write $\widetilde S = S_\Theta$. If $\Theta$ is an operator then \eqref{para}
takes the form
$$
S_\Theta =S^*\upharpoonright \ker(\Gamma_1 - \Theta\Gamma_0)
$$
One verifies $(S_\Theta)^*=S_{\Theta^*}$ and hence the self-adjoint extensions of $S$ in
$\sH$ correspond to the self-adjoint relations  $\Theta$ in $\cH$. We shall use that
$\Theta$ in \eqref{para} is an operator (and not a multivalued linear relation) if and
only if the extension $S_\Theta$ and $A_0=S^*\upharpoonright\ker(\Gamma_0)$ are
disjoint, that is, $A_0\cap S_\Theta=S$.

Next we recall the notions and some important properties of $\gamma$-fields and Weyl functions.
For an ordinary boundary triple they go back to \cite{DM85,DM91}, for $B$-genera\-lized
boundary triples we refer the reader to \cite{DM95}. In the following let $\{\cH,\Gamma_0,\Gamma_1\}$ be a
$B$-generalized  boundary triple for $S^*$; the special case of an ordinary boundary triple is then covered as well.
Observe first  that  for each
$z\in\rho(A_0)$, $A_0 = T \upharpoonright \ker(\Gamma_0)$,  the following direct sum decomposition  holds
\begin{equation}\label{eq:2.8}
  \dom (T) = \dom (A_0)\,\dot +\,\ker(T-z)
  = \ker(\Gamma_0)\,\dot +\,\ker(T-z).
\end{equation}
Hence the restriction of the mapping $\Gamma_0$ to $\ker(T-z)$ is injective.
\begin{definition}[\cite{DM95}]\label{def_weyl_function}
Let $\Pi = \{\cH,\Gamma_0,\Gamma_1\}$ be a $B$-generalized  boundary triple. The {\em
$\gamma$-field} $\gamma(\cdot)$ and the {\em Weyl function} $M(\cdot)$ corresponding to
$\Pi$ are defined by
\begin{equation*}
  \gamma(z):=\bigl(\Gamma_0\upharpoonright\ker(T-z)\bigr)^{-1}\quad\text{and}\quad  M(z) := \Gamma_1 \gamma(z),
  \quad z \in \rho(A_0),
\end{equation*}
respectively.
   \end{definition}
It follows from \eqref{eq:2.8} that for $z\in\rho(A_0)$ the values $\gamma(z)$ of the $\gamma$-field and the values $M(z)$ of the Weyl
function are both well defined linear operators on $\ran(\Gamma_0) = \cH$. Moreover, $\gamma(z)\in\cB(\cH, \sH)$
maps onto $\ker(T- z)\subset\ker(S^*-z)\subset\sH$ and for all $z,\xi\in\rho(A_0)$ the relations
\begin{equation}\label{gform1}
 \gamma(z) = \bigl(I + (z - \xi)(A_0 - z)^{-1}\bigr)\gamma(\xi)=(A_0-\xi)(A_0-z)^{-1}\gamma(\xi)
\end{equation}
and
\begin{equation}\label{gstar}
 \gamma(z)^*=\Gamma_1(A_0-\bar{z})^{-1}\in\cB(\sH, \cH)
\end{equation}
hold. In particular, $\ran(\gamma(z)^*)=\ran(\Gamma_1\upharpoonright\dom (A_0))$ does not
depend on the point $z\in\rho(A_0)$ and
\begin{equation*}
 \bigl(\ran\gamma(z)^*\bigr)^\bot=\ker\gamma(z)=\{0\}
\end{equation*}
shows that $\ran(\gamma(z)^*)$ is dense in $\cH$ for all $z\in\rho(A_0)$. Furthermore, it follows from \eqref{gform1}
that $\gamma(\cdot)$ is holomorphic on $\rho(A_0)$.

The values of the Weyl function $M(\cdot)$ are operators in  $\cB(\cH)$ and $M(z)$ maps $\cH$ into the dense
subspace $\ran(\Gamma_1)\subset\cH$.
The Weyl function and the $\gamma$-field are related by the identity
\begin{equation}\label{gutgut}
 M(z) - M(\xi)^* =  (z-\bar\xi)\gamma(\xi)^*\gamma(z), \quad z,\xi\in\rho(A_0),
\end{equation}
and, in particular, $M(\bar z)=M(z)^*$ for all $z\in\rho(A_0)$.
It follows from \eqref{gform1} and  \eqref{gutgut} that $M(\cdot)$ is holomorphic on $\rho(A_0)$.
Setting $\xi=z$ in \eqref{gutgut} one gets
\begin{equation}\label{imm}
\Im M(z) = \frac{1}{2i}(M(z) - M(z)^*) =  (\Im z)\,\gamma(z)^*\gamma(z)
\end{equation}
and hence $\Im M(z)\geq 0$ for $z\in\dC_+$.
This identity also yields
$$\ker (\Im M(z)) = \ker (\gamma(z)) = \{0\},\qquad  z\in \mathbb C_{\pm},$$
and together
with the holomorphy of $M(\cdot)$ on $\rho(A_0)$ we conclude that  $M(\cdot)$ is a so-called {\em strict} Nevanlinna function with values in $\cB(\cH)$; we shall denote this by
$M(\cdot)\in R^s[\cH]$.
If $\Pi$ is a double $B$-generalized boundary triple then the Weyl function corresponding to the transposed $B$-generalized
boundary triple $\Pi^\top=\{\cH,\Gamma_1,-\Gamma_0\}$ is given
by $-M(\cdot)^{-1}$ and also belongs to the class $R^s[\cH]$, in particular, for $z\in\rho(A_0)\cap\rho(A_1)$ the values $M(z)$ of the Weyl function
of a double $B$-generalized boundary triple are bounded and boundedly invertible operators.

If $\Pi$ is an ordinary boundary triple
then the operators $\gamma(z)$ are boundedly invertible when viewed as operators from $\cH$ onto $\ker(S^*-z)$. In this case it follows from \eqref{imm}
that $\Im M(z)$ is a uniformly positive operator for $z\in\dC_+$, and hence the Weyl function corresponding to an ordinary boundary triple belongs to the class  $R^u[\cH]$
of the so-called {\em uniformly strict} Nevanlinna functions with values in $\cB(\cH)$; cf. \cite{DHMS06}.

\subsection{Resolvent comparability  and $\sS_p$-regular Weyl functions}\label{sec2.2}

Let $\Pi = \{\cH,\Gamma_0,\Gamma_1\}$ be a $B$-generalized boundary triple for $S^*$
with the corresponding Weyl function $M(\cdot)$, and let $A_0 =
S^*\upharpoonright\ker(\Gamma_0)$ and $A_1=S^*\upharpoonright\ker(\Gamma_1)$. It is
important to characterize the property of the resolvent comparability of the operators
$A_0$  and $A_1$  in terms of the Weyl function $M(\cdot)$.  To this end we  introduce
the notion of $\sS_p$-regular Nevanlinna functions in the next definition.
\begin{definition}\label{sregm}
A Nevanlinna function $M(\cdot)\in R[\cH]$ is called $\sS_p$-regular for some $p\in(0,\infty]$
if it admits a  representation
\be\label{eq:2.8a}
M(z) = C + K(z), \qquad K(\cdot): \dC_+ \longrightarrow \sS_p(\cH), \quad z \in \dC_+,
\ee
where $C\in\cB(\cH)$ is a self-adjoint operator such that $0\in\rho(C)$ and $K(\cdot)$ is a strict Nevanlinna function with values in $\cB(\cH)$, that is, $K(\cdot) \in R^s[\cH]$.
The class of $\sS_p$-regular Nevanlinna functions is denoted by $R^{\rm reg}_{\sS_p}[\cH]$.
\end{definition}
In other words, a Nevanlinna function is $\sS_p$-regular if it differs from a strict Nevanlinna function with values in $\sS_p$ by a bounded and boundedly invertible
self-adjoint constant.

\begin{lemma}\label{lem:2.6}
If $M(\cdot) \in R^{\rm reg}_{\sS_p}[\cH]$ for some $p\in(0,\infty]$, then
$-M(\cdot)^{-1} \in R^{\rm reg}_{\sS_p}[\cH]$.
\end{lemma}
\begin{proof}
Since $M(\cdot) \in R^{\rm reg}_{\sS_p}[\cH]$ for some $p\in(0,\infty]$, there exists a
boundedly invertible self-adjoint operator $C$ and a strict Nevanlinna function
$K(\cdot)\in R^s[\cH]$ such that
\begin{equation}\label{mck}
M(z)=C+K(z),\qquad z\in\dC_+.
\end{equation}
Observe first that $\ker (M(z))=\{0\}$ holds for all $z\in\dC_+$. In fact, $M(z)\varphi=0$ yields
$((C+\Re K(z))\varphi,\varphi)=0$ and $(\Im K(z)\varphi,\varphi)=0$, and as $K(\cdot)$ is strict we conclude $\varphi=0$ from the latter.
Furthermore, as $0\in \rho(C)$ and $K(z)\in\sS_p(\cH)$ it follows from the Fredholm alternative (see, e.g. \cite[Corollary to Theorem~VI.14]{RS72}) that
$0\in \rho(M(z))$ for all $z \in \dC_+$. It is clear that
\begin{equation}\label{mdl}
-M(z)^{-1}=D+L(z),\qquad z\in\dC_+,
\end{equation}
holds with $L(z):=C^{-1}-M(z)^{-1}$, $z\in\dC_+$ and the boundedly invertible self-adjoint operator $D:=-C^{-1}$. Since
\bed
L(z)=C^{-1} - M(z)^{-1} = C^{-1}K(z)M(z)^{-1}, \quad z \in \dC_+,
\eed
and $K(z)\in\sS_p(\cH)$ we conclude $L(z) \in \sS_p(\cH)$, $z \in \dC_+$. Moreover, as
$C^{-1}$ is a bounded self-adjoint operator one gets
\begin{equation*}
 \begin{split}
  \imag L(z) = \imag \bigl(-M(z)^{-1}\bigr) 
 = (M(z)^*)^{-1}\bigl(\imag K(z)\bigr)M(z)^{-1}, \quad z\in\dC_+,
 \end{split}
\end{equation*}
%
%
%%for $z\in\dC_+$,
where in the last equality we have used \eqref{mck}. As $K(\cdot) \in
R^s[\cH]$ by assumption we have $\ker(\Im K(z)) = \{0\}$ and this yields $\ker(\Im L(z))
= \{0\}$ for all $z \in \dC_+$. We have shown that
$L(\cdot):\dC_+\longrightarrow\sS_p(\cH)$ is a strict Nevanlinna function, $L(\cdot) \in
R^s[\cH]$, and hence it follows from \eqref{mdl} that $-M^{-1}(\cdot) \in R^{\rm
reg}_{\sS_p}[\cH]$.
\end{proof}

The assertions in the next lemma on the boundary values of $\sS_1$-regular Nevanlinna functions follow from well-known results due
to Birman and {\`E}ntina \cite{BirEnt67}, de Branges \cite{dB62}, and Naboko \cite{1987Naboko}; cf. \cite[Theorem 2.2]{GMN99}.

\begin{lemma}\label{lemmilein}
Let $M(\cdot)$ be an $\sS_1$-regular Nevanlinna function, $M(\cdot) \in R^{\rm reg}_{\sS_1}[\cH]$. Then the following assertions hold.
\begin{itemize}
 \item [{\rm (i)}] $M(\lambda+i0)  = \lim_{\varepsilon \rightarrow  +0} M(\lambda + i\varepsilon)$ exists for a.e. $\lambda\in\dR$ in the norm of $\cB(\cH)$;
 \item [{\rm (ii)}] $M(\lambda+i0)$ is boundedly invertible in $\cH$ for a.e. $\lambda\in\dR$;
 \item [{\rm (iii)}] $M(\lambda + i\varepsilon) - M(\lambda + i0) \in \sS_p(\cH)$ for $p\in(1,\infty]$, $\varepsilon>0$ and a.e. $\lambda\in\dR$, and
   \begin{equation*}
\lim_{\varepsilon \rightarrow +0} \|M(\lambda+ i\varepsilon) - M(\lambda+i0)\|_{\sS_p(\cH)} =0;
  \end{equation*}
 \item [{\rm (iv)}] $\Im M(\lambda+i0) = \lim_{\varepsilon \rightarrow  +0} \Im M(\lambda + i\varepsilon)$ exists for a.e. $\lambda\in\dR$ in the
 $\sS_1$-norm.
\end{itemize}
\end{lemma}
\begin{proof}
By assumption there exists a Nevanlinna function $K(\cdot)$ with values in $\sS_1(\cH)$ such that $M(z)=C+K(z)$, $z\in\dC_+$, holds with some bounded and
boundedly invertible self-adjoint operator $C$. It follows from \cite{BirEnt67,dB62,1987Naboko} (see, e.g. \cite[Theorem 2.2]{GMN99}) that
the limit $K(\lambda+i0)$ exists for a.e. $\lambda\in\dR$ in the $\sS_p$-norm for all $p>1$, and that the limit $\Im K(\lambda+i0)$ exists for a.e. $\lambda\in\dR$
in the $\sS_1$-norm. This yields assertions (i), (iii), and (iv).

In order to prove (ii) we recall that $-M(\cdot)^{-1}$ is $\sS_1$-regular by
Lemma \ref{lem:2.6} and hence the boundary values $M(\lambda+i0)^{-1}$ exist for a.e. $\lambda\in\dR$ in the operator norm. Hence (ii) follows from the
identity
\bed
M(\lambda + i\varepsilon)M(\lambda + i\varepsilon)^{-1} = M(\lambda + i\varepsilon)^{-1}M(\lambda + i\varepsilon) =  I_{\cH}, \quad \lambda \in \dR,
\eed
after passing to the limit $\varepsilon\rightarrow +0$ in the operator norm.
\end{proof}

In the next lemma we investigate $B$-generalized boundary triples with $\sS_p$-regular Weyl functions. In particular, it turns out that the symmetric
extension $A_1=T\upharpoonright\ker(\Gamma_1)$ is self-adjoint and a Krein type resolvent formula is obtained; cf. \cite{BL12,BLL13-2,DM91,DM95}.

%In the following we restrict the considerations to the ideals $\sS_p(\cH)$, $p \in (0,\infty]$.
%
%
\begin{proposition}\label{prop:2.4}
Let $\Pi = \{\cH,\Gamma_0,\Gamma_1\}$ be a  $B$-generalized boundary triple  for $S^*$ such that
the corresponding Weyl function  $M(\cdot)$ is $\sS_p$-regular for some $p\in(0,\infty]$.
Then the following assertions hold.

\begin{itemize}
  \item [{\rm (i)}]
$\Pi$ is a double $B$-generalized boundary triple for $S^*$;

  \item [{\rm (ii)}]
The Weyl function corresponding to the transposed $B$-generalized boundary triple $\Pi^\top= \{\cH,\Gamma_1,-\Gamma_0\}$ is $\sS_p$-regular;

  \item [{\rm (iii)}]
The operators $A_0$ and $A_1$ are $\sS_p$-resolvent comparable and
 \begin{equation}\label{p-resol_compar_cond_new}
(A_1- z)^{-1} - (A_0 - z)^{-1}=-\gamma(z)M(z)^{-1}\gamma(\bar z)^*\in\sS_p(\sH)
\end{equation}
holds for all $z\in \rho(A_0)\cap\rho(A_1)$.
\end{itemize}
\end{proposition}
\begin{proof}
(i) Since the Weyl function $M(\cdot)$ is $\sS_p$-regular by assumption,
 Lemma~\ref{lem:2.6} implies, in particular, that $M(z)^{-1}\in\cB(\cH)$ for all
$z\in\dC\setminus\dR$. This yields $\ran(\Gamma_1)=\ran (M(z))=\cH$. Next we check that
$A_1=T\upharpoonright\ker(\Gamma_1)$ is self-adjoint in $\sH$. First of all it follows
from the abstract Green's identity \eqref{Green_formula_for_BGBtrip} that $A_1$ is
symmetric. Let $z\in\dC\setminus\dR$, fix $f\in\sH$ and consider
\begin{equation*}
 h:=(A_0-z)^{-1}f-\gamma(z)M(z)^{-1}\gamma(\bar z)^*f.
\end{equation*}
From Definition~\ref{def_weyl_function} and \eqref{gstar} we obtain
\begin{equation*}
\Gamma_1 h=\Gamma_1(A_0-z)^{-1}f-\Gamma_1\gamma(z)M(z)^{-1}\gamma(\bar z)^*f= 0
\end{equation*}
and hence $h\in\dom (A_1)$. Since $\ran \gamma(z)\subset  \ker(T-z))$  one gets
$$
(A_1-z)h=(T-z)\bigl((A_0-z)^{-1}f-\gamma(z)M(z)^{-1}\gamma(\bar z)^*f\bigr) =f
$$
and we conclude the Krein type resolvent formula \eqref{p-resol_compar_cond_new}
in (iii) and $\ran(A_1-z)=\cH$ for $z\in\dC\setminus\dR$. Hence the symmetric operator
$A_1$ is self-adjoint in $\sH$ and it follows that $\Pi$ is a double $B$-generalized
boundary triple for $S^*$.

(ii) The Weyl function corresponding to the transposed
$B$-generalized boundary triple $\Pi^\top = \{\cH,\Gamma_1, -\Gamma_0\}$ is given by
\begin{equation}
 M^\top(z)=-M(z)^{-1},\qquad z\in\rho(A_0)\cap\rho(A_1),
\end{equation}
which is $\sS_p$-regular by Lemma \ref{lem:2.6}.

(iii)  Since $M(\cdot)$ is $\sS_p$-regular it follows that $\Im M(z)\in\sS_p(\cH)$ for $z \in \dC\setminus\dR$
and hence $\gamma(z)^*\gamma(z)\in \sS_p(\cH)$ by \eqref{imm}. This implies $\gamma(z)\in\sS_{2p}(\cH,\sH)$
and $\gamma(z)^*\in\sS_{2p}(\sH,\cH)$ for $z\in\dC\setminus\dR$, and the resolvent formula in \eqref{p-resol_compar_cond_new}
together with $0\in \rho(M(z))$,  $z \in \dC\setminus\dR$, yields the $\sS_p$-property of the resolvent difference in \eqref{p-resol_compar_cond_new}
for $z\in\dC\setminus\dR$, and hence for all $z\in\rho(A_0)\cap\rho(A_1)$.
\end{proof}

Proposition~\ref{prop:2.4}~(iii)  admits the following  useful improvement.
   \begin{corollary}\label{lem:strong_res_compar_cond}
Let $\Pi = \{\cH,\Gamma_0,\Gamma_1\}$ be a  $B$-generalized boundary triple  for $S^*$ such that
the corresponding Weyl function  $M(\cdot)$ is $\sS_\infty$-regular and assume that $\Im M(z)\in\sS_p(\cH)$ for some $p\in(0,\infty)$ and $z\in\dC_+$. Then
 \begin{equation}\label{p-resol_compar_cond_third}
(A_1 - z)^{-1} - (A_0 - z)^{-1}\in\sS_p(\sH), \quad
z\in \rho(A_0)\cap\rho(A_1).
 \end{equation}
   \end{corollary}
\begin{proof}
The assumption $\Im M(z)\in\sS_p(\cH)$ for some $p\in(0,\infty)$ and $z\in\dC_+$ together with \eqref{imm}
yields  $\gamma(z)^* \gamma(z) \in\sS_p(\cH)$, and hence
$\gamma(z) \in\sS_{2p}(\cH,\sH)$. The Krein type formula in \eqref{p-resol_compar_cond_new} implies \eqref{p-resol_compar_cond_third} for $z\in\dC_+$,
and hence also for all $z\in\rho(A_0)\cap\rho(A_1)$.
  \end{proof}
  Next we show that the $p$-resolvent comparability condition \eqref{p-resol_compar_cond_new}
guarantees  the  existence of  a $B$-generalized boundary triple such that the corresponding
Weyl function is $\sS_p$-regular.
\begin{proposition}\label{prop:2.5}
Let $A$ and $B$ be self-adjoint operators in  $\sH$ and assume that the closed symmetric
operator $S=A\cap B$ is densely defined. Then
$$
\dom(A)+\dom(B)
$$
is dense in $\dom (S^*)$ with respect to the graph norm and the following
assertions hold.

\begin{itemize}
  \item [{\rm (i)}]
 There is a  $B$-generalized boundary triple $\Pi = \{\cH,\Gamma_0,\Gamma_1\}$ for $S^*$ such that
   \begin{equation}\label{Covering_cond}
A = T\upharpoonright\ker(\Gamma_0) = A_0 \qquad \text{and} \qquad B = T\upharpoonright\ker(\Gamma_1) = A_1.
   \end{equation}

\item [{\rm (ii)}]  If for some $z\in\dC\setminus\dR$ and some $p\in (0,\infty]$ the condition
 \begin{equation}\label{p-resol_compar_cond}
  (B- z)^{-1}-(A - z)^{-1}\in\sS_p(\sH)
 \end{equation}
is satisfied, then there exists a double $B$-generalized boundary triple $\Pi=
\{\cH,\Gamma_0,\Gamma_1\}$ such that \eqref{Covering_cond} holds and the  corresponding
Weyl function $M(\cdot)$ is $\sS_p$-regular.
\end{itemize}
\end{proposition}
\begin{proof}
In order to see that $\dom(A)+\dom(B)$ is dense in $\dom(S^*)$ with respect to the graph norm assume that $h\in\dom(S^*)$
is such that
\begin{equation*}
 (f_A+f_B,h)+\bigl(S^*(f_A+f_B),S^*h\bigr)=0\quad\text{for all}\,\, f_A\in\dom(A),\, f_B\in\dom(B).
\end{equation*}
Then $(Af_A,S^*h)=(f_A,-h)$ and $(Bf_B,S^*h)=(f_B,-h)$ for all $f_A\in\dom(A)$ and
$f_B\in\dom(B)$ yield $S^*h\in\dom(A)\cap\dom(B)=\dom(S)$ and $(I+ SS^*)h=0$. Since the operator $I+ SS^*$ is uniformly positive one gets $h=0$,
that is, $\dom(A)+\dom(B)$ is dense in $\dom(S^*)$ with respect to the graph
norm.

(i)  Observe first that $S=A\cap B$ is a densely defined, closed, symmetric operator with equal deficiency
indices. Hence there exists an ordinary boundary triple
 $\Pi' = \{\cH, \Gamma'_0, \Gamma'_1\}$ for $S^*$ such that $B = S^*\upharpoonright\ker (\Gamma'_0)$; cf.
 \cite{DM95,DM85}. Furthermore, as
 $A$ and $B$  are  disjoint self-adjoint extensions of $S$ there exists a
self-adjoint  operator  $\Theta = \Theta^* \in \cC(\cH)$ such that
$$
A = S^*\upharpoonright \dom (A), \qquad   \dom (A) =
\ker (\Gamma'_1 - \Theta\Gamma'_0),
$$
see, e.g. \cite[Proposition 1.4]{DM95}.
We consider the mappings
\begin{equation*}
 \Gamma_0  := \Gamma'_1 - \Theta\Gamma'_0\quad\text{and}\quad\Gamma_1  := - \Gamma'_0
\end{equation*}
defined on
\begin{equation*}
 \dom(\Gamma_0)=\dom(\Gamma_1):=\dom(A) + \dom(B),
\end{equation*}
and  set
\begin{equation*}
T := S^*\upharpoonright \dom(T), \qquad  \dom(T) := \dom(A) + \dom(B).
\end{equation*}
We claim that $\Pi=\{\cH,\Gamma_0,\Gamma_1\}$ is a $B$-generalized boundary triple for
$S^*$ such that \eqref{Covering_cond} holds. Note first that $A =
T\upharpoonright\ker(\Gamma_0) = A_0$, $B = T\upharpoonright\ker(\Gamma_1) = A_1$, and
that $A$ and $B$ are disjoint self-adjoint extensions of $S$ by construction.
Therefore the argument in the beginning of the proof implies that $\dom(T)= \dom(A) + \dom(B)$ is dense in
$\dom(S^*)$ equipped with the graph norm and hence $\overline{T} = S^*$. Moreover, since
$\Theta=\Theta^*$ and the abstract Green's identity \eqref{Green_formula_for_BGBtripOBT}
holds for the ordinary boundary triple $\Pi'$ we obtain for  $f,g\in\dom (T)$
\begin{equation*}
 \begin{split}
 (\Gamma_1 f,\Gamma_0 g)-(\Gamma_0f,\Gamma_1 g)&=\bigl(-\Gamma_0'f,(\Gamma'_1 - \Theta\Gamma'_0)g\bigr)-\bigl((\Gamma'_1 - \Theta\Gamma'_0)f,-\Gamma_0'g\bigr) \\
 &=(\Gamma_1' f,\Gamma_0' g)-(\Gamma_0'f,\Gamma_1' g) %%\\ &
=(Tf,g)-(f,Tg),
 \end{split}
\end{equation*}
that is, the abstract Green's identity \eqref{Green_formula_for_BGBtrip} holds. In order
to verify $\ran(\Gamma_0)=\cH$ fix $h\in\cH$. Since $\Pi'$ is an ordinary boundary
triple there exists  $f_0\in \dom(B) = \ker (\Gamma'_0)$ such that $\Gamma'_1 f_0 = h$.
We then obtain
$$
\Gamma_0 f_0 = (\Gamma'_1 - \Theta \Gamma'_0)f_0 =  \Gamma'_1 f_0 = h,
$$
and hence $\ran(\Gamma_0)=\cH$. Summing up, we have shown that $\Pi$ is a $B$-generalized boundary
triple such that \eqref{Covering_cond} holds.

(ii)  Now we choose  an ordinary boundary triple
 $\Pi'' = \{\cH, \Gamma_0'', \Gamma_1''\}$ for $S^*$ such that $A = S^*\upharpoonright\ker (\Gamma_0'')$.
Since   $A$ and $B$  are  disjoint extensions of $S$  there exists an operator  $\Theta = \Theta^* \in \cC(\cH)$ such that
\begin{equation}\label{b--}
B = S^*\upharpoonright \dom (B), \qquad   \dom (B) =
\ker (\Gamma''_1 - \Theta\Gamma''_0).
\end{equation}
It follows from \cite[Theorem 2]{DM91} that the condition \eqref{p-resol_compar_cond} is equivalent to the condition
$(\Theta-\xi)^{-1}\in\mathfrak S_{p}(\cH)$ for all $\xi\in\rho(\Theta)$. In particular, $\rho(\Theta)\cap\mathbb R\not =\emptyset$,
and in the following we assume without loss of generality that $0\in\rho(\Theta)$. Denote the spectral function of the self-adjoint operator $\Theta$ by
$E_\Theta(\cdot)$, let
 $\sgn(\Theta)=\int_\dR\sgn(t)dE_\Theta(t)$
and recall the polar decomposition
     \begin{equation*}
 \Theta = |\Theta|^{1/2}\sgn(\Theta)|\Theta|^{1/2} = \sgn(\Theta)|\Theta|=\vert\Theta \vert \sgn(\Theta).
     \end{equation*}
As $\Theta^{-1}\in\sS_p(\cH)$ we have $|\Theta|^{-1/2} \in \sS_{2p}(\cH)$ and $\ker(|\Theta|^{-1/2}) = \{0\}$. We consider the  mappings
   \begin{equation}\label{WScat_Def-n_of_Gamma}
\Gamma_0 := |\Theta|^{1/2}\Gamma''_0 \quad \text{and}\quad
\Gamma_1 :=|\Theta|^{-1/2}(\Gamma''_1 - \Theta\Gamma''_0)
   \end{equation}
defined on
    \begin{equation}\label{2.8MMM}
\dom(\Gamma_0)=\dom(\Gamma_1) := \bigl\{f\in\dom (S^*):\Gamma''_0  f\in\dom(|\Theta|^{1/2})\bigr\}.
    \end{equation}
We set
    \begin{equation*}
T := S^*\upharpoonright \dom(T), \qquad  \dom(T) := \dom(\Gamma_0) = \dom(\Gamma_1),
   \end{equation*}
and we claim that $\Pi=\{\cH,\Gamma_0,\Gamma_1\}$ is a double $B$-generalized boundary triple for $S^*$. First of all we have for $f,g\in\dom (T)$
\begin{equation*}
\begin{split}
 &(\Gamma_1 f,\Gamma_0 g)-(\Gamma_0 f,\Gamma_1g)\\
 &\quad =\bigl(|\Theta|^{-1/2}(\Gamma''_1 - \Theta\Gamma''_0)f,|\Theta|^{1/2}\Gamma''_0g\bigr)-\bigl(|\Theta|^{1/2}\Gamma''_0 f,|\Theta|^{-1/2}(\Gamma''_1 - \Theta\Gamma''_0)g\bigr)\\
 &\quad =\bigl((\Gamma''_1 - \Theta\Gamma''_0)f,\Gamma''_0g\bigr)-\bigl(\Gamma''_0 f,(\Gamma''_1 - \Theta\Gamma''_0)g\bigr)\\
 &\quad =(\Gamma''_1 f,\Gamma''_0g)-(\Gamma''_0 f,\Gamma''_1 g)\\
\end{split}
\end{equation*}
and since $\Pi''$ is an ordinary boundary triple the abstract Green's identity \eqref{Green_formula_for_BGBtrip} follows.
The condition $\ran(\Gamma_0)=\cH$ is satisfied
since $0\in\rho(\Theta)$, and thus also $0\in\rho(\vert\Theta\vert^{1/2})$.
It is also clear from the definition of $\Gamma_0$ in \eqref{WScat_Def-n_of_Gamma}-\eqref{2.8MMM} that
\begin{equation}\label{aa}
\ker(\Gamma_0)  = \ker(\Gamma''_0) = \dom(A).
\end{equation}
Next it will be shown that
\begin{equation}\label{bb}
 \ker(\Gamma_1) =\dom(B)
\end{equation}
holds. In fact, the inclusion $\ker(\Gamma_1)\subset\dom (B)$ in \eqref{bb} follows from the definition of $\Gamma_1$ in  \eqref{WScat_Def-n_of_Gamma}-\eqref{2.8MMM} and
$\ker(|\Theta|^{-1/2}) = \{0\}$. For the remaining inclusion let $f\in\dom (B)$. Then $\Gamma_1'' f=\Theta\Gamma_0'' f$ by \eqref{b--} and, in particular,
$$\Gamma_0''f\in\dom(\Theta)\subset\dom(\vert\Theta\vert^{1/2}).$$
Hence $\dom(B)\subset\dom(T)$ and $\Gamma_1 f=0$ is clear, that is,
$\dom(B)\subset\ker(\Gamma_1)$ and thus \eqref{bb} is shown. Combining  \eqref{aa}
with  \eqref{bb} yields  \eqref{Covering_cond}. Moreover, we have $\overline{T} = S^*$
since
$$ \dom(A) + \dom(B)=\ker(\Gamma_0)+\ker(\Gamma_1)\subset\dom(T)
$$
and $\dom(A) + \dom(B)$ is dense in $\dom(S^*)$ equipped with the graph norm (as $A$ and
$B$ are disjoint self-adjoint extensions of $S$). Summing up, we have shown that $\Pi=
\{\cH, \Gamma_0, \Gamma_1\}$ is a $B$-generalized boundary triple for $S^*$ such that
\eqref{Covering_cond} holds.

It remains to verify that the Weyl function corresponding to $\Pi$
is $\sS_p$-regular; Proposition~\ref{prop:2.4}~(i) then implies that $\Pi$ is a double $B$-generalized boundary triple.
For this denote the Weyl function corresponding to the ordinary boundary
triple  $\Pi'' $ by  $M''(\cdot)$  and recall that $M''(z)\Gamma_0''f_z=\Gamma_1'' f_z$ for $f_z\in\ker(S^*-z)$ and $z\in\rho(A)$.
We claim that the Weyl function corresponding to $\Pi$ is given by
\begin{equation}\label{weyli}
 M(z) = |\Theta|^{-1/2}M''(z)|\Theta|^{-1/2} - \sgn(\Theta), \quad z \in \rho(A).
\end{equation}
In fact, for $f_z\in\ker(T-z)$ we compute
\begin{equation*}
 \begin{split}
  \bigl(|\Theta|^{-1/2}&M''(z)|\Theta|^{-1/2} - \sgn(\Theta)\bigr)\Gamma_0f_z\\
  &\qquad\qquad=|\Theta|^{-1/2}M''(z)\Gamma_0'' f_z - \sgn(\Theta)\vert\Theta\vert^{1/2}\Gamma_0''f_z\\
  &\qquad\qquad=|\Theta|^{-1/2}\bigl(\Gamma_1''f_z - |\Theta|^{1/2}\sgn(\Theta)\vert\Theta\vert^{1/2}\Gamma_0''f_z\bigr)\\
  &\qquad\qquad=|\Theta|^{-1/2}\bigl(\Gamma_1''f_z - \Theta\Gamma_0''f_z\bigr) 
  = \Gamma_1 f_z
 \end{split}
\end{equation*}
and hence \eqref{weyli} follows by Definition~\ref{def_weyl_function}. Let $K(z) :=
|\Theta|^{-1/2}M''(z)|\Theta|^{-1/2}$,  $z \in \dC_+$  and  let $C := -\sgn(\Theta)$.
Note that $C$ is a boundedly invertible self-adjoint operator and that
$|\Theta|^{-1/2}\in\sS_{2p}(\cH)$ and $M''(z)\in\cB(\cH)$ yield $K(z) \in \sS_p(\cH)$,
$z \in \dC_+$. Moreover, as $M''(\cdot) \in R^u[\cH]$ it follows that $K(\cdot) \in
R^s[\cH]$, and hence the Weyl function $M(\cdot)$ is $\sS_p$-regular. 
\end{proof}
In  applications to scattering problems it is important to know
whether the resolvent $p$-comparability condition \eqref{p-resol_compar_cond_new}, \eqref{p-resol_compar_cond} yields the $\sS_p$-regularity of the Weyl function.
 Apparently a converse statement to Proposition~\ref{prop:2.4}  is false
 for arbitrary double $B$-generalized boundary triples,
 while Proposition \ref{prop:2.5}  ensures  the existence of such a double $B$-generalized boundary triple.
However in the following proposition  we present an affirmative answer to this question
under certain  additional explicit assumptions.
  \begin{proposition}\label{prop:2.7}
Let $A$  and $B$ be self-adjoint operators in $\sH$ such that
 \begin{equation}\label{p-resol_compar_cond2}
  R_{B,A}(z):=(B- z)^{-1}-(A - z)^{-1}\in\sS_p(\sH)
 \end{equation}
holds for some  $z\in\dC\setminus\dR$ and some $p\in (0,\infty]$,
and assume that the closed symmetric operator $S=A\cap B$ is densely defined.
Assume, in addition, that there exists  $\lambda_0 \in \rho(A)\cap \rho(B)\cap \dR $ such that
  \begin{equation}\label{positive_res_differ}
\pm R_{B,A}(\lambda_0) \ge 0.
  \end{equation}
If $\Pi = \{\cH,\Gamma_0,\Gamma_1\}$  is  a double $B$-generalized boundary triple for
$S^*$ such that
condition \eqref{Covering_cond} holds then the corresponding Weyl function
$M(\cdot)$ is $\sS_p$-regular.
\end{proposition}
   \begin{proof}
Since $\Pi$ is a double $B$-generalized boundary triple the values of the Weyl function $M(\cdot)$ and the function $-M(\cdot)^{-1}$ are in $\cB(\cH)$.
Moreover, the assumption $\lambda_0 \in \rho(A)\cap \rho(B) \cap \mathbb R$  ensures
that $-M(\lambda_0)^{-1}\in\cB(\cH)$ is a self-adjoint operator and we have
\begin{equation}\label{Krein_type_for-la}
R_{B,A}(\lambda_0)  = (B - \lambda_0 )^{-1}-(A - \lambda_0)^{-1} = -\gamma(\lambda_0)M(\lambda_0)^{-1} \gamma(\lambda_0)^*
\end{equation}
by Proposition~\ref{prop:2.4}~(iii).
Assume that $R_{A,B}(\lambda_0)\geq 0$ in \eqref{positive_res_differ}. Then by \eqref{Krein_type_for-la}
\begin{equation*}
(R_{A,B}(\lambda_0)f,f)=\bigl(-M(\lambda_0)^{-1} \gamma(\lambda_0)^*f,\gamma(\lambda_0)^*f\bigr)\geq 0,\quad f\in\sH,
\end{equation*}
and since $\ran(\gamma(\lambda_0)^*)$ is dense in $\cH$ (see Section \ref{sec1.1}) we have $-M(\lambda_0)^{-1}\geq 0$.
Setting $T(\lambda_0) := \gamma(\lambda_0) (-M(\lambda_0))^{-1/2}\in\cB(\cH,\sH)$ and using the assumption \eqref{p-resol_compar_cond2} for some, and hence for all, $z\in\rho(A)\cap\rho(B)$
we conclude from
\eqref{Krein_type_for-la} that
\bed
R_{B,A}(\lambda_0) = T(\lambda_0)T(\lambda_0)^*
\in \mathfrak S_{p}(\sH).
\eed
This relation yields   $T(\lambda_0)^*  \in  \mathfrak S_{2p}(\sH,\cH)$ and  $T(\lambda_0) \in  \mathfrak S_{2p}(\cH,\sH)$,
and hence $\gamma(\lambda_0) = T(\lambda_0)(-M(\lambda_0))^{1/2}\in  \mathfrak S_{2p}(\cH,\sH)$. It then follows from \eqref{gform1} that
$$\gamma(z) \in
\mathfrak S_{2p}(\cH,\sH)\quad\text{and}\quad \gamma(\xi)^* \in \mathfrak
S_{2p}(\sH,\cH),\qquad z,\xi\in \rho(A).$$ Combining this with  \eqref{gutgut}  implies
$M(z) - M(\lambda_0) \in \mathfrak S_{p}(\cH)$. Therefore, setting $C := M(\lambda_0)$
and $K(z) := M(z) - M(\lambda_0)$, $z \in \dC_+$, we arrive at the representation
\eqref{eq:2.8a}.  Note that $C = M(\lambda_0)$ is a boundedly invertible self-adjoint
operator. Furthermore, since $\Im K(z)=\Im M(z)$ and $M(\cdot) \in R^s[\cH]$ we conclude
$K(\cdot) \in R^s[\cH]$, that is, the Weyl function $M(\cdot)$ is $\sS_p$-regular.
  \end{proof}
\begin{remark}\label{rem_on_Sp-ideals}
{\rm Condition \eqref{positive_res_differ} is satisfied if the symmetric operator $S = A
\cap B$ is semibounded from below and $A$ is chosen to be its Friedrichs
extension. In this case  \eqref{p-resol_compar_cond2} yields the semiboundedness of the
operator $B$ and the inequality \eqref{positive_res_differ} holds for any $\lambda_0$
smaller than the lower bound of $B$. }
\end{remark}
\begin{remark}
{\rm
The density of  $\dom(A)+\dom(B)$ in $\frak H$ under the conditions of
Proposition \ref{prop:2.5} is well known (see for instance \cite{DM85}). The simple proof
presented here and which does not exploit the second Neumann formula  seems to be
new.}
\end{remark}
\begin{remark}
{\rm Proposition \ref{prop:2.4}(i) can also be viewed as an immediate consequence from the fact that
the values of $M^{-1}(\cdot)$ are in $\cB(\cH)$;  cf. \cite{DHMS06,DM95}. For the convenience of the reader
we have presented a simple direct proof.
}
\end{remark}

\section{A representation of the scattering matrix}\label{scatsec}

Let $A$ and $B$ be self-adjoint operators in a Hilbert space $\sH$ and assume that they
are resolvent comparable, i.e. their  resolvent  difference
is a trace class operator,
\begin{equation}\label{trace}
 (B- i)^{-1}-(A - i)^{-1}\in\sS_1(\sH).
\end{equation}
Denote by $\sH^{ac}(A)$ the
absolutely continuous subspace of $A$ and let $P^{ac}(A)$ be the orthogonal projection
in $\sH$ onto $\sH^{ac}(A)$. In accordance with the Birman-Krein theorem, under the assumption \eqref{trace}
the {\it wave operators}
\begin{equation*}
 W_\pm(A,B):=s-\lim_{t\rightarrow\pm\infty} e^{itB} e^{-itA} P^{ac}(A)
\end{equation*}
exist and are complete, i.e. the ranges of $W_\pm(B,A)$ coincide with the absolutely continuous subspace $\sH^{ac}(B)$ of $B$;
cf. \cite{BW83,K76,RS79,W03,Y92}.
The {\it scattering operator} $S(A,B)$ of the {\it scattering system} is defined by
  \begin{equation*}
 S(A,B)  = W_+(A,B)^*W_-(A,B).
  \end{equation*}
The operator  $S(A,B)$ commutes with $A$ and is  unitary  in $\sH^{ac}(A)$, hence it  is unitarily equivalent
to a multiplication operator induced by a family
$\{S(A,B;\lambda)\}_{\lambda\in\dR}$ of unitary operators in a spectral representation of
the absolutely continuous part $ A^{ac}$ of $A$,
\begin{equation*}
 A^{ac}:=A\upharpoonright \dom (A)\cap\sH^{ac}(A).
\end{equation*}
The family $\{S(A,B;\lambda)\}_{\lambda\in\dR}$ is called the {\it scattering matrix} of the scattering system $\{A,B\}$.

In Theorem~\ref{th:3.1} and Corollary~\ref{cor:3.2} below we shall provide a representation of the scattering matrix
$\{S(A,B;\lambda)\}_{\lambda\in\dR}$ of the system $\{A,B\}$ in an extension theory
framework using $B$-generalized  boundary triples and their Weyl functions. It is
assumed that the closed
symmetric operator $S=A\cap B$ is densely defined; in the more general
framework of non-densely defined symmetric operators this assumption can be dropped.
First we discuss the case that $S=A\cap B$ is simple, i.e. $S$ does not contain a self-adjoint part or, equivalently, the condition
\begin{equation*}
\sH=\text{clsp} \bigl\{\ker(S^*-z): z\in \dC\setminus\dR\bigr\}
\end{equation*}
is satisfied; cf. \cite{K49}.
In the sequel the abbreviation a.e. means "almost everywhere with respect to the Lebesgue measure".
\begin{theorem}\label{th:3.1}
Let $A$ and $B$ be self-adjoint  operators in a Hilbert space $\sH$,
assume that the closed symmetric operator $S = A\cap B$ is densely defined and simple, and let
$\Pi = \{\cH,\Gamma_0,\Gamma_1\}$ be a  $B$-generalized boundary triple for $S^*$ such that
$A=T\upharpoonright\ker(\Gamma_0)$ and $B = T\upharpoonright\ker(\Gamma_1)$.
Assume, in addition, that  the Weyl function $M(\cdot)$ corresponding to $\Pi$ is $\sS_1$-regular.

Then
$\{A,B\}$ is a complete scattering system and
\begin{equation*}
 L^2(\dR,d\lambda,\cH_\lambda),\qquad \cH_\lambda:=\overline{\ran (\Im M(\lambda+i0))},
\end{equation*}
forms a spectral representation of $A^{ac}$ such that for a.e. $\lambda \in \dR$
the scattering matrix $\{S(A,B;\lambda)\}_{\lambda\in\dR}$ of the scattering system $\{A,B\}$ admits
the representation
\begin{equation*}
S(A,B;\lambda) = I_{\cH_\lambda}-2i\sqrt{{\Im M(\lambda + i0)}} \,M(\lambda + i0)^{-1} \sqrt{{\Im M(\lambda + i0)}}.
\end{equation*}
\end{theorem}
\begin{proof}
The proof of Theorem~\ref{th:3.1} consists of three separate steps and is essentially based on Theorem~\ref{att}.
Parts of the proof follow the lines in \cite[Proof of Theorem 3.1]{BMN09},
where the special case of a symmetric operator
$S$ with finite deficiency indices  was  treated.

First of all we note that the $\sS_1$-regularity assumption on $M(\cdot)$ together with Proposition~\ref{prop:2.4}~(iii)
ensures that the resolvent difference of $A$ and $B$ is a trace class operator.
Hence the wave operators $W_\pm(A,B)$  exist and are complete and $\{A,B\}$ is a complete scattering system, see, e.g. \cite[Theorem VI.5.1]{Y92}.
\\[-2ex]

\noindent {\it Step 1.} According to Proposition~\ref{prop:2.4}~(iii) the resolvent difference of $A$ and $B$ in
\eqref{trace} can be written in a Krein type resolvent formula of the form
\begin{equation}\label{krein}
 (B-z)^{-1}-(A-z)^{-1} =-\gamma(z)M(z)^{-1}\gamma(\bar z)^*,
\qquad z\in\rho(A)\cap\rho(B).
     \end{equation}
In particular, from \eqref{krein} and \eqref{gform1} we get
\begin{equation*}
\begin{split}
 (B-i)^{-1}&-(A-i)^{-1} = -{\gamma(i)}\,{M(i)^{-1}}\gamma(-i)^*\\
 &=-(A+i)(A-i)^{-1}{\gamma(-i)}\,{M(i)^{-1}}\gamma(-i)^* =\phi(A)C G C^*
\end{split}
 \end{equation*}
where
\begin{equation}\label{WScat3.2A}
 \phi(t) :=\frac{t+i}{t-i},\quad t\in\dR,\qquad C := {\gamma(-i)}\quad\text{and}\quad G := -{M(i)^{-1}}.
\end{equation}
We claim that the condition
   \begin{equation}\label{okja}
\sH^{ac}(A)= \text{clsp}\,\bigl\{E_A^{ac}(\delta)\ran C:\delta\in\mathfrak B(\dR)\bigr\}
  \end{equation}
in Theorem~\ref{att} is satisfied. In fact, since $S$ is assumed to be simple we have
\begin{equation*}
 \sH=\text{clsp}\,\bigl\{\ker(S^*-z):z\in\dC\setminus\dR\bigr\}.
\end{equation*}
Furthermore, using
$\ker(S^*-z)=\overline{\ker(T-z)}$,  $z\in\dC\setminus\dR$, which follows from \eqref{eq:2.8},
and $\ran(\gamma(z))=\ker(T-z)$, $z\in\dC\setminus\dR$, it follows that
\begin{equation*}
\begin{split}
 \sH&=\text{clsp}\,\bigl\{\ker(T-z):z\in\dC\setminus\dR\bigr\}\\
    &=\text{clsp}\,\bigl\{\gamma(z)h:z\in\dC\setminus\dR,\, h\in \cH\bigr\}\\
    &=\text{clsp}\,\bigl\{(A+i)(A-z)^{-1}\gamma(-i)h:z\in\dC\setminus\dR,\, h\in \cH\bigr\}\\
    &=\text{clsp}\,\bigl\{(A+i)(A-z)^{-1}C h:z\in\dC\setminus\dR,\,h\in \cH \bigr\}\\
    &=\text{clsp}\,\bigl\{E_A(\delta) C h:h\in \cH,\,\delta\in\mathfrak B(\dR)\bigr\}
\end{split}
    \end{equation*}
    and hence
\begin{equation*}
 \sH^{ac}(A)=\text{clsp}\,\bigl\{P^{ac}(A) E_A(\delta) C h: h\in\cH,\,\delta\in\mathfrak B(\dR)\bigr\}.
\end{equation*}
Since $E^{ac}_A(\delta)=P^{ac}(A) E_A(\delta)$ this implies \eqref{okja}.\\[-2ex]

\noindent {\it Step 2.}
Now we apply Theorem~\ref{att} to obtain a
preliminary form of the scattering matrix $\{S(A,B;\lambda)\}_{\lambda\in\dR}$.
Since $M(\cdot)$ is $\sS_1$-regular by assumption we have $\Im M(i)=\gamma(i)^*\gamma(i)\in\sS_1(\cH)$ (see \eqref{imm}) and hence
$\gamma(i)\in\mathfrak S_2(\cH,\sH)$ and
\begin{equation*}
 C = {\gamma(-i)}=\bigl(I-2i (A+i)^{-1}\bigr)\gamma(i)\in\mathfrak S_2(\cH,\sH).
\end{equation*}
Therefore the function $\lambda\mapsto C^*E_A((-\infty,\lambda))C$  is  $\mathfrak S_1(\cH)$-valued and in accordance with \cite[Lemma 2.2]{BirEnt67}
this function is $\mathfrak S_1(\cH)$-differentiable  for a.e. $\lambda\in\dR$.
We compute  its derivative
\begin{equation*}
 \lambda\mapsto K(\lambda)=\frac{d}{d\lambda} C^* E_A((-\infty,\lambda))C
\end{equation*}
and the  square root $\lambda\mapsto \sqrt{K(\lambda)}$ for a.e. $\lambda\in\dR$. First
we note that by the $\mathfrak S_1(\cH)$-generalization of the Fatou theorem (see
\cite[Lemma 2.4]{BirEnt67})
\begin{equation}\label{WScat3.4A}
\begin{split}
 K(\lambda)& = \lim_{\varepsilon\rightarrow 0+}\frac{1}{2\pi i} C^* \bigl((A-\lambda-i\varepsilon)^{-1}-
 (A-\lambda+i\varepsilon)^{-1}\bigr)C  \\
 &=\lim_{\varepsilon\rightarrow 0+}\frac{\varepsilon}{\pi } C^* \bigl((A-\lambda-i\varepsilon)^{-1}(A-\lambda+i\varepsilon)^{-1}\bigr)C
 \end{split}
\end{equation}
for a.e. $\lambda\in\dR$. On the other hand, inserting formula
$$
\gamma(\lambda+i\varepsilon) = (A+i)(A-\lambda-i\varepsilon)^{-1}\gamma(-i) =
(A+i)(A-\lambda-i\varepsilon)^{-1}C
$$
(see  \eqref{gform1})  into  \eqref{imm} gives
 \begin{equation*}
\begin{split}
 \Im M(\lambda+i\varepsilon)& = \varepsilon\gamma(\lambda + i\varepsilon)^*\gamma(\lambda +
 i\varepsilon) \\
    & =\varepsilon
C^*(I+A^2)\bigl(A-\lambda+i\varepsilon\bigr)^{-1}
\bigl(A-\lambda-i\varepsilon\bigr)^{-1}C.
 \end{split}
\end{equation*}
Combining this relation with \eqref{WScat3.4A}  implies
\begin{equation*}
  \Im M(\lambda+i 0)= \lim_{\varepsilon\rightarrow 0+} {\Im M(\lambda+i\varepsilon)}=\pi(1+\lambda^2) K(\lambda)
\end{equation*}
for a.e. $\lambda\in\dR$. In particular, $\ran({\Im M(\lambda+i0)})= \ran (K(\lambda))$
for a.e. $\lambda\in\dR$ and hence
\begin{equation*}
 \cH_\lambda=\overline{\ran\bigl({\Im M(\lambda+i0)}\bigr)}=\overline{\ran (K(\lambda))}\quad \text{for a.e.}\,\,\,\lambda\in\dR.
\end{equation*}
Therefore $L^2(\dR,d\lambda,\cH_\lambda)$ is a spectral representation of $A^{ac}$ and in accordance with
Theorem~\ref{att}  the scattering matrix $\{S(A,B;\lambda)\}_{\lambda\in\dR}$ is given by
\begin{equation}\label{sss}
\begin{split}
S(A,B;\lambda)&=I_{\cH_\lambda}+2\pi i(1+\lambda^2)^2\sqrt{K(\lambda)}Z(\lambda)  \sqrt{K(\lambda)}\\
               &=I_{\cH_\lambda}+2 i(1+\lambda^2)\sqrt{{\Im M(\lambda+i0)}}Z(\lambda)  \sqrt{{\Im M(\lambda+i0)}}
 \end{split}
\end{equation}
for a.e. $\lambda\in\dR$, where  $Z(\cdot)$ is  given by \eqref{WScatA8},
   \begin{equation}\label{z}
 Z(\lambda)=\frac{1}{\lambda+i}Q^*Q+\frac{1}{(\lambda+i)^2}\phi(\lambda)G+\lim_{\varepsilon\rightarrow 0+}
 Q^*\bigl(B-(\lambda+i\varepsilon)\bigr)^{-1}Q,
\end{equation}
and
\begin{equation*}
Q=\phi(A)CG = -(A+i)(A-i)^{-1}{\gamma(-i)}\,{M(i)^{-1}}= -
{\gamma(i)}\,{M(i)^{-1}}\in\sS_2(\cH,\sH).
\end{equation*}
Observe that due to the last inclusion the limit in \eqref{z} exists for   a.e. $\lambda\in\dR$ in every
$\mathfrak S_p$-norm with $p>1$  and  the operator-valued function
 $Z(\cdot)$ in \eqref{z} is well defined a.e. on $\dR$; cf. Lemma~\ref{lemmilein}.\\[-2ex]

\noindent{\it Step 3.}
In the third and final step we prove that
\begin{equation}\label{zz}
 Z(\lambda) = -\frac{1}{1+\lambda^2} {M(\lambda+i0)^{-1}}
\end{equation}
for a.e. $\lambda\in\dR$.
Then inserting this expression in \eqref{sss} one arrives at the asserted
form of the scattering matrix.

Applying the mapping $\Gamma_0$ to   \eqref{krein} and using $\ker (\Gamma_0) = \dom (A)$ and Definition~\ref{def_weyl_function}
one gets
\begin{equation}\label{WScat_3.9}
 \Gamma_0(B-z)^{-1}= \Gamma_0(A-z)^{-1}-\Gamma_0 \gamma(z) M(z)^{-1}\gamma(\bar z)^*
                         =-M(z)^{-1}\gamma(\bar z)^*
 \end{equation}
for $z\in\rho(A)\cap\rho(B)$ and hence
\begin{equation*}
  \Gamma_0(B+i)^{-1} = - {M(-i)^{-1}}\gamma(i)^* =\bigl(- {\gamma(i)}\, {M(i)^{-1}}\bigr)^* = Q^*.
\end{equation*}
This yields
\begin{equation}\label{opop}
\begin{split}
Q^*(B- z)^{-1}Q&=\Gamma_0(B+i)^{-1} (B- z)^{-1}Q \\
&=\Gamma_0\bigl(Q^*(B-\bar z)^{-1}(B-i)^{-1}\bigr)^*\\
&=\Gamma_0\bigl(\Gamma_0(B+i)^{-1}(B-\bar z)^{-1}(B-i)^{-1}\bigr)^*.
 \end{split}
\end{equation}
In order to compute this expression we note that
\begin{equation*}
 \begin{split}
  (B&+i)^{-1}(B-\bar z)^{-1}(B-i)^{-1}\\
    &=\frac{-1}{1+\bar z^2}\bigl((B+i)^{-1}-(B-\bar z)^{-1}\bigr) +\frac{1}{2i(\bar z-i)}\bigl((B+i)^{-1}-(B-i)^{-1}\bigr)
 \end{split}
\end{equation*}
and hence \eqref{WScat_3.9}  implies
\begin{equation*}
 \begin{split}
\Gamma_0(B+i)^{-1}(B-&\bar z)^{-1}(B-i)^{-1}=\frac{1}{1+\bar z^2}\bigl(M(-i)^{-1}\gamma(i)^*-M(\bar z)^{-1}\gamma(z)^*\bigr)\\
  &-\frac{1}{2i(\bar z-i)}\bigl(M(-i)^{-1}\gamma(i)^*-M(i)^{-1}\gamma(-i)^*\bigr).
 \end{split}
\end{equation*}
Taking into account that $(M(\bar\mu)^{-1})^* = {M(\mu)^{-1}}$ for $\mu\in\rho(A)\cap\rho(B)$ we obtain for the adjoint
\begin{equation*}
 \begin{split}
  \bigl(\Gamma_0(B+i)^{-1}(B-\bar z)^{-1}(B-i)^{-1}\bigr)^*&=
 \frac{1}{1+z^2}\bigl({\gamma(i)}\,{M(i)^{-1}} - {\gamma(z)}\,{M(z)^{-1}}\bigr)\\
 &+ \frac{1}{2i(z+i)}\bigl({\gamma(i)}\, {M(i)^{-1}}- {\gamma(-i)}\,{M(-i)^{-1}}\bigr).
 \end{split}
\end{equation*}
In turn, combining this identity with  \eqref{opop}  yields
\begin{equation*}
 \begin{split}
Q^*(B - z)^{-1}Q h &= \Gamma_0\bigl(\Gamma_0(B+i)^{-1}(B-\bar z)^{-1}(B-i)^{-1}\bigr)^*   \\
&=
 \frac{1}{1 + z^2}\bigl(M(i)^{-1}- M(z)^{-1}\bigr) + \frac{1}{2i(z + i)}\bigl(M(i)^{-1}-M(-i)^{-1}\bigr)
 \end{split}
\end{equation*}
for
$z \in\rho(A)\cap\rho(B)$. Setting  here $z=\lambda + i \varepsilon \in \mathbb C_+$  and passing to the limit as $\varepsilon\rightarrow 0$ one derives
\begin{equation}\label{WScat3.11}
\begin{split}
 \lim_{\varepsilon\rightarrow 0+} Q^*\bigl(B-(\lambda+i\varepsilon)\bigr)^{-1}Q&=
 \frac{1}{1+\lambda^2}\bigl({M(i)^{-1}}- {M(\lambda+i0)^{-1}}\bigr)\\
 &\qquad +\frac{1}{2i(\lambda+i)} \bigl({M(i)^{-1}}- {M(-i)^{-1}}\bigr)
\end{split}
 \end{equation}
for a.e. $\lambda\in\dR$; note that by Lemma~\ref{lemmilein} the limit ${M(\lambda+i0)^{-1}}\in\cB(\cH)$ exists for a.e. $\lambda\in\dR$.

Moreover, we have
 \begin{equation*}
 \begin{split}
  Q^*Q  =& \bigl({\gamma(i)}\, {M(i)^{-1}}\bigr)^*  {\gamma(i)}\, {M(i)^{-1}} = M(-i)^{-1}\gamma(i)^* \gamma(i)M(i)^{-1}  \\
             %%%%=&M(-i)^{-1}\gamma(i)^* \gamma(i)M(i)^{-1}h
             = & \frac{1}{2i}M(-i)^{-1}\bigl(M(i)-M(-i)\bigr)M(i)^{-1}           =  \frac{1}{2i} \bigl(M(-i)^{-1}-M(i)^{-1}\bigr).
 \end{split}
 \end{equation*}
Inserting this relation and  \eqref{WScat3.11} into    \eqref{z}  and taking notations \eqref{WScat3.2A} into account
we obtain for a.e. $\lambda\in\dR$
\begin{equation*}
 \begin{split}
  Z(\lambda)  & =\frac{1}{\lambda+i}Q^*Q + \frac{1}{(\lambda + i)^2}\phi(\lambda)G + Q^*\bigl(B-(\lambda+i0)\bigr)^{-1}Q  \\
  &=\frac{1}{2i(\lambda+i)} \bigl(M(-i)^{-1}-M(i)^{-1}\bigr)  - \frac{1}{1+\lambda^2} M(i)^{-1}   \\
  &\quad +\frac{1}{1+\lambda^2}\bigl(M(i)^{-1} - {M(\lambda+i0)^{-1}}\bigr)
          +\frac{1}{2i(\lambda+i)} \bigl(M(i)^{-1}-M(-i)^{-1}\bigr)   \\
  &=-\frac{1}{1+\lambda^2} {M(\lambda + i0)^{-1}},
 \end{split}
\end{equation*}
that is, \eqref{zz} holds.
\end{proof}

\begin{remark}
{\rm
Instead of the assumption that the Weyl function is $\sS_1$-regular one may assume in Theorem~\ref{th:3.1} that
$R_{B,A}(z)=(B- z)^{-1}-(A - z)^{-1}\in\sS_1(\sH)$ holds for some $z\in\rho(A)\cap\rho(B)$ and $R_{B,A}(\lambda_0)\geq 0$ for some $\lambda_0\in\dR\cap\rho(A)\cap\rho(B)$;
cf. Proposition~\ref{prop:2.7}.
}
\end{remark}

Our next task is to drop the assumption  of the simplicity of  $S$ in
Theorem~\ref{th:3.1}.   If $S=A\cap B$ is not simple then the Hilbert space $\sH$ admits an
orthogonal decomposition $\sH=\sH_0\oplus\sH^\prime$ with $\sH_0\not=\{0\}$ such that
\begin{equation}\label{sssi}
 S=S_0\oplus S^\prime,
\end{equation}
where $S_0$ is a self-adjoint operator in the Hilbert space $\sH_0$ and $S^\prime$ is a simple symmetric operator in the Hilbert space $\sH^\prime$; cf. \cite{K49}.
It follows that there exist self-adjoint
extensions $A'$ and $B'$ of $S'$ in  $\sH'$ such that
\begin{equation*}
 A= S_0\oplus A'\qquad\text{and}\qquad B= S_0\oplus B'.
\end{equation*}
By restricting the boundary maps of a $B$-generalized boundary triple for $S^*$ one obtains a
$B$-generalized  boundary triple for the operator $(S')^*$ with the same Weyl function. Applying Theorem~\ref{th:3.1} to the pair $\{A',B'\}$ yields the
following variant of Theorem~\ref{th:3.1}; cf. \cite[Proof of Theorem 3.2]{BMN09} for the same argument in the special case of finite rank perturbations.
\begin{corollary}\label{cor:3.2}

Let $A$ and $B$ be self-adjoint  operators in a Hilbert space $\sH$,
assume that the closed symmetric operator $S = A\cap B$ is densely defined and decomposed in $S=S_0\oplus S^\prime$ as in \eqref{sssi}, and
let $L^2(\dR,d\lambda,\cG_\lambda)$ be a spectral representation of $S_0^{ac}$.
Let $\Pi = \{\cH,\Gamma_0,\Gamma_1\}$ be a  $B$-generalized boundary triple for $S^*$ as in Theorem~\ref{th:3.1} such that the corresponding
Weyl function $M(\cdot)$ is $\sS_1$-regular.

Then
$\{A,B\}$ is a complete scattering system and
\begin{equation*}
 L^2(\dR,d\lambda,\cH_\lambda\oplus \cG_\lambda),\qquad \cH_\lambda:=\overline{\ran (\Im M(\lambda+i0))},
\end{equation*}
forms a spectral representation of $A^{ac}$ such that for a.e. $\lambda \in \dR$
the scattering matrix $\{S(A,B;\lambda)\}_{\lambda\in\dR}$ of the scattering system $\{A,B\}$ admits
the representation
\begin{equation*}
 S(A,B; \lambda) = \begin{pmatrix} S(A',B'; \lambda) & 0 \\ 0 & I_{\cG_\lambda}\end{pmatrix},
\end{equation*}
where
\begin{equation*}
 S(A',B'; \lambda)=I_{\cH_\lambda}-2i\sqrt{{\Im M(\lambda+i0)}}\;M(\lambda+i0)^{-1} \sqrt{{\Im
 M(\lambda+i0)}}.
\end{equation*}
\end{corollary}

\section{Scattering matrices for Schr\"{o}dinger operators on exterior domains}\label{nrsec}

Our main objective in this section is to derive representations of the scattering matrices for pairs of self-adjoint Schr\"{o}dinger operators with
Dirichlet, Neumann and Robin boundary conditions on unbounded domains with smooth compact boundaries in terms of
Dirichlet-to-Neumann and Neumann-to-Dirichlet maps. After some necessary preliminaries in Sections~\ref{prelisec} and
\ref{prelisec2} we formulate and prove our main results Theorem~\ref{th:4.1} and Theorem~\ref{nrthm} in Sections~\ref{D-N_sect} and \ref{nrsec_N-Roben}, respectively.
Both theorems follow in a similar way from our general result Theorem~\ref{th:3.1} by fixing a suitable $B$-generalized boundary triple and
verifying that the corresponding Weyl function is $\sS_1$-regular. We also mention that along the way we obtain classical results on singular value estimates
of resolvent differences due to Birman, Grubb and others without any extra efforts; cf. Remarks~\ref{remarki} and~\ref{remarkii}.

\subsection{Preliminaries on Sobolev spaces, trace maps, and Green's second identity}\label{prelisec}
Let $\Omega\subset\dR^n$  be an exterior domain, that is, $\dR^n\setminus\Omega$ is
bounded and assume that the boundary $\partial\Omega$ of $\Omega$ is $C^\infty$-smooth.
We denote by $H^s(\Omega)$, $s\in\dR$, the usual $L^2$-based Sobolev spaces on the unbounded exterior domain $\Omega$, and by $H^r(\partial\Omega)$, $r\in\dR$,
the corresponding Sobolev spaces on the compact $C^\infty$-boundary $\partial\Omega$. The corresponding scalar products will be denoted by $(\cdot,\cdot)$, and sometimes
the space is used as an index.

Recall that the Dirichlet and Neumann trace operators $\gamma_D$ and $\gamma_N$, originally defined as linear mappings from
$C_0^\infty(\overline\Omega)$ to $C^\infty(\partial\Omega)$, admit continuous extensions onto $H^2(\Omega)$ such that the mapping
\begin{equation}\label{gammdgammn}
 \begin{pmatrix} \gamma_D \\ \gamma_N\end{pmatrix}:H^2(\Omega)\rightarrow H^{3/2}(\partial\Omega)\times H^{1/2}(\partial\Omega)
\end{equation}
is surjective. The spaces
\begin{equation}\label{hs1}
 H^s_\Delta(\Omega)=\bigl\{f\in H^s(\Omega):\Delta f\in L^2(\Omega) \bigr\},\qquad s\in [0,2],
\end{equation}
equipped with the Hilbert scalar products
\begin{equation}\label{hs2}
 (f,g)_{H^s_\Delta(\Omega)}=(f,g)_{H^s(\Omega)}+(\Delta f,\Delta g)_{L^2(\Omega)},\qquad f,g\in H^s_\Delta(\Omega),
\end{equation}
will play an important role. In particular, we will use that the Dirichlet
trace operator can be extended by continuity to surjective mappings
\begin{equation}\label{gammadcont}
 \gamma_D:H^{3/2}_\Delta(\Omega)\rightarrow H^1(\partial\Omega)\quad\text{and}\quad
 \gamma_D:H^1_\Delta(\Omega)\rightarrow H^{1/2}(\partial\Omega),
\end{equation}
and the Neumann trace operator can be extended by continuity to surjective
mappings
\begin{equation}\label{gammancont}
 \gamma_N:H^{3/2}_\Delta(\Omega)\rightarrow L^2(\partial\Omega)\quad\text{and}\quad
 \gamma_N:H^1_\Delta(\Omega)\rightarrow H^{-1/2}(\partial\Omega);
\end{equation}
cf. \cite[Theorems 7.3 and 7.4, Chapter 2]{LioMag71} for the case of a bounded smooth
domain and, e.g. \cite[Lemma 3.1 and Lemma 3.2]{GM11}. At the same time the second
Green's identity
\begin{equation}\label{green32}
 (-\Delta f,g)_{L^2(\Omega)}-(f,-\Delta g)_{L^2(\Omega)}=(\gamma_D f,\gamma_N g)_{L^2(\partial\Omega)}-(\gamma_N f,\gamma_D g)_{L^2(\partial\Omega)},
\end{equation}
well known for $f,g\in H^2(\Omega)$, remains valid for $f,g\in H^{3/2}_\Delta(\Omega)$ and extends further to functions $f,g\in H^1_\Delta(\Omega)$
\begin{equation}\label{green12}
 (-\Delta f,g)_{L^2(\Omega)}-(f,-\Delta g)_{L^2(\Omega)}=\langle\gamma_D f,\gamma_N g\rangle-\langle\gamma_N f,\gamma_D g\rangle,
\end{equation}
where $\langle\cdot,\cdot\rangle$ denotes the extension of the $L^2(\partial\Omega)$-inner product onto the dual pair
$H^{1/2}(\partial\Omega)\times H^{-1/2}(\partial\Omega)$ and $H^{-1/2}(\partial\Omega)\times H^{1/2}(\partial\Omega)$, respectively. As usual, here
\begin{equation}\label{rig}
H^{1/2}(\partial\Omega)\hookrightarrow L^2(\partial\Omega)\hookrightarrow H^{-1/2}(\partial\Omega)
\end{equation}
is viewed as a rigging of Hilbert spaces, that is, some uniformly positive self-adjoint operator $\jmath$ in $L^2(\partial\Omega)$ with $\dom(\jmath)=H^{1/2}(\partial\Omega)$
is fixed and viewed as an isomorphism
\begin{equation}\label{jjj}
\jmath:H^{1/2}(\partial\Omega)\longrightarrow L^2(\partial\Omega).
\end{equation}
As scalar product on $H^{1/2}(\partial\Omega)$ we choose
$(\varphi,\psi)_{H^{1/2}(\partial\Omega)}:=(\jmath\varphi,\jmath\psi)_{L^2(\partial\Omega)}$; it follows that
$H^{-1/2}(\partial\Omega)$ coincides with the completion of $L^2(\partial\Omega)$ with respect to $(\jmath^{-1} \cdot,\jmath^{-1}\cdot)_{L^2(\partial\Omega)}$,
and $\jmath^{-1}$ admits an extension to an isomorphism
$$\widetilde{\jmath^{-1}}:H^{-1/2}(\partial\Omega)\longrightarrow L^2(\partial\Omega).$$
The inner product $\langle\cdot,\cdot\rangle$ on the right hand side of \eqref{green12} is
\begin{equation}\label{green123}
\langle\varphi,\psi\rangle:=\bigl(\jmath \varphi,\widetilde{\jmath^{-1}}\psi\bigr)_{L^2(\partial\Omega)},\qquad \varphi\in H^{1/2}(\partial\Omega),\quad \psi\in
H^{-1/2}(\partial\Omega),
\end{equation}
and extends the $L^2(\partial\Omega)$ scalar product in the sense that $\langle\varphi,\psi\rangle=(\varphi,\psi)_{L^2(\partial\Omega)}$ for $\varphi\in H^{1/2}(\partial\Omega)$ and
$\psi\in L^2(\partial\Omega)$.
A standard and convenient choice for $\jmath$ in \eqref{jjj} in many situations is
\begin{equation}\label{jd}
\jmath_\Delta := (-\Delta_{\partial\Omega} + I)^{1/4}:
H^{1/2}(\partial\Omega)\longrightarrow L^2(\partial\Omega),
\end{equation}
where $-\Delta_{\partial\Omega}$  denotes the Laplace-Beltrami operator  in
$L^2(\partial\Omega)$; cf. Remark~\ref{rem:4.2} for other natural choices of $\jmath$.
Note in this connection that $\jmath_\Delta$ maps
$H^{s}(\partial\Omega)$ isomorphically onto $H^{s-1/2}(\partial\Omega)$ for any $s\in \mathbb R$.

In this context we also recall the following lemma, which is essentially a consequence of the
asymptotics of the eigenvalues of the Laplace–Beltrami operator on compact manifolds; cf.
\cite[Proof of Proposition 5.4.1]{Ag90}, \cite[Theorem 2.1.2]{Agr97}, and \cite[Lemma 4.7]{BLL13-2}.

\begin{lemma}\label{hurra}
Let $\cK$ be a Hilbert space and assume that $X\in\cB(\cK,H^s(\partial\Omega))$
has the property $\ran X \subset H^r(\partial\Omega)$ for some $r>s\geq 0$. Then
$$
X\in\cS_{\frac{n-1}{r-s}}\bigl(\cK,H^s(\partial\Omega)\bigr)
$$
and hence $X\in\sS_p(\cK,H^s(\partial\Omega))$
for $p>\frac{n-1}{r-s}$.
\end{lemma}

As a useful consequence of Lemma~\ref{hurra} we note that for $r> 0$ the canonical embeddings $\iota_r:H^r(\partial\Omega)\longrightarrow L^2(\partial\Omega)$
and $\iota_{-r}:L^2(\partial\Omega)\longrightarrow H^{-r}(\partial\Omega)$ satisfy
 $$ \iota_r\in\cS_{\frac{n-1}{r}}\bigl(H^r(\partial\Omega),L^2(\partial\Omega)\bigr)\quad\text{and}\quad
    \iota_{-r}\in\cS_{\frac{n-1}{r}}\bigl(L^2(\partial\Omega),H^{-r}(\partial\Omega)\bigr),
 $$
respectively. In fact, the assertion for $\iota_r$ follows after fixing a unitary operator $U:L^2(\partial\Omega)\longrightarrow H^r(\partial\Omega)$,
applying Lemma~\ref{hurra} to the operator $X=\iota_r U$ and noting that the singular values  of $X$ and $\iota_r$ are the same. Since the dual operator
$\iota_r^\prime:L^2(\partial\Omega)\longrightarrow H^{-r}(\partial\Omega)$ coincides with the canonical embedding $\iota_{-r}$ of $L^2(\partial\Omega)$
into $H^{-r}(\partial\Omega)$ the second assertion follows. By composition and \eqref{idealipp} we also conclude
\begin{equation}\label{hohopasst}
 \iota_{-r}\circ\iota_r \in\cS_{\frac{n-1}{2r}}\bigl(H^r(\partial\Omega),H^{-r}(\partial\Omega)\bigr).
\end{equation}

\subsection{Schr\"{o}dinger operators with Dirichlet, Neumann, and Robin boundary conditions}\label{prelisec2}

Let $\Omega\subset\dR^n$ be an exterior domain as in Section~\ref{prelisec}. In the following we consider a
Schr\"{o}dinger differential  expression with a bounded, measurable, real valued potential $V$,
\begin{equation}\label{Schrod_expres}
 \mathscr{L} = -\Delta  + V,\qquad V \in L^\infty(\Omega).
 \end{equation}
With the differential expression in \eqref{Schrod_expres} one
naturally associates the minimal operator
\begin{equation}\label{smin}
\begin{split}
S_{min} f&=\mathscr{L}f,\\
\dom (S_{min})&= H^2_0(\Omega)= \bigl\{f\in
H^2(\Omega): \gamma_D f=\gamma_N f=0\bigr\},
\end{split}
\end{equation}
and the maximal operator
\begin{equation*}
\begin{split}
S_{max} f&=\mathscr{L}f,\\ 
\dom (S_{max})&=\bigl\{f\in L^2(\Omega): -\Delta f+Vf \in L^2(\Omega)\bigr\},
\end{split}
\end{equation*}
in $L^2(\Omega)$; the expression $\Delta f$ in $\dom (S_{max})$ is understood in the sense of distributions.
We note that $\dom (S_{max})$ equipped with the graph norm coincides with the Hilbert space $H^0_\Delta(\Omega)$ introduced above. In the next lemma we collect some
well-known properties of $S_{min}$ and $S_{max}$; for the simplicity of $S$ we refer to \cite[Proposition 2.2]{BR16}
and the density of $H^s_\Delta(\Omega)$ in $\dom (S^*)$ equipped with the graph norm is shown (for the case of a bounded domain) in
\cite[Chapter~2,Theorem~6.4]{LioMag71}.

\begin{lemma}\label{slem}
The operator $S:=S_{min}$ is a densely defined, closed, simple, symmetric operator in $L^2(\Omega)$. The deficiency indices of $S$ coincide and are both infinite,
\begin{equation*}
 \dim\bigl(\ran(S-i)^\bot\bigr)=\dim\bigl(\ran(S+i)^\bot\bigr)=\infty.
\end{equation*}
The adjoint of the minimal operator is the maximal operator,
\begin{equation*}
 S^*=S_{min}^*=S_{max}\quad\text{and}\quad S=S_{min}=S_{max}^*,
\end{equation*}
and the spaces $H^s_\Delta(\Omega)$, $s\in[0,2]$, are dense in $\dom (S^*)$ equipped with the graph norm.
\end{lemma}

In Sections~~\ref{D-N_sect} and \ref{nrsec_N-Roben} we are interested in scattering systems consisting of different self-adjoint realizations of $\mathscr{L}$ in $L^2(\Omega)$.
The self-adjoint Dirichlet and Neumann operators associated to the densely defined, semibounded, closed quadratic forms
\begin{equation*}
\begin{split}
 \mathfrak a_D[f,g]&=(\nabla f,\nabla g)_{(L^2(\Omega))^n}+(Vf,g)_{L^2(\Omega)},\qquad \dom (\mathfrak a_D)=H^1_0(\Omega),\\
 \mathfrak a_N[f,g]&=(\nabla f,\nabla g)_{(L^2(\Omega))^n}+(Vf,g)_{L^2(\Omega)},\qquad \dom (\mathfrak a_N)=H^1(\Omega),
 \end{split}
 \end{equation*}
are  given by
\begin{equation}\label{adan}
 \begin{split}
  A_D f &= \mathscr{L} f,\qquad \dom (A_D) = \bigl\{f\in H^2(\Omega):\gamma_D f = 0\bigr\},\\
  A_N f &= \mathscr{L} f,\qquad \dom (A_N) = \bigl\{f\in H^2(\Omega):\gamma_N f = 0\bigr\},
 \end{split}
\end{equation}
and for a real valued function $\alpha\in L^\infty(\partial\Omega)$ the quadratic form
\begin{equation*}
\mathfrak a_\alpha[f,g]=\mathfrak a_N[f,g]-(\alpha\gamma_D f,\gamma_D g)_{L^2(\partial\Omega)},\qquad \dom (\mathfrak a_\alpha)=H^1(\Omega),
\end{equation*}
is also densely defined, closed and semibounded from below, and hence gives rise to a semibounded self-adjoint operator in $L^2(\Omega)$, which has the form
\begin{equation}\label{robin}
 A_\alpha f=\mathscr{L}f,\qquad \dom (A_\alpha)=\bigl\{f\in H^{3/2}_\Delta(\Omega): \alpha\gamma_D f= \gamma_N f \bigr\}.
  \end{equation}
We remark that the $H^2$-regularity of the functions in $\dom (A_D)$ and $\dom (A_N)$ is
a classical fact (see the monographs \cite{A65,LU68,LioMag71}) and the
$H^{3/2}$-regularity of the functions in $\dom (A_\alpha)$ can be found in, e.g.
\cite[Corollary 6.25]{BL12}; in the case that
the coefficient $\alpha$ in the Robin boundary condition is continuously differentiable
also $\dom (A_\alpha)$ is contained in $H^2(\Omega)$; cf. \cite[Theorem 4.18]{McL00}.

\subsection{Scattering matrix for the Dirichlet and Robin realization}\label{D-N_sect}

In this subsection we consider the pair $\{A_D,A_\alpha\}$ consisting of the self-adjoint Dirichlet and Robin operator associated to $\mathscr{L}$ in \eqref{adan}
and \eqref{robin} on an exterior domain $\Omega\subset\dR^2$; here we restrict ourselves to the two dimensional situation in order to ensure that the trace class
condition \eqref{trace} for the resolvent  difference is satisfied; cf. Remark~\ref{remarki}.

Before formulating and proving our main result on the system $\{A_D,A_\alpha\}$ we recall the definition and some useful properties of the
Dirichlet-to-Neumann map. First we note that for any $\psi\in H^{1/2}(\partial\Omega)$ and $z \in\rho(A_D)$ there exists a unique
solution $f_z \in H_{\Delta}^1(\Omega)$ of the boundary  value problem
\begin{equation}\label{bvp_D-to-N}
 -\Delta f_z + Vf_z = zf_z,\qquad \gamma_D f_z = \psi\in H^{1/2}(\partial\Omega).
\end{equation}
The corresponding solution operator is given  by
\begin{equation}\label{sol_oper_for_Dir}
P_D(z): H^{1/2}(\partial\Omega) \longrightarrow H^{1}_\Delta(\Omega)\subset
L^2(\Omega),\qquad  \psi\mapsto f_z.
\end{equation}
For $z\in\rho(A_D)$ the {\em Dirichlet-to-Neumann map} $\Lambda_{1/2}(z)$ is defined by
  \begin{equation}\label{Dir-to-Neuman_map}
\Lambda_{1/2}(z): H^{1/2}(\partial\Omega)\longrightarrow H^{-1/2}(\partial\Omega),\qquad
\psi\mapsto\gamma_N P_D(z)\psi,
  \end{equation}
and takes Dirichlet boundary values $\gamma_D f_{z}$ of the solution $f_z\in
H^1_{\Delta}(\Omega)$ of \eqref{bvp_D-to-N}  to their Neumann boundary values $\gamma_N f_z\in
H^{-1/2}(\partial\Omega)$.

Now we are ready to formulate and prove a representation of the scattering matrix for the pair $\{A_D,A_\alpha\}$.
\begin{theorem}\label{th:4.1}
Let $\Omega\subset \dR^2$ be an exterior domain  with a $C^\infty$-smooth boundary, let $V\in
L^\infty(\Omega)$ and $\alpha  \in L^\infty(\partial\Omega)$ be real valued functions, and let  $A_D$
and  $A_\alpha$  be the self-adjoint Dirichlet and Robin realizations of $\mathscr{L} = -\Delta+V$ in
$L^2(\Omega)$ in \eqref{adan} and \eqref{robin}, respectively.
Moreover, let $\Lambda_{1/2}(\cdot)$ be the Dirichlet-to-Neumann map defined in \eqref{Dir-to-Neuman_map}
and let
  \begin{equation}\label{eq:4.13}
    M^D_\alpha(z) := \widetilde{\jmath^{-1}}(\alpha - \Lambda_{1/2}(z))\jmath^{-1}, \quad z\in \rho(A_D),
   \end{equation}
where $\jmath: H^{1/2}(\partial\Omega) \longrightarrow L^2(\partial\Omega)$ denotes some uniformly positive self-adjoint operator
in $L^2(\partial\Omega)$ with $\dom(\jmath)=H^{1/2}(\partial\Omega)$ as in \eqref{rig}--\eqref{jjj}.

Then
$\{A_D,A_\alpha\}$ is a complete scattering system and
\begin{equation*}
 L^2(\dR,d\lambda,\cH_\lambda),\qquad \cH_\lambda:=\overline{\ran (\Im M^D_\alpha(\lambda+i0))},
\end{equation*}
forms a spectral representation of $A_D^{ac}$ such that for a.e. $\lambda \in \dR$
the scattering matrix $\{S(A_D,A_\alpha;\lambda)\}_{\lambda\in\dR}$ of the scattering system $\{A_D,A_\alpha\}$ admits
the representation
  \begin{equation*}
 S(A_D,A_\alpha;\lambda)=I_{\cH_\lambda}-2i\sqrt{\Im M^D_\alpha(\lambda+i0)}\,M^D_\alpha(\lambda+i0)^{-1} \sqrt{\Im
 M^D_\alpha(\lambda + i0)}.
   \end{equation*}
\end{theorem}
\begin{proof}
It follows from \eqref{adan} and \eqref{robin} that the operator $A_\alpha \cap A_D$ coincides with the minimal
operator $S = \mathscr{L}_{\min}$ associated with $\mathscr{L}$ in  \eqref{smin}, which is closed, densely defined and simple by Lemma~\ref{slem}.
Define the operator $T$ as a restriction of $S^*$ to the domain $H^{1}_{\Delta}(\Omega)$,
\begin{equation*}
Tf = -\Delta f+Vf, \qquad \dom (T) = H^{1}_{\Delta}(\Omega),
\end{equation*}
and let
\begin{equation}\label{eq:4.18}
\Gamma_0 f := \jmath\,\gamma_D f\quad\text{and} \quad \Gamma_1 f := \widetilde{\jmath^{-1}}(\alpha\gamma_D-\gamma_N )f,\quad f\in\dom (T).
\end{equation}
We claim that $\Pi^D_\alpha = \{L^2(\partial\Omega),\Gamma_0,\Gamma_1\}$ is a $B$-generalized boundary triple for $S^*$ with the $\sS_1$-regular
Weyl function $M^D_\alpha(\cdot)$ in \eqref{eq:4.13} such that
\begin{equation}\label{kernis}
 A_D=T\upharpoonright\ker(\Gamma_0)\quad\text{and}\quad A_\alpha=T\upharpoonright\ker(\Gamma_1).
\end{equation}

In fact, for $f,g\in\dom (T)$ we use \eqref{green12} and the fact that $\alpha$ is real valued, and compute
%q
\begin{equation*}
\begin{split}
&(\Gamma_1 f,\Gamma_0 g)-(\Gamma_0 f,\Gamma_1 g)\\
&\quad=\bigl(\widetilde{\jmath^{-1}}(\alpha\gamma_D-\gamma_N)f,\jmath\,\gamma_D g\bigr)-\bigl(\jmath\,\gamma_D f,\widetilde{\jmath^{-1}}(\alpha\gamma_D-\gamma_N)g\bigr)\\
&\quad=\bigl\langle\alpha\gamma_D f -\gamma_N f ,\gamma_D g\bigr\rangle-\bigl\langle \gamma_D f,\alpha\gamma_D g-\gamma_Ng  \rangle\\
&\quad=\langle\gamma_D f,\gamma_N g\rangle-\langle\gamma_N f,\gamma_D g\rangle \\
&\quad=(Tf,g)-(f,Tg)
\end{split}
\end{equation*}
and hence Green's identity \eqref{Green_formula_for_BGBtrip} is satisfied. Furthermore,
the mapping
\begin{equation*}
\gamma_D:\dom (T)\to H^{1/2}(\partial\Omega)
 \end{equation*}
is well defined and surjective according to \eqref{gammadcont}, and since
$\jmath:H^{1/2}(\partial\Omega)\rightarrow L^2(\partial\Omega)$ is an isomorphism we
conclude
\begin{equation*}
\ran(\Gamma_0) = L^2(\partial\Omega),
\end{equation*}
i.e., $\Gamma_0$ is surjective. From Lemma~\ref{slem} we directly obtain that $\dom (T) =
H^{1}_{\Delta}(\Omega)$ is dense in $\dom (S^*)$ equipped with the graph norm (which is equal to the space $H^0_\Delta(\Omega)$)
and hence
we have $\overline T=S^*$. Moreover, it follows from Green's identity \eqref{Green_formula_for_BGBtrip} that the restrictions $T\upharpoonright\ker(\Gamma_0)$
and $T\upharpoonright\ker(\Gamma_1)$ are both symmetric operators in $L^2(\Omega)$ and from the definition of the boundary maps it is clear that the self-adjoint operators
$A_D$ and $A_\alpha$ are contained in the symmetric operators $T\upharpoonright\ker(\Gamma_0)$
and $T\upharpoonright\ker(\Gamma_1)$, and hence they coincide.
Therefore, $\Pi^D_\alpha = \{L^2(\partial\Omega),\Gamma_0,\Gamma_1\}$ is a $B$-generalized boundary triple for $S^*$
such that \eqref{kernis} holds.

In order to see that the Weyl function is given by
\begin{equation}\label{eq:4.13a}
    M^D_\alpha(z) = \widetilde{\jmath^{-1}}(\alpha - \Lambda_{1/2}(z))\jmath^{-1}, \quad z\in \rho(A_D),
   \end{equation}
   we recall that $\Lambda_{1/2}(z)\gamma_D f_z=\gamma_N f_z$ for $f_z\in\ker(T-z)$, $z\in\rho(A_D)$,
according to the definition of the Dirichlet-to-Neumann map $\Lambda_{1/2}(\cdot)$ in \eqref{Dir-to-Neuman_map}. Hence we obtain
\begin{equation*}
\widetilde{\jmath^{-1}}\bigl(\alpha - \Lambda_{1/2}(z)\bigr)\jmath^{-1}\Gamma_0 f_z
   =\widetilde{\jmath^{-1}}\bigl(\alpha\gamma_D f_z - \Lambda_{1/2}(z)\gamma_D f_z\bigr)
   =\Gamma_1 f_z
\end{equation*}
for $f_z\in\ker(T-z)$ and $z\in\rho(A_D)$, and this yields \eqref{eq:4.13a} and \eqref{eq:4.13}.

It remains to verify that $M^D_\alpha(\cdot)$ is $\sS_1$-regular.
For this we denote the $\gamma$-field associated to $\Pi^D_\alpha$ by $\gamma_\alpha^D(\cdot)$ and use the relation
\begin{equation}\label{lplp}
 M^D_\alpha(z)=M^D_\alpha(\xi)^*+(z-\bar\xi)\gamma_\alpha^D(\xi)^* \gamma_\alpha^D(z)
\end{equation}
(see \eqref{gutgut}) with some $\xi\in \rho(A_D)\cap\rho(A_\alpha)\cap\rho(A_N)\cap\dR$ and all $z\in\rho(A_D)$. Observe that \eqref{gstar} and the
choice of $\Gamma_1$ in \eqref{eq:4.18} yield
\begin{equation}\label{plkj}
 \gamma_\alpha^D(\xi)^*h=\Gamma_1(A_D-\bar \xi)^{-1}h=-\widetilde{\jmath^{-1}}\gamma_N(A_D-\bar \xi)^{-1}h
\end{equation}
for all $h\in L^2(\Omega)$. Since $\dom (A_D)\subset H^2(\Omega)$ we conclude from \eqref{gammdgammn} that the range of the mapping  $\gamma_N(A_D-\bar \xi)^{-1}$ is contained in
$H^{1/2}(\partial\Omega)$. Furthermore we have $\gamma_\alpha^D(\xi)^*\in\cB(L^2(\Omega),L^2(\partial\Omega))$. Then it follows from \eqref{plkj} that
$$
\gamma_N(A_D-\bar \xi)^{-1}\in\cB\bigl(L^2(\Omega),H^{-1/2}(\partial\Omega)\bigr),
$$
and, in particular, this operator is closed. But then $\gamma_N(A_D-\bar \xi)^{-1}$ is also closed when viewed as an operator from $L^2(\Omega)$ into $H^{1/2}(\partial\Omega)$,
and since this operator is defined on the whole space $L^2(\Omega)$ we conclude
$$
\gamma_N(A_D-\bar \xi)^{-1}\in\cB\bigl(L^2(\Omega),H^{1/2}(\partial\Omega)\bigr).
$$
Now we use that the canonical embedding operator $\iota_{-1/2}\circ\iota_{1/2}:H^{1/2}(\partial\Omega)\longrightarrow H^{-1/2}(\partial\Omega)$
is compact and belongs to $\cS_1(H^{1/2}(\partial\Omega),H^{-1/2}(\partial\Omega))$ by \eqref{hohopasst}. Thus we have
\begin{equation*}
 \gamma_N(A_D-\bar \xi)^{-1}\in\cS_1\bigl(L^2(\Omega),H^{-1/2}(\partial\Omega)\bigr)
\end{equation*}
and hence \eqref{plkj} yields
\begin{equation*}
\gamma_\alpha^D(\xi)^*\in\cS_1\bigl(L^2(\Omega),L^2(\partial\Omega)\bigr).
\end{equation*}
It follows that also $\gamma_\alpha^D(\xi)\in\cS_1(L^2(\partial\Omega),L^2(\Omega))$ and hence for all $z\in\rho(A_D)$
\begin{equation}\label{gammaspsp456}
 \gamma_\alpha^D(z)=\bigl(I+(z-\xi)(A_D-z)^{-1}\bigr)\gamma_\alpha^D(\xi)\in\cS_1\bigl(L^2(\partial\Omega),L^2(\Omega)\bigr).
\end{equation}
Therefore
\begin{equation}\label{plö}
 (z-\bar\xi)\gamma_\alpha^D(\xi)^* \gamma_\alpha^D(z)\in\cS_{1/2}\bigl(L^2(\partial\Omega)\bigr),\qquad z\in\rho(A_D).
\end{equation}
Since $\cS_{1/2}(L^2(\partial\Omega))\subset \sS_1(L^2(\partial\Omega))$ and $M^D_\alpha(\xi) = M^D_\alpha(\xi)^*$
we conclude from \eqref{lplp} and \eqref{plö} that
$$K(z) := M^D_\alpha(z) - M^D_\alpha(\xi) \in \sS_1\bigl(L^2(\partial\Omega)\bigr),\quad z \in \dC_+.$$ 
Since $M^D_\alpha(\cdot)$ is a strict Nevanlinna function
$K(\cdot)$ is a strict Nevanlinna function.
It remains to show that
$$C := M^D_\alpha(\xi) = \widetilde{\jmath^{-1}}\alpha\jmath^{-1} - \widetilde{\jmath^{-1}}\Lambda_{1/2}(\xi)\jmath^{-1}$$ 
is boundedly invertible.
Using that the maps \eqref{gammadcont} and \eqref{gammancont} are surjective and $\xi\in\rho(A_D)\cap\rho(A_N)\cap\dR$ we find that the self-adjoint operator 
$\widetilde{\jmath^{-1}}\Lambda_{1/2}(\xi)\jmath^{-1}$
is surjective, and hence boundedly invertible in $L^2(\partial\Omega)$. From $\ran(\alpha\jmath^{-1}) \subseteq L^2(\partial\Omega)$ we obtain that $\widetilde{\jmath^{-1}}\alpha\jmath^{-1}$ is compact
and therefore $M^D_\alpha(\xi)$ is a Fredholm operator. Furthermore, $\ker(M^D_\alpha(\xi)) \not= \{0\}$ as otherwise there is a non-trivial function $f_\xi \in \ker(T - \xi)$ with
$\Gamma_1f_\xi = 0$, so that $f_\xi \in \ker(A_\alpha - \xi)$. But $\xi$ in \eqref{lplp} is also in $\rho(A_\alpha)$ and hence $f_\xi=0$; a contradiction. 
Thus $\ker(M^D_\alpha(\xi)) = \{0\}$ and hence $C = M^D_\alpha(\xi)$ is boundedly invertible.  Therefore $M^D_\alpha(\cdot)$ is an $\sS_1$-regular Weyl function.
Now the assertions in Theorem~\ref{th:4.1} follow from Theorem~\ref{th:3.1}.
\end{proof}

\begin{remark}\label{remarki}
{\rm
We note that \eqref{gammaspsp456} yields
$\gamma_\alpha^D(z)^*\in\sS_1(L^2(\Omega),L^2(\partial\Omega))$ for all $z\in\rho(A_D)$ and since
$M_\alpha^D(z)^{-1}\in\cB(L^2(\partial\Omega))$, $z\in\rho(A_D)\cap\rho(A_\alpha)$, we conclude from Krein's formula in Proposition~\ref{prop:2.4}~(iii) that
\begin{equation*}
 (A_\alpha-z)^{-1}-(A_D-z)^{-1}=-\gamma_\alpha^D(z) M_\alpha^D(z)^{-1}\gamma_\alpha^D(\bar z)^*\in\cS_{1/2}(L^2(\Omega)).
\end{equation*}
For $n=3,4,\dots$ one obtains in the same way as in the proof of Theorem~\ref{th:4.1} using \eqref{hohopasst} that
\begin{equation*}
 \gamma_\alpha^D(z)\in\cS_{n-1}\bigl(L^2(\partial\Omega),L^2(\Omega)\bigr)\quad\text{and}\quad\gamma_\alpha^D(z)^*\in\cS_{n-1}\bigl(L^2(\Omega),L^2(\partial\Omega)\bigr)
\end{equation*}
for all $z\in\rho(A_D)$ and hence
\begin{equation*}
 (A_\alpha-z)^{-1}-(A_D-z)^{-1}=-\gamma_\alpha^D(z) M_\alpha^D(z)^{-1}\gamma_\alpha^D(\bar z)^*\in\cS_{\frac{n-1}{2}}(L^2(\Omega)).
\end{equation*}
for all $z\in\rho(A_D)\cap\rho(A_\alpha)$ by Proposition~\ref{prop:2.4}~(iii).
This well known result goes back to Birman \cite{Bir62}
(see also  \cite{BLL13-2,GeMa08,Gru74,Gru84,M10} for more details on singular value estimates in this context).
}
\end{remark}

\begin{remark}\label{rem:4.2}
{\rm
There are several possibilities to choose the operator $\jmath$
in \eqref{jjj} used for the extension \eqref{green123} of the $L^2(\partial\Omega)$
scalar product in the rigging \eqref{rig}. Besides the choice
$\jmath_\Delta = (-\Delta_{\partial\Omega} + I)^{1/4}$ in \eqref{jd}
the following choice is very convenient for the scattering matrix, since it allows to express it completely in terms of the Dirichlet-to-Neumann map:
Fix some $\lambda_0<\min\{\sigma(A_D),\sigma(A_N)\}$ and note that the restriction $\Lambda_1(\lambda_0)$ (see also the beginning of Section~\ref{deltasec3})
of the Dirichlet-to-Neumann map
$\Lambda_{1/2}(\lambda_0)$ onto $H^1(\partial\Omega)$ is a non-negative self-adjoint operator in $L^2(\partial\Omega)$ with a bounded everywhere defined inverse
$\Lambda_1(\lambda_0)^{-1}$ in $L^2(\partial\Omega)$; the Neumann-to-Dirichlet map. Then also the square root $\sqrt{\Lambda_1(\lambda_0)}$ is a non-negative
self-adjoint operator in $L^2(\partial\Omega)$ which is boundedly invertible, and we have $\dom (\sqrt{\Lambda_1(\lambda_0)})=H^{1/2}(\partial\Omega)$; cf.
\cite[Proposition 3.2~(iii)]{BLLR16}. Hence $$\jmath=\sqrt{\Lambda_1(\lambda_0)}:H^{1/2}(\partial\Omega)\longrightarrow L^2(\partial\Omega)$$
is a possible choice for the definition of the scalar product $\langle\cdot,\cdot\rangle$ in \eqref{green123}.
}
\end{remark}

Following \cite[Section 1]{Ber78} one defines the adjoint $X^+$ of an operator $$X\in\cB\bigl(H^{1/2}(\partial\Omega),H^{-1/2}(\partial\Omega)\bigr)$$
in the rigging $H^{1/2}(\partial\Omega)\hookrightarrow L^2(\partial\Omega)\hookrightarrow H^{-1/2}(\partial\Omega)$ via
\begin{equation*}
 \langle X \varphi,\psi\rangle = \langle  \varphi,X^+\psi\rangle,\qquad \varphi,\psi\in H^{1/2}(\partial\Omega).
\end{equation*}
The imaginary part of the operator $X$ is defined by $\Im X=\frac{1}{2i}(X-X^+)$, the operator $X$ is self-adjoint if $X=X^+$
and $X$ is non-negative if $\langle X \varphi,\varphi\rangle\geq 0$ for all $\varphi\in H^{1/2}(\partial\Omega)$.

From the fact that the function $M^D_\alpha(\cdot)$
   in \eqref{eq:4.13} is $\sS_1$-regular with values in $\cB(L^2(\partial\Omega))$ we conclude
\begin{equation*}
 \Lambda_{1/2}(z)\in\cB\bigl(H^{1/2}(\partial\Omega),H^{-1/2}(\partial\Omega)\bigr),\qquad z\in\dC_+.
\end{equation*}
Together with Lemma~\ref{lemmilein} this yields the following corollary.

\begin{corollary}
Let $\Omega\subset \dR^2$ be an exterior domain  with a $C^\infty$-smooth boundary
and let $\Lambda_{1/2}(\cdot)$ be the Dirichlet-to-Neumann map defined in \eqref{Dir-to-Neuman_map}. Then the following holds.
\begin{itemize}
 \item [{\rm (i)}] The limit $\Lambda_{1/2}(\lambda+i0)  = \lim_{\varepsilon \rightarrow  +0} \Lambda_{1/2}(\lambda + i\varepsilon)$ exists for a.e. $\lambda\in\dR$ in the norm
 of $\cB(H^{1/2}(\partial\Omega),H^{-1/2}(\partial\Omega))$;
 \item [{\rm (ii)}] $\Lambda_{1/2}(\lambda+i0)\in\cB(H^{1/2}(\partial\Omega),H^{-1/2}(\partial\Omega))$ is boundedly invertible for a.e. $\lambda\in\dR$;
 \item [{\rm (iii)}] $\Lambda_{1/2}(\lambda + i\varepsilon) - \Lambda_{1/2}(\lambda + i0) \in \sS_p(H^{1/2}(\partial\Omega),H^{-1/2}(\partial\Omega))$
 for $p\in(1,\infty]$, $\varepsilon>0$ and a.e. $\lambda\in\dR$, and
   \begin{equation*}
\lim_{\varepsilon \rightarrow +0} \bigl\Vert\Lambda_{1/2}(\lambda+ i\varepsilon) - \Lambda_{1/2}(\lambda+i0)\bigr\Vert_{\sS_p(H^{1/2}(\partial\Omega),H^{-1/2}(\partial\Omega))} =0;
  \end{equation*}
 \item [{\rm (iv)}] $\Im \Lambda_{1/2}(\lambda+i0) = \lim_{\varepsilon \rightarrow  +0} \Im \Lambda_{1/2}(\lambda + i\varepsilon)$ exists for a.e. $\lambda\in\dR$ in the
 $\sS_1(H^{1/2}(\partial\Omega),H^{-1/2}(\partial\Omega))$-norm and $-\Im \Lambda_{1/2}(\lambda + i0) \ge 0$.
\end{itemize}
\end{corollary}

\subsection{Scattering matrix for the Neumann and Robin realization}\label{nrsec_N-Roben}

In this subsection we discuss a representation of the scattering matrix for the pair $\{A_N,A_\alpha\}$ consisting of the self-adjoint Neumann
and Robin operator associated to $\mathscr{L}$ in \eqref{adan}
and \eqref{robin}. Here $\Omega$ is an exterior domain in $\dR^2$ or $\dR^3$; in both situations it is known from \cite{BLLLP10,G11} that the trace class
condition \eqref{trace} for the resolvent  difference is satisfied; cf. Remark~\ref{remarkii}.

In a similar way as in the previous subsection we first define the
Neumann-to-Dirichlet map $\cN(z)$ as an operator in
$L^2(\partial\Omega)$ for all $z\in\rho(A_N)$. Recall first that for
$\varphi\in L^2(\partial\Omega)$ and $z \in\rho(A_N)$ the boundary value problem
\begin{equation}\label{bvpnn}
 -\Delta f_z + Vf_z = z f_z,\qquad \gamma_N f_z=\varphi,
\end{equation}
admits a unique solution $f_z \in H^{3/2}_\Delta(\Omega)$. The corresponding solution
operator is given  by
\begin{equation}\label{pp}
 P_N(z):L^2(\partial\Omega) \longrightarrow H^{3/2}_\Delta(\Omega)\subset L^{2}(\Omega),\qquad \varphi\mapsto f_z.
\end{equation}
For $z\in\rho(A_N)$ the {\it Neumann-to-Dirichlet map} is defined by
\begin{equation}\label{ndmap}
 \cN(z):L^2(\partial\Omega)\longrightarrow L^2(\partial\Omega),\qquad \varphi\mapsto \gamma_D P_N(z)\varphi.
\end{equation}
It is clear that $\cN(z)$ maps
Neumann boundary values $\gamma_N f_z$ of the solutions $f_z\in
H^{3/2}_\Delta(\Omega)$ of \eqref{bvpnn} onto their Dirichlet boundary values $\gamma_D
f_z$; here $\gamma_N$ and $\gamma_D$ denote the extensions of the Dirichlet and
Neumann trace operators onto $H^{3/2}_\Delta(\Omega)$ from \eqref{gammadcont} and \eqref{gammancont}, respectively.
Since \eqref{bvpnn} admits a unique solution for  each  $\varphi\in L^2(\partial\Omega)$ it is
clear that the operators $P_N(z)$ and $\cN(z)$ are well defined on
$L^2(\partial\Omega)$.

In the next theorem the scattering matrix of the pair $\{A_N,A_\alpha\}$ is expressed in
terms of the limit values  of the Neumann-to-Dirichlet map $\cN(z)$ and the parameter
$\alpha$ in the boundary condition of the Robin realization $A_\alpha$. In contrast to Theorem~\ref{th:4.1} here it is also assumed that
$\alpha^{-1}\in L^\infty(\partial\Omega)$.

\begin{theorem}\label{nrthm}
Let $\Omega\subset \dR^n$, $n=2,3$, be an exterior domain  with a $C^\infty$-smooth boundary, let $V\in
L^\infty(\Omega)$ and $\alpha  \in L^\infty(\partial\Omega)$ be real valued functions such that $\alpha^{-1}\in L^\infty(\partial\Omega)$, and let  $A_N$
and  $A_\alpha$  be the self-adjoint Neumann and Robin realizations of $\mathscr{L} = -\Delta+V$ in
$L^2(\Omega)$ in \eqref{adan} and \eqref{robin}, respectively.
Moreover, let $\cN(\cdot)$ be the Neumann-to-Dirichlet map defined in \eqref{ndmap}.

Then
$\{A_N,A_\alpha\}$ is a complete scattering system and
\begin{equation*}
 L^2(\dR,d\lambda,\cH_\lambda),\qquad \cH_\lambda:=\overline{\ran (\Im \cN(\lambda+i0))},
\end{equation*}
forms a spectral representation of $A_N^{ac}$ such that for a.e. $\lambda \in \dR$
the scattering matrix $\{S(A_N,A_\alpha;\lambda)\}_{\lambda\in\dR}$ of the scattering system $\{A_N,A_\alpha\}$ admits
the representation
  \begin{equation*}
 S(A_N,A_\alpha;\lambda)=I_{\cH_\lambda}+2i\sqrt{\Im
\cN(\lambda + i0)}\,\bigl(I-\alpha\cN(\lambda+i0)\bigr)^{-1}\alpha  \sqrt{\Im
\cN(\lambda+i0)}.
   \end{equation*}
\end{theorem}
\begin{proof}
First we note that the assumption $\alpha^{-1}\in L^\infty(\partial\Omega)$
implies $A_N \cap A_\alpha = S$, where $S$ is the minimal operator associated to $\mathscr{L}$ in \eqref{smin}.
Recall that $S$ is closed, densely defined and simple by Lemma~\ref{slem}.
Define the operator $T$ as a restriction of $S^*$ by
\begin{equation*}
Tf = -\Delta f+Vf, \qquad \dom (T) = H^{3/2}_{\Delta}(\Omega),
\end{equation*}
and let
\begin{equation}\label{zhb}
\Gamma_0 f := \gamma_N f\quad\text{and} \quad \Gamma_1 f := \gamma_D f-\frac{1}{\alpha}\gamma_N f,\quad f\in\dom (T).
\end{equation}
We claim that $\Pi^N_\alpha = \{L^2(\partial\Omega),\Gamma_0,\Gamma_1\}$ is a $B$-generalized boundary triple for $S^*$ with the $\sS_1$-regular
Weyl function
\begin{equation}\label{man}
M^N_\alpha(z)=\cN(z)-\frac{1}{\alpha},\qquad z\in\rho(A_N),
\end{equation}
such that
\begin{equation}\label{kernis2}
 A_N=T\upharpoonright\ker(\Gamma_0)\quad\text{and}\quad A_\alpha=T\upharpoonright\ker(\Gamma_1).
\end{equation}

In fact, Green's identity \eqref{Green_formula_for_BGBtrip} is an immediate consequence of the definition of the boundary mappings and \eqref{green32},
and $\ran\Gamma_0=L^2(\partial\Omega)$ holds by \eqref{gammancont}. Moreover, $\dom (T)$ is dense in $\dom (S^*)$ with respect to the graph norm
by Lemma~\ref{slem} and
$A_\alpha=T\upharpoonright\ker(\Gamma_1)$ is clear from \eqref{robin}. Furthermore, the self-adjoint operator
$A_N$ in \eqref{adan} is contained
in $T\upharpoonright\ker(\Gamma_0)$ and since the latter is symmetric (a consequence of Green's identity \eqref{Green_formula_for_BGBtrip})
both operators coincide, that is, \eqref{kernis2} holds, and $\Pi^N_\alpha$ is a $B$-generalized boundary triple. For $f_z\in\ker(T-z)$, $z\in\rho(A_N)$, we have
\begin{equation*}
 \left(\cN(z) -
\frac{1}{\alpha}\right)\Gamma_0 f_z=\cN(z)\gamma_N f_z-\frac{1}{\alpha}\gamma_N f_z=\gamma_D f_z-\frac{1}{\alpha}\gamma_N f_z=\Gamma_1 f_z,
\end{equation*}
and hence the Weyl function $M_\alpha^N(\cdot)$ corresponding to $\Pi_\alpha^N$ is given by \eqref{man}.

It remains to check that the Weyl function $M_\alpha^N(\cdot)$
is $\sS_1$-regular. This is done in a similar way as in Theorem~\ref{th:4.1}.
Denote the $\gamma$-field associated to $\Pi^N_\alpha$ by $\gamma_\alpha^N(\cdot)$ and use
\begin{equation*}
 M^N_\alpha(z)=M^N_\alpha(\xi)^*+(z-\bar\xi)\gamma_\alpha^N(\xi)^* \gamma_\alpha^N(z)
\end{equation*}
with some fixed $\xi\in\rho(A_N)\cap\rho(A_\alpha)\cap\dR$ and all $z\in\rho(A_N)$. From \eqref{zhb}, \eqref{adan}, and \eqref{gammdgammn} we obtain
\begin{equation*}
 \gamma_\alpha^N(\xi)^*h=\Gamma_1(A_N-\bar \xi)^{-1}h=\gamma_D(A_N-\bar \xi)^{-1}h \in H^{3/2}(\partial\Omega)
\end{equation*}
and hence Lemma~\ref{hurra} yields
\begin{equation}\label{gammaspnn}
\gamma_\alpha^N(\xi)^*\in\cS_\frac{2(n-1)}{3}\bigl(L^2(\Omega),L^2(\partial\Omega)\bigr)
\end{equation}
and
\begin{equation}\label{gammaspspnn}
 \gamma_\alpha^N(z)\in\cS_\frac{2(n-1)}{3}\bigl(L^2(\partial\Omega),L^2(\Omega)\bigr)
\end{equation}
for all $z\in\rho(A_N)$. Now \eqref{idealipp} shows
\begin{equation*}
 (z-\bar\xi)\gamma_\alpha^N(\xi)^* \gamma_\alpha^N(z)\in\cS_{\frac{n-1}{3}}\bigl(L^2(\partial\Omega)\bigr),\qquad z\in\rho(A_N).
\end{equation*}
Since $\cS_{(n-1)/3}(L^2(\partial\Omega))\subset \sS_1(L^2(\partial\Omega))$
for $n=2,3$, and $M^N_\alpha(\xi) = M^N_\alpha(\xi)^*$ we conclude that 
$$K(z) := M^N_\alpha(z) - M^N_\alpha(\xi) \in \sS_1\bigl(L^2(\partial\Omega)\bigr),\qquad z \in \dC_+.$$
Because $M^N_\alpha(\cdot)$ is a strict Nevanlinna function $K(\cdot)$ is strict.
Let us show that $C := M^N_\alpha(\xi) = \cN(\xi) - \frac{1}{\alpha}$ is invertible. In fact, since
$\frac{1}{\alpha}$ is a boundedly invertible operator and $\cN(\xi)$ is a compact operator it follows that $M^N_\alpha(\xi)$ is a Fredholm operator.
Furthermore, $\ker(M^N_\alpha(\xi))$ is trivial as otherwise there is a non-trivial function $f_\xi \in \ker(T-\xi)$ such that $\Gamma_1f_\xi = 0$,
that is, $f_\xi \in \ker(A_\alpha - \xi)$. But $\xi \in \rho(A_N) \cap \rho(A_\alpha) \cap \dR$ yields $f_\xi = 0$; a contradiction. Thus 
$\ker(M^N_\alpha(\xi)) = \{0\}$ and hence $C := M^N_\alpha(\xi)$ is boundedly invertible. Therefore,
the Weyl function $M_\alpha^N(\cdot)$ is $\sS_1$-regular.
Now the assertions in Theorem~\ref{nrthm} follow from Theorem~\ref{th:3.1},
\begin{equation*}
 \Im M_\alpha^N(z)=\Im\cN(z),\quad M_\alpha^N(z)^{-1}=-\bigl(I-\alpha \cN(z)\bigr)^{-1}\alpha,\qquad z\in\dC_+,
\end{equation*}
and
\begin{equation*}
 \Im M_\alpha^N(\lambda+i0)=\Im\cN(\lambda+i0),\quad M_\alpha^N(\lambda+i0)^{-1}=-\bigl(I-\alpha \cN(\lambda+i0)\bigr)^{-1}\alpha
\end{equation*}
for a.e. $\lambda\in\dR$. 
\end{proof}

\begin{remark}\label{remarkii}
{\rm
From \eqref{gammaspnn} and \eqref{gammaspspnn} one concludes in the same way as in Remark~\ref{remarki} that
Krein's formula in Proposition~\ref{prop:2.4}~(iii) and the property \eqref {idealipp} leads to
\begin{equation*}
 (A_\alpha-z)^{-1}-(A_D-z)^{-1}=-\gamma_\alpha^N(z) M_\alpha^N(z)^{-1}\gamma_\alpha^N(\bar z)^*\in\cS_{\frac{n-1}{3}}(L^2(\Omega));
\end{equation*}
for all $z\in\rho(A_\alpha)\cap\rho(A_N)$; cf. \cite{BLLLP10,G11}.
}
\end{remark}

\begin{remark}\label{rem:4.3}
{\rm The definition of the boundary triples $\Pi^D_\alpha$ and $\Pi^N_\alpha$ in
Theorems~\ref{th:4.1} and \ref{nrthm} given for an exterior domain $\Omega$, and
the form and properties of the corresponding Weyl functions remain the same in the case
of a bounded domain $\Omega$ with smooth boundary. The constructions and properties are
only based on the compactness and smoothness of $\partial\Omega$. }
\end{remark}

\section{Schr\"{o}dinger operators with interactions supported on hypersurfaces}\label{deltasec}

In this section we investigate scattering systems consisting of Schr\"{o}dinger operators in $\dR^n$. Here the Euclidean space is decomposed into
a smooth bounded domain and its complement,
and the usual self-adjoint Schr\"{o}dinger operator on the whole space is compared with
the orthogonal sum of the Dirichlet or Neumann operators on the subdomains in Section~\ref{dcoup} and \ref{ncoup}, and with a Schr\"{o}dinger operator with a singular
$\delta$-potential supported on the interface in Section~\ref{deltasec3}. In our main results Theorem~\ref{Th_Scat_Mat_for_Dir-Free}, \ref{nfthm}, and \ref{dddthm}
we obtain explicit forms of the scattering matrices in terms of Dirichlet-to-Neumann or Neumann-to-Dirichlet maps. As in Section~\ref{nrsec} the strategy in the
proofs is to apply the general result Theorem~\ref{th:3.1} to suitable $B$-generalized boundary triples. Here we shall assume for convenience that a simplicity condition
for the underlying symmetric operator is satisfied; this condition can be dropped in which case Corollary \ref{cor:3.2} would yield a slightly more
involved representation of the scattering matrix. We also refer the interested reader to Remarks~\ref{remarkiii}, \ref{remarkiv}, and \ref{remarkv}, where singular value
estimates due to Birman, Grubb and others are revisited.

\subsection{Preliminaries on orthogonal sums and couplings of Schr\"{o}dinger operators}

Let $\Omega_-\subset\dR^n$ be a bounded domain with $C^\infty$-smooth boundary $\partial\Omega_-$ and let
$\Omega_+ := \dR^n\setminus\overline{\Omega_-}$ be the corresponding
$C^\infty$-smooth exterior domain. Denote the common boundary of $\Omega_+$ and $\Omega_-$ by
$\cC := \partial\Omega_{\pm}$. Throughout this section we consider a
Schr\"{o}dinger differential  expression with a bounded, measurable, real valued potential $V$ on $\dR^n$,
\begin{equation}\label{Schrod_expres2}
 \mathscr{L} = -\Delta  + V,\qquad V \in L^\infty(\dR^n).
\end{equation}

In the following we shall adapt the notation from Section~\ref{prelisec} in an obvious way, e.g. $H^s(\Omega_\pm)$ and $H^r(\cC)$ denote the Sobolev spaces
on $\Omega_\pm$ and the common boundary (or interface) $\cC$, respectively, the spaces $H^s_\Delta(\Omega_\pm)$, $s\in [0,2]$, are defined and equipped with scalar
products as in \eqref{hs1}--\eqref{hs2}, and we shall use the notation
\begin{equation*}
 H^s_\Delta(\dR^n\setminus\cC):=H^s_\Delta(\Omega_+)\times H^s_\Delta(\Omega_-),\qquad s\in [0,2].
\end{equation*}
A function $f:\dR^n\rightarrow\dC$ is often written in a two component form $f=\{f^+,f^-\}$,
where $f^\pm:\Omega_\pm\rightarrow\dC$ denote the restrictions of $f$ onto $\Omega_\pm$.
The Dirichlet and Neumann trace operators will be denoted by $\gamma_D^\pm$ and $\gamma_N^\pm$,
and we emphasize that the Neumann trace is taken with respect to the outer normal of $\Omega_\pm$. In particular, $\gamma_N^+f^++\gamma_N^-f^-=0$ for a function
$f=\{f^+,f^-\}\in H^2(\dR^n)$. We also note that the mapping properties of the Dirichlet and Neumann trace operators in \eqref{gammadcont} and \eqref{gammancont}
are valid for both domains $\Omega_+$ and $\Omega_-$, and the same is true for the extensions of Green's identity in \eqref{green32} and \eqref{green12}, respectively.
Furthermore, we shall use in the proofs in Section~\ref{dcoup} and Section~\ref{ncoup} that $\gamma_D^\pm$ and $\gamma_N^\pm$ admit continuous extensions
\begin{equation*}
 \gamma_D^\pm:H^0_\Delta(\Omega_\pm)\rightarrow H^{-1/2}(\cC)\quad\text{and}\quad\gamma_N^\pm:H^0_\Delta(\Omega_\pm)\rightarrow H^{-3/2}(\cC)
\end{equation*}
and that Green's identity extends to $f_\pm\in H^2(\Omega_\pm)$ and $g_\pm\in H^0_\Delta(\Omega_\pm)$ in the form
\begin{equation}\label{greenmax}
 (-\Delta f_\pm,g_\pm)_{L^2(\Omega_\pm)}-(f_\pm,-\Delta g_\pm)_{L^2(\Omega_\pm)}=\langle \gamma_D^\pm f_\pm,\gamma_N^\pm g_\pm\rangle - \langle \gamma_N^\pm f_\pm,\gamma_D^\pm g_\pm\rangle;
\end{equation}
cf. \cite{LioMag71} and \cite[Chapter I, Theorem 3.3 and Corollary 3.3]{Gru68}. In \eqref{greenmax} the inner products
$\langle\cdot,\cdot\rangle$ on the right hand side denote the continuations of the $L^2(\cC)$ inner product onto $H^{3/2}(\cC)\times H^{-3/2}(\cC)$ and
$H^{1/2}(\cC)\times H^{-1/2}(\cC)$, respectively, and in the following it will always be clear from the context which duality is used; cf. \eqref{rig}--\eqref{green123}.

The differential expression \eqref{Schrod_expres2} induces self-adjoint operators in $L^2(\dR^n)$. The natural self-adjoint realization is the free
Schr\"{o}dinger operator,
\begin{equation}\label{aaa}
 A_{\rm free} f = \mathscr{L}f,\qquad \dom (A_{\rm free})=H^2(\dR^n),
\end{equation}
which is semibounded from below. Clearly the functions in $\dom (A_{\rm free})$ do not reflect the decomposition of $\dR^n$ into the domains
$\Omega_+$ and $\Omega_-$. Furthermore, we will make use of the self-adjoint orthogonal sum
\begin{equation}\label{Dirich_Oper_in_Rn}
\begin{split}
A_D &= A^+_D\oplus A^-_D,\\
\dom (A_D) &=  \bigl\{f=\{f^+,f^-\}\in H^2(\Omega_+) \oplus H^2(\Omega_-):\gamma_D^+f^+=\gamma_D^-f^-=0\bigr\},
\end{split}
\end{equation}
of the self-adjoint Dirichlet operators $A_D^\pm$ in $L^2(\Omega_\pm)$ in \eqref{adan},
and of the self-adjoint orthogonal sum
\begin{equation}\label{eq:5.5}
\begin{split}
A_N &= A^+_N\oplus A^-_N,\\
\dom (A_N) &=  \bigl\{f=\{f^+,f^-\}\in H^2(\Omega_+) \oplus H^2(\Omega_-):\gamma_N^+f^+=\gamma_N^-f^-=0\bigr\},
\end{split}
\end{equation}
of the self-adjoint Neumann operators $A_N^\pm$ in $L^2(\Omega_\pm)$ in \eqref{adan}. We shall sometimes refer to $A_D$ as Dirichlet realization of $\mathscr{L}$
with respect to $\cC$ and to $A_N$ as Neumann realization of $\mathscr{L}$
with respect to $\cC$.
The properties of $A_D^\pm$ and $A_N^\pm$
extend in a natural way to their orthogonal sums $A_D$ and $A_N$ in \eqref{Dirich_Oper_in_Rn} and \eqref{eq:5.5}, respectively.
In particular, the Dirichlet realization $A_D$ and the Neumann realization $A_N$ of $\mathscr{L}$ with respect to $\cC$ are both semibounded from below.

\subsection{Scattering matrix for the free Schr\"{o}dinger operator and the Dirichlet realization with respect to $\cC$}\label{dcoup}
We shall derive a representation for the scattering matrix of the scattering  system
$\{A_D,A_{\rm free}\}$ in $\dR^2$.
Let $\Lambda^\pm_{1/2}(z):  H^{1/2}(\cC) \mapsto  H^{-1/2}(\cC)$ be the Dirichlet-to-Neumann map
defined in \eqref{Dir-to-Neuman_map}  with respect to $\Omega_\pm$, that is,
\begin{equation}\label{dnpl}
 \Lambda^\pm_{1/2}(z)\gamma_D^\pm f^\pm_z=\gamma_N^\pm f^\pm_z
\end{equation}
holds for any solution $f^\pm_z \in
H^{1}(\Omega_\pm)$ of the equation $-\Delta f^\pm_z + V_\pm f^\pm_z =z f^\pm_z$ and $z\in\rho(A_D^\pm)$. Furthermore, define
the operator-valued function
$\Lambda_{1/2}(\cdot)$ by
\begin{equation}\label{eq:5.10}
\Lambda_{1/2}(z) := \Lambda^+_{1/2}(z) + \Lambda^-_{1/2}(z): H^{1/2}(\cC) \longrightarrow H^{-1/2}(\cC), \quad z \in  \rho(A_D).
\end{equation}
\begin{theorem}\label{Th_Scat_Mat_for_Dir-Free}
Let $\Omega_\pm\subset\dR^2$  be as above, let $V\in L^\infty(\dR^2)$ be a real valued function,
and let $A_{\rm free}$ and $A_{D}$ be the self-adjoint Schr\"{o}dinger operators
in $L^2(\dR^2)$ in \eqref{aaa} and \eqref{Dirich_Oper_in_Rn}, respectively.  Moreover, let
$\Lambda_{1/2}(\cdot)$  be given by \eqref{eq:5.10} and let
  \begin{equation}\label{Weyl-func_for_Dir_in_whole_space}
M^D_{\rm free}(z) := -\widetilde{\jmath^{-1}}\Lambda_{1/2}(z)\jmath^{-1}, \quad z\in \dC_+,
     \end{equation}
where $\jmath: H^{1/2}(\cC) \longrightarrow L^2(\cC)$ denotes some uniformly positive self-adjoint operator
in $L^2(\cC)$ with $\dom(\jmath)=H^{1/2}(\cC)$ as in \eqref{rig}--\eqref{jjj}.

Then $\{A_D,A_{\rm free}\}$ is a complete scattering system. If the symmetric operator $S := A_D \cap A_{\rm free}$ has no eigenvalues then
\begin{equation*}
L^2(\dR,d\lambda,\cH_\lambda),\qquad \cH_\lambda:=\overline{\ran\bigl(\Im M^D_{\rm free}(\lambda+i0)\bigr)},
\end{equation*}
forms a spectral representation of $A^{ac}_D$ such that for a.e. $\lambda\in\dR$ the scattering matrix
$\{S(A_D,A_{\rm free};\lambda)\}_{\lambda\in\dR}$ of the scattering system $\{A_{D}, A_{\rm free}\}$
admits the representation
\begin{equation*}
S(A_D,A_{\rm free};\lambda) = I_{\cH_\lambda}-2i\sqrt{\Im M^D_{\rm free}(\lambda+i0)} M^D_{\rm free}(\lambda+i0)^{-1} \sqrt{\Im M^D_{\rm free}(\lambda+i0)}.
 \end{equation*}
\end{theorem}
\begin{proof}
The closed symmetric operator $S=A_D\cap A_{\rm free}$ in $L^2(\dR^2)$ is given by
\begin{equation}\label{ssi}
 \begin{split}
  Sf&=\mathscr{L}f,\\
  \dom (S)&=\bigl\{f=\{f^+,f^-\}\in H^2(\dR^2):\gamma_D^+f^+=\gamma_D^-f^-=0\bigr\}.
 \end{split}
\end{equation}
It is clear that $S$
is a closed extension of the orthogonal sum of the minimal operators $S^+\oplus S^-$ associated to the restriction of
$\mathscr{L}$ onto $\Omega_+$ and $\Omega_-$ as in \eqref{smin} and Lemma~\ref{slem}. It follows that $S$ is densely defined and since we have
assumed that $S$ has no eigenvalues
it follows from
\cite[Corollary 4.4]{BR15} that $S$ is simple. We claim that the adjoint $S^*$ is given by
\begin{equation*}
 \begin{split}
  S^*f&=\mathscr{L}f,\\
  \dom (S^*)&=\bigl\{f=\{f^+,f^-\}\in H^0_\Delta(\dR^2\setminus\cC):\gamma_D^+f^+=\gamma_D^-f^-\bigr\}.
 \end{split}
\end{equation*}
In fact, since $S^*\subset (S^+)^*\oplus (S^-)^*$ it follows that
$$\dom (S^*)\subset  H^0_\Delta(\dR^2\setminus\cC)= \dom (S^+)^* \times \dom (S^-)^*$$ and that
$S^* f=\mathscr L f$ for $f\in\dom (S^*)$. Therefore, we only have to verify that $f=\{f^+,f^-\}\in\dom (S^*)$ satisfies the interface condition
\begin{equation}\label{dcondi}
 \gamma_D^+f^+=\gamma_D^-f^-.
\end{equation}
Assume that for $f=\{f^+,f^-\}\in \dom (S^*)$ and all $h=\{h^+,h^-\}\in \dom (S)$ we have
\begin{equation*}
 (Sh,f)_{L^2(\dR^2)}=(h,S^*f)_{L^2(\dR^2)},
\end{equation*}
that is,
\begin{equation*}
\begin{split}
(-\Delta h^+,f^+)_{L^2(\Omega_+)}+&(-\Delta h^-,f^-)_{L^2(\Omega_-)}\\&\quad =(h^+,-\Delta f^+)_{L^2(\Omega_+)}+(h^-,-\Delta f^-)_{L^2(\Omega_-)}.
\end{split}
\end{equation*}
Then it follows from Green's identity \eqref{greenmax} and the conditions $\gamma_D^\pm h^\pm=0$ and $\gamma_N^+h^++\gamma_N^-h^-=0$ that
\begin{equation*}
\begin{split}
 0&=(-\Delta h^+,f^+)_{L^2(\Omega_+)}-(h^+,-\Delta f^+)_{L^2(\Omega_+)}\\
  &\qquad\qquad\qquad\qquad\qquad\qquad +(-\Delta h^-,f^-)_{L^2(\Omega_-)} -(h^-,-\Delta f^-)_{L^2(\Omega_-)}\\
 &=\langle \gamma_D^+ h^+,\gamma_N^+ f^+\rangle - \langle \gamma_N^+ h^+,\gamma_D^+ f^+\rangle
   + \langle \gamma_D^- h^-,\gamma_N^- f^-\rangle - \langle \gamma_N^- h^-,\gamma_D^- f^-\rangle\\
 &=  \langle \gamma_N^- h^-,\gamma_D^+ f^+ - \gamma_D^- f^-\rangle
\end{split}
\end{equation*}
holds for all $h=\{h^+,h^-\}\in \dom (S)$. This implies \eqref{dcondi}.

Now we proceed in a similar manner as in the proofs of Theorem~\ref{th:4.1} and Theorem~\ref{nrthm} in the previous section.
We consider the operator $T$ defined as a restriction of $S^*$ by
 \begin{equation*}
 \begin{split}
  Tf&=\mathscr{L} f,\\
  \dom (T)&=\bigl\{f=\{f^+,f^-\}\in H^1_{\Delta}(\dR^2\setminus\cC): \gamma^+_{D}f^+ = \gamma^-_{D}f^-\bigr\},
 \end{split}
 \end{equation*}
and for $f\in\dom (T)$ we agree the notation
\begin{equation}\label{eq:5.18}
\gamma_Df := \gamma^+_Df^+ = \gamma^-_Df^-, \quad f = \{f^+,f^-\}\in\dom (T).
\end{equation}

We claim that $\Pi^D_{\rm free} = \{L^2(\cC),\Gamma_0,\Gamma_1\}$, where
  \begin{equation*}
\Gamma_0 f := \jmath\,\gamma_D f\quad  \text{and}\quad
\Gamma_1 f := -\widetilde{\jmath^{-1}} \left(\gamma^+_N f^+ + \gamma^-_N
f^-\right),\quad f\in\dom (T),
  \end{equation*}
is a $B$-generalized boundary triple with an $\sS_1$-regular Weyl function given by \eqref{Weyl-func_for_Dir_in_whole_space} such that
\begin{equation}\label{Covering_conddd}
A_D=T\upharpoonright\ker(\Gamma_0)\quad\text{and}\quad A_{\rm free} =T\upharpoonright\ker(\Gamma_1).
\end{equation}
In fact, for $f=\{f^+,f^-\}$, $g=\{g^+,g^-\}\in\dom (T)$ we compute with the help of Green's identity \eqref{green12} and \eqref{green123} that
\begin{equation*}
 \begin{split}
 &(\Gamma_1 f,\Gamma_0 g)-(\Gamma_0 f,\Gamma_1 g)\\
 &\qquad=\langle - \gamma^+_N f^+ - \gamma^-_N f^-,\gamma_D g\rangle-\langle\gamma_D f,- \gamma^+_N g^+ - \gamma^-_N g^-\rangle\\
 &\qquad=\langle\gamma_D^+f^+,\gamma_N^+g^+\rangle-\langle\gamma_N^+f^+,\gamma_D^+g^+\rangle+\langle\gamma_D^-f^-,\gamma_N^-g^-\rangle-\langle\gamma_N^-f^-,\gamma_D^-g^-\rangle\\
 &\qquad=(-\Delta f^+,g^+)-(f^+,-\Delta g^+)+(-\Delta f^-,g^-)-(f^-,-\Delta g^-)\\
 &\qquad=(Tf,g)-(f,Tg)
 \end{split}
\end{equation*}
and \eqref{gammadcont} implies $\ran(\Gamma_0)=L^2(\cC)$ in the present situation; cf. the proof of Theorem~\ref{th:4.1}. Since $T\upharpoonright\ker(\Gamma_0)$ and
$T\upharpoonright\ker(\Gamma_1)$ are both symmetric operators by \eqref{Green_formula_for_BGBtrip}, and contain the self-adjoint operators
$A_D$ and $A_{\rm free}$, respectively, it follows that \eqref{Covering_conddd} is satisfied. Furthermore, as $S=A_D\cap A_{\rm free}$ it is clear that
the self-adjoint operator $A_D$ and $A_{\rm free}$ are disjoint extensions of $S$. It follows that
\begin{equation}\label{domsss}
\dom (A_D)+\dom (A_{\rm free})
\end{equation}
is dense in $\dom (S^*)$ with respect to the graph norm; cf. Proposition~\ref{prop:2.5}. Since the space \eqref{domsss} is contained in $\dom (T)\subset\dom (S^*)$ we conclude
$\overline T=S^*$.
Therefore $\Pi^D_{\rm free}$ is $B$-generalized boundary triple
such that \eqref{Covering_conddd} holds.

Next we show that the Weyl function $M^D_{\rm free}(\cdot)$ corresponding to $\Pi^D_{\rm free}$ is $\sS_1$-regular and has the form in \eqref{Weyl-func_for_Dir_in_whole_space}.
Let $f_z=\{f_z^+,f_z^-\}\in\ker(T-z)$, $z\in\rho(A_D)$, and use \eqref{dnpl} and \eqref{eq:5.10} to compute
\begin{equation*}
 \begin{split}
 -\widetilde{\jmath^{-1}}\Lambda_{1/2}(z)\jmath^{-1}\Gamma_0f_z&=-\widetilde{\jmath^{-1}}\bigl(\Lambda_{1/2}^+(z)+\Lambda_{1/2}(z)^-\bigr)\gamma_D f_z\\
 &=-\widetilde{\jmath^{-1}}(\gamma_N^+f_z^++\gamma_N^-f_z^-)\\
 &=\Gamma_1 f_z.
 \end{split}
\end{equation*}
Hence the Weyl function is $M^D_{\rm free}(z)=-\widetilde{\jmath^{-1}}\Lambda_{1/2}(z)\jmath^{-1}$. In order to see that $M^D_{\rm free}(\cdot)$ is $\sS_1$-regular we
proceed in the same way as in the proof of Theorem~\ref{th:4.1}. Let $\gamma^D_{\rm free}(\cdot)$ be the $\gamma$-field corresponding to
the $B$-generalized boundary triple  $\Pi^D_{\rm free}$ and use
\begin{equation}\label{lplp44}
 M^D_{\rm free}(z)=M^D_{\rm free}(\xi)^*+(z-\bar\xi)\gamma_{\rm free}^D(\xi)^* \gamma_{\rm free}^D(z)
\end{equation}
(see \eqref{gutgut}) with some $\xi \in \rho(A_D) \cap \rho(A_{\rm free}) \cap (-\infty,\essinf V)$ and all $z\in\rho(A_D)$.
For $h=\{h^+,h^-\}\in L^2(\dR^n)$ we have
\begin{equation}\label{gammafreesp0}
\begin{split}
 \gamma_{\rm free}^D(\xi)^*h&=\Gamma_1(A_D-\bar \xi)^{-1}h\\
 &=-\widetilde{\jmath^{-1}}\bigl(\gamma_N^+(A_D^+-\bar \xi)^{-1}h^+ + \gamma_N^-(A_D^--\bar \xi)^{-1}h^-\bigr)
 \end{split}
\end{equation}
and since $\dom (A_D)\subset H^2(\Omega_+)\times H^2(\Omega_-)$ we conclude from \eqref{gammdgammn} that
\begin{equation*}
 \gamma_N^+(A_D^+-\bar \xi)^{-1}h^+ + \gamma_N^-(A_D^--\bar \xi)^{-1}h^-\in H^{1/2}(\cC).
\end{equation*}
 As in the proof of Theorem~\ref{th:4.1} it then follows from \eqref{hohopasst} that
\begin{equation}\label{gammafreespff}
\gamma_{\rm free}^D(\xi)^*\in\cS_1\bigl(L^2(\dR^n),L^2(\cC)\bigr)
\end{equation}
and $\gamma_{\rm free}^D(z)\in\cS_1(L^2(\cC),L^2(\dR^n))$ for all $z\in\rho(A_D)$. Hence \eqref{lplp44} yields that
$$K(z) := M^D_{\rm free}(z) - M^D_{\rm free}(\xi) \in \sS_1\bigl(L^2(\partial\Omega)\bigr),\quad z \in \dC_+,$$ 
where it was used that
$M^D_{\rm free}(\xi)^* = M^D_{\rm free}(\xi)$. Let us show that $M^D_{\rm free}(\xi)$ is boundedly invertible. For
$\xi \in \rho(A_D) \cap \rho(A_{\rm free}) \cap (-\infty,\essinf V)$ one checks that
the operators $\widetilde{\jmath^{-1}}\Lambda^\pm_{1/2}(\xi)\jmath^{-1}$ are non-negative and
the same considerations as in the end of the proof of Theorem \ref{th:4.1} show that these operators are surjective and boundedly invertible, 
and hence uniformly positive. This implies that also
$$\widetilde{\jmath^{-1}}\bigl(\Lambda^+_{1/2}(\xi) + \Lambda^-_{1/2}(\xi)\bigr)\jmath^{-1}$$ is uniformly positive.
Hence, $M^D_{\rm free}(\xi)$ is boundedly invertible which shows that $M^D_{\rm free}(\cdot)$
is $\sS_1$-regular.
Now the assertions follow directly from Theorem~\ref{th:3.1}. 
\end{proof}

\begin{remark}\label{remarkiii}
{\rm
As in Remarks \ref{remarki} and \ref{remarkii} it follows from \eqref{gammafreesp0} and \eqref{hohopasst} in the same way as in \eqref{gammafreespff} that
\begin{equation*}
 \gamma_{\rm free}^D(z)^*\in\cS_{n-1}\bigl(L^2(\dR^n),L^2(\cC)\bigr)
\end{equation*}
for $z\in\rho(A_D)$. This yields $\gamma_{\rm free}^D(z)\in\cS_{n-1}(L^2(\cC),L^2(\dR^n))$ for $z\in\rho(A_D)$
and hence  Krein's formula in Proposition~\ref{prop:2.4}~(iii) implies
\begin{equation*}
 (A_{\rm free}-z)^{-1}-(A_D-z)^{-1}=-\gamma_{\rm free}^D(z) M_{\rm free}^D(z)^{-1}\gamma^D_{\rm free}(\bar z)^*\in\cS_{\frac{n-1}{2}}(L^2(\dR^n));
\end{equation*}
for all $z\in\rho(A_{\rm free})\cap\rho(A_D)$; cf. \cite{Bir62,Gru84}. For further development with applications to the scattering theory
we also refer the reader to
\cite{DeiSim76} and \cite{RS79}.
}
\end{remark}

\begin{remark}
{\rm
In a similar way as in Remark~\ref{rem:4.2} there is a particularly convenient choice of the operator
$\jmath$ in \eqref{rig}--\eqref{jjj} in the present context. Namely, since for $z<\min\{\sigma(A_D^\pm),\sigma(A_N^\pm)\}$ the self-adjoint
operators
$$\sqrt{\Lambda_{1/2}^+(z)}\qquad\text{and}\qquad \sqrt{\Lambda_{1/2}^-(z)}$$
defined on $H^{1/2}(\cC)$ are non-negative and boundedly invertible in $L^2(\cC)$ it follows that
\begin{equation*}
 \jmath:=\sqrt{\Lambda_{1/2}^+(z)}+\sqrt{\Lambda_{1/2}^-(z)}: H^{1/2}(\cC)\longrightarrow L^2(\cC)
\end{equation*}
is a possible choice for the definition of the scalar product $\langle\cdot,\cdot\rangle$ in \eqref{green123}.
}
\end{remark}

\subsection{Scattering matrix for the free Schr\"{o}dinger operator and the Neumann realization with respect to $\cC$}\label{ncoup}
In this section we consider the pair $\{A_N,A_{\rm free}\}$ consisting of the orthogonal sum $A_N=A_N^+\oplus A_N^-$ of the Neumann operators in \eqref{eq:5.5} and the
free Schr\"{o}dinger operator in \eqref{aaa}. We first define the Neumann-to-Dirichlet maps
\begin{equation*}
\cN^\pm_{-1/2}(z):H^{-1/2}(\cC)\longrightarrow H^{1/2}(\cC),\qquad z\in\rho(A_N),
\end{equation*}
as extensions of the Neumann-to-Dirichlet maps on $L^2(\cC)$ defined in the beginning of Section~\ref{nrsec_N-Roben}.
More precisely, we recall that for $\phi^\pm \in H^{-1/2}(\cC)$ and $z \in \rho(A^\pm_N)$ the boundary value problem
\begin{equation}\label{eq:5.27}
-\Delta f^\pm + V_\pm f^\pm = zf^\pm, \quad \gamma^\pm_Nf^\pm = \phi^\pm,
\end{equation}
admits a unique solution $f^\pm_z \in H^1_\Delta(\Omega_\pm)$. The corresponding solution operator is denoted by
\begin{equation*}
\cP^\pm_N(z): H^{-1/2}(\cC) \longrightarrow H^1_\Delta(\cC)\subset L^2(\cC), \quad \phi^\pm \mapsto f^\pm_z.
\end{equation*}
Note that the restriction of $\cP^\pm_N(z)$ onto $L^2(\cC)$ coincides with the solution operator defined in \eqref{pp}.
For $z\in\rho(A_N^\pm)$ the {\it Neumann-to-Dirichlet map} is defined by
\begin{equation}\label{Neum-to-Dir_maps_for_inner-exter_dom}
\cN^\pm_{-1/2}(z): H^{-1/2}(\cC) \longrightarrow H^{1/2}(\cC), \quad \phi^\pm \mapsto \gamma^\pm_D\cP^\pm_N(z)\phi^\pm.
\end{equation}
Clearly, $\cN^\pm_{-1/2}(z)$ is an extension of the Neumann-to-Dirichlet map defined in \eqref{ndmap} onto $H^{-1/2}(\cC)$, the operators in \eqref{Neum-to-Dir_maps_for_inner-exter_dom}
map Neumann boundary values $\gamma^\pm_Nf^\pm_z$ of solutions $f^\pm_z \in H^1_\Delta(\Omega_\pm)$ of \eqref{eq:5.27} to the corresponding Dirichlet boundary values
$\gamma^\pm_Df^\pm_z \in H^{1/2}(\cC)$.

In the next theorem we obtain an expression for the scattering matrix of the pair $\{A_N,A_{\rm free}\}$ in terms of the sum
\begin{equation}\label{eq:5.28}
\cN_{-1/2}(z) := \cN^+_{-1/2}(z) + \cN^-_{-1/2}(z): H^{-1/2}(\cC) \longrightarrow H^{1/2}(\cC), \quad z \in \rho(A_N),
\end{equation}
of the Neumann-to-Dirichlet maps in \eqref{Neum-to-Dir_maps_for_inner-exter_dom}.
\begin{theorem}\label{nfthm}
Let $\Omega_\pm\subset\dR^2$  be as above, let $V\in L^\infty(\dR^2)$ be a real valued function,
and let $A_{\rm free}$ and $A_N$ be the self-adjoint Schr\"{o}dinger operators
in $L^2(\dR^2)$ in \eqref{aaa} and \eqref{eq:5.5}, respectively.  Moreover, let
$\cN_{-1/2}(\cdot)$  be given by \eqref{eq:5.28} and let
\begin{equation*}
M^N_{\rm free}(z) := \jmath\, \cN_{-1/2}(z)\,\widetilde\jmath, \quad z\in \dC_+,
    \end{equation*}
where $\jmath: H^{1/2}(\cC) \longrightarrow L^2(\cC)$ denotes some uniformly positive self-adjoint operator
in $L^2(\cC)$ with $\dom(\jmath)=H^{1/2}(\cC)$ as in \eqref{rig}--\eqref{jjj}.

Then $\{A_N,A_{\rm free}\}$ is a complete scattering system. If the symmetric operator $S := A_N \cap A_{\rm free}$ has no eigenvalues then
\begin{equation*}
L^2(\dR,d\lambda,\cH_\lambda),\qquad \cH_\lambda:=\overline{\ran\bigl(\Im M^N_{\rm free}(\lambda+i0)\bigr)},
\end{equation*}
forms a spectral representation of $A^{ac}_N$ such that for a.e. $\lambda\in\dR$ the scattering matrix
$\{S(A_N,A_{\rm free};\lambda)\}_{\lambda\in\dR}$ of the scattering system $\{A_N, A_{\rm free}\}$
admits the representation
\begin{equation*}
S(A_N,A_{\rm free};\lambda) = I_{\cH_\lambda}-2i\sqrt{\Im M^N_{\rm free}(\lambda+i0)} M^N_{\rm free}(\lambda+i0)^{-1} \sqrt{\Im M^N_{\rm free}(\lambda+i0)}.
 \end{equation*}
\end{theorem}
\begin{proof}
The proof of Theorem~\ref{nfthm} is very similar to the proof of Theorem~\ref{Th_Scat_Mat_for_Dir-Free}, and hence we present a sketch only.
Consider the closed symmetric operator
$S=A_N\cap A_{\rm free}$ in $L^2(\dR^2)$ which is given by
\begin{equation*}
 \begin{split}
  Sf&=\mathscr{L}f,\\
  \dom (S)&=\bigl\{f=\{f^+,f^-\}\in H^2(\dR^2):\gamma_N^+f^+=\gamma_N^-f^-=0\bigr\}.
 \end{split}
\end{equation*}
It follows that $S$ is densely defined, the assumption $\sigma_p(S)=\emptyset$ and same arguments as in \cite[Proof of Lemma 4.3]{BR15}
ensure that $S$ is simple, and a similar consideration as in the proof of Theorem~\ref{Th_Scat_Mat_for_Dir-Free} shows that
the adjoint $S^*$ is given by
\begin{equation*}
 \begin{split}
  S^*f&=\mathscr{L}f,\\
  \dom (S^*)&=\bigl\{f=\{f^+,f^-\}\in H^0_\Delta(\dR^2\setminus\cC):\gamma_N^+f^+=\gamma_N^-f^-\bigr\}.
 \end{split}
\end{equation*}
Next we consider the operator $T$ defined as a restriction of $S^*$ by
 \begin{equation*}
 \begin{split}
  Tf&=\mathscr{L} f,\\
  \dom (T)&=\bigl\{f=\{f^+,f^-\}\in H^1_{\Delta}(\dR^2\setminus\cC): \gamma^+_Nf^+ = \gamma^-_Nf^-\bigr\},
 \end{split}
 \end{equation*}
and one verifies in the same way as in the proof of Theorem~\ref{Th_Scat_Mat_for_Dir-Free} that $\Pi^N_{\rm free} = \{L^2(\cC),\Gamma_0,\Gamma_1\}$, where
  \begin{equation*}
\Gamma_0 f := \widetilde{\jmath^{-1}}\,\gamma_N^+ f^+\quad  \text{and}\quad
\Gamma_1 f := \jmath \left(\gamma^+_D f^+ - \gamma^-_D f^-\right),\quad f\in\dom (T),
  \end{equation*}
is a $B$-generalized boundary triple with the Weyl function $M_{\rm free}^N(\cdot)$ given by \eqref{Weyl-func_for_Dir_in_whole_space} such that
\begin{equation*}
A_N=T\upharpoonright\ker(\Gamma_0)\quad\text{and}\quad A_{\rm free} =T\upharpoonright\ker(\Gamma_1).
\end{equation*}
Let us show that the Weyl function $M_{\rm free}^N(\cdot)$ is $\sS_1$-regular.
Denote the $\gamma$-field corresponding to
the $B$-generalized boundary triple  $\Pi^N_{\rm free}$ by $\gamma_{\rm free}^N(\cdot)$ and use
\begin{equation}\label{lplp442}
 M^N_{\rm free}(z)=M^N_{\rm free}(\xi)^*+(z-\bar\xi)\gamma_{\rm free}^N(\xi)^* \gamma_{\rm free}^N(z)
\end{equation}
with some fixed $\xi\in\rho(A_N)\cap\rho(A_{\rm free})\cap (-\infty,\essinf V)$ and all $z\in\rho(A_N)$.
From \eqref{gammdgammn} and $\dom (A_N)\subset H^2(\Omega_+)\times H^2(\Omega_-)$ we conclude
for $h=\{h^+,h^-\}\in L^2(\dR^n)$ that
\begin{equation}\label{gammafreesp00}
\begin{split}
 \jmath^{-1}\gamma_{\rm free}^N(\xi)^*h&=\jmath^{-1}\Gamma_1(A_N-\bar \xi)^{-1}h\\
 &=\gamma_D^+(A_N^+-\bar \xi)^{-1}h^+ - \gamma_D^-(A_N^--\bar \xi)^{-1}h^-\in H^{3/2}(\cC).
 \end{split}
\end{equation}
Since $\jmath^{-1}\gamma_{\rm free}^N(\xi)^*\in\cB(L^2(\dR^2), H^{1/2}(\cC))$
Lemma~\ref{hurra} yields
\begin{equation*}
 \jmath^{-1}\gamma_{\rm free}^N(\xi)^*\in\cS_1\bigl(L^2(\dR^2), H^{1/2}(\cC)\bigr)
\end{equation*}
and hence
\begin{equation}\label{gammafreespqa}
\gamma_{\rm free}^N(\xi)^*\in\cS_1\bigl(L^2(\dR^2),L^2(\cC)\bigr).
\end{equation}
Therefore $\gamma_{\rm free}^N(z)\in\cS_1(L^2(\cC),L^2(\dR^2))$ for all $z\in\rho(A_N)$. Now it follows from \eqref{lplp442} that
$$K(z) := M^N_{\rm free}(z) - M^N_{\rm free}(\xi) \in \sS_1\bigl(L^2(\cC)\bigr),\quad z\in\dC_+,$$
where we have used that $M^N_{\rm free}(\xi) = M^N_{\rm free}(\xi)^*$.
It remains to show that $M^N_{\rm free}(\xi)$ is invertible, which follows from the same reasoning as in the end of the proof of Theorem \ref{Th_Scat_Mat_for_Dir-Free}.
Hence $M^N_{\rm free}(\cdot)$  is $\sS_1$-regular
and the assertions of Theorem~\ref{nfthm} follow directly from Theorem~\ref{th:3.1}. 
\end{proof}

\begin{remark}\label{remarkiv}
{\rm
As in Remark \ref{remarkiii} the considerations in \eqref{gammafreesp00} and \eqref{gammafreespqa} together with Lemma~\ref{hurra} show
\begin{equation*}
 \gamma_{\rm free}^N(z)^*\in\cS_{n-1}\bigl(L^2(\dR^n),L^2(\cC)\bigr),\quad\gamma_{\rm free}^N(z)\in\cS_{n-1}\bigl(L^2(\cC),L^2(\dR^n)\bigr)
 \end{equation*}
for all $z\in\rho(A_N)$. Hence
\begin{equation*}
 (A_{\rm free}-z)^{-1}-(A_N-z)^{-1}=-\gamma_{\rm free}^N(z) M_{\rm free}^N(z)^{-1}\gamma^N_{\rm free}(\bar z)^*\in\cS_{\frac{n-1}{2}}(L^2(\dR^n));
\end{equation*}
for all $z\in\rho(A_{\rm free})\cap\rho(A_N)$; cf. \cite{Gru84}.
}
\end{remark}

\subsection{Schr\"{o}dinger operators with $\delta$-potentials supported on hypersurfaces}\label{deltasec3}

In this third and last application on scattering matrices for coupled Schr\"{o}dinger operators we consider the pair $\{A_{\rm free}, A_{\delta,\alpha}\}$,
where  $\alpha\in L^\infty(\cC)$ is a real valued function and  $A_{\delta,\alpha}$
is a Schr\"{o}dinger operator with $\delta$-potential of strength $\alpha$
supported on the hypersurface $\cC$ defined by
\begin{equation}\label{adelta}
\begin{split}
 A_{\delta,\alpha} f&=-\Delta f +Vf,\\
 \dom (A_{\delta,\alpha})&=\left\{f= \begin{pmatrix}f^+\\f^-\end{pmatrix} \in H^{3/2}_\Delta(\dR^n\setminus\cC):
 \begin{matrix}\gamma_D^+f^+=\gamma_D^-f^-,\qquad\quad\!\\
\alpha\gamma_D^\pm f^\pm=\gamma_N^+ f^++\gamma_N^-f^-\end{matrix}\right\}.
\end{split}
 \end{equation}
Such type of Schr\"{o}dinger operators with singular interactions have attracted a lot of attention in the past; cf. \cite{E08} for a survey and e.g.
\cite{BLL13} for further references and an approach via boundary mappings closely related to the present considerations.
According to \cite[Theorem 3.5, Proposition 3.7, and Theorem 3.16]{BLL13} the operator  $A_{\delta,\alpha}$ in \eqref{adelta} is  self-adjoint in $L^2(\dR^n)$, semibounded
from below and coincides with the
self-adjoint operator associated to the closed  quadratic form
\begin{equation*}
 \mathfrak a_{\delta,\alpha}[f,g]=(\nabla f,\nabla g)+(Vf,g)-(\alpha\gamma_D^\pm f,\gamma_D^\pm g)_{L^2(\cC)},\quad f,g\in H^1(\dR^n).
\end{equation*}

We define the Dirichlet-to-Neumann maps
\begin{equation*}
\Lambda^\pm_1(z):H^1(\cC)\longrightarrow L^2(\cC),\qquad z\in\rho(A_D^\pm),
\end{equation*}
as restrictions of the Dirichlet-to-Neumann maps on $H^{1/2}(\cC)$ in \eqref{Dir-to-Neuman_map}; cf. Remark~\ref{rem:4.2}.
More precisely,
for $\phi^\pm \in H^1(\cC)$ and $z \in \rho(A^\pm_D)$ the boundary value problem
\begin{equation*}
-\Delta f^\pm + V_\pm f^\pm = zf^\pm, \quad \gamma^\pm_D f^\pm = \phi^\pm,
\end{equation*}
admits a unique solution $f^\pm_z \in H^{3/2}_\Delta(\Omega_\pm)$. The corresponding solution operators are denoted by
\begin{equation*}
\cP^\pm_D(z): H^1(\cC) \longrightarrow H^{3/2}_\Delta(\cC)\subset L^2(\cC), \quad \phi^\pm \mapsto f^\pm_z,
\end{equation*}
and it is clear that the restriction of $P^\pm_D(z)$ in \eqref{sol_oper_for_Dir} onto $H^1(\cC)$ coincides with $\cP_D^\pm(z)$.
For $z\in\rho(A_D^\pm)$ the Dirichlet-to-Neumann maps $\Lambda_1^\pm(\cdot)$ on $H^1(\cC)$ are given by
\begin{equation}\label{dneinser}
\Lambda^\pm_1(z): H^1(\cC) \longrightarrow L^2(\cC), \quad \phi^\pm \mapsto \gamma^\pm_N\cP^\pm_D(z)\phi^\pm,
\end{equation}
and by construction  $\Lambda^\pm_1(z)$ are the restrictions of the Dirichlet-to-Neumann maps $\Lambda^\pm_{1/2}(z)$
in \eqref{Dir-to-Neuman_map} onto $H^1(\cC)$.

In the next theorem we obtain an expression for the scattering matrix of the pair $\{A_{\rm free},A_{\delta,\alpha}\}$ in terms of the sum
\begin{equation}\label{eq:5.41}
\Lambda_1(z) := \Lambda^+_1(z) + \Lambda^-_1(z): H^1(\cC) \longrightarrow L^2(\cC), \quad z \in \rho(A_D),
\end{equation}
of the Dirichlet-to-Neumann maps in \eqref{dneinser}. Theorem~\ref{dddthm} and its proof can be viewed as a variant of Theorem~\ref{nrthm}; in the same way as in
Theorem~\ref{nrthm} it is assumed that $\alpha^{-1}\in L^\infty(\cC)$.
\begin{theorem}\label{dddthm}
Let $\Omega_\pm\subset \dR^n$, $n=2,3$, be as above, let $V\in
L^\infty(\dR^n)$ and $\alpha  \in L^\infty(\cC)$ be real valued functions such that $\alpha^{-1}\in L^\infty(\cC)$, and let
$A_{\rm free}$ and $A_{\delta,\alpha}$ be the self-adjoint realizations of the Schr\"odinger expression given by \eqref{aaa} and \eqref{adelta}, respectively.
Moreover, let $\Lambda_1(\cdot)$ be as in \eqref{eq:5.41}.

Then $\{A_{\rm free},A_{\delta,\alpha}\}$ is a complete scattering system.  If the symmetric operator $S := A_{\rm free}\cap A_{\delta,\alpha}$ has no eigenvalues then
\begin{equation*}
 L^2(\dR,d\lambda,\cH_\lambda),\qquad \cH_\lambda:=\overline{\ran (\Im (\Lambda_1(\lambda+i0))^{-1})},
\end{equation*}
forms a spectral representation of $A_{\rm free}^{ac}$ such that for a.e. $\lambda \in \dR$
the scattering matrix $\{S(A_{\rm free},A_{\delta,\alpha};\lambda)\}_{\lambda\in\dR}$ of the scattering system $\{A_{\rm free},A_{\delta,\alpha}\}$ admits
the representation
\begin{equation*}
\begin{split}
&S(A_{\rm free},A_{\delta,\alpha};\lambda)\\
&\quad=I_{\cH_\lambda}+2i\sqrt{\Im
\Lambda_1(\lambda + i0)^{-1}}\,\bigl(I-\alpha\Lambda_1(\lambda+i0)^{-1}\bigr)^{-1}\alpha  \sqrt{\Im\Lambda_1(\lambda+i0)^{-1}}.
\end{split}
\end{equation*}
\end{theorem}
\begin{proof}
Note first that the assumptions $\alpha^{-1}\in L^\infty(\cC)$ implies that the closed symmetric operator $S=A_{\rm free}\cap A_{\delta,\alpha}$ is given by
\begin{equation*}
 \begin{split}
  Sf&=\mathscr{L}f,\\
  \dom (S)&=\bigl\{f=\{f^+,f^-\}\in H^2(\dR^n):\gamma_D^+f^+=\gamma_D^-f^-=0\bigr\}
 \end{split}
\end{equation*}
and hence coincides with the symmetric operator $A_D\cap A_{\rm free}$ in \eqref{ssi} (in the case $n=2$). It follows from \cite[Corollary 4.4]{BR15} that
the operator $S$ is simple and as in the proof of Theorem~\ref{Th_Scat_Mat_for_Dir-Free} one verifies that
its adjoint $S^*$ is given by
\begin{equation*}
 \begin{split}
  S^*f&=\mathscr{L}f,\\
  \dom (S^*)&=\bigl\{f=\{f^+,f^-\}\in H^0_\Delta(\dR^n\setminus\cC):\gamma_D^+f^+=\gamma_D^-f^-\bigr\}.
 \end{split}
\end{equation*}

Next we define the operator $T$ by
\begin{equation}\label{tli}
 \begin{split}
  Tf&=\mathscr{L}f,\\
  \dom (T)&=\bigl\{f=\{f^+,f^-\}\in H^{3/2}_\Delta(\dR^n\setminus\cC):\gamma_D^+f^+=\gamma_D^-f^-\bigr\}
 \end{split}
\end{equation}
and for $f=\{f^+,f^-\}\in\dom (T)$ we write $\gamma_Df := \gamma^+_Df^+ = \gamma^-_Df^-$ as in \eqref{eq:5.18}.
We will show that $\Pi^{\rm free}_{\delta,\alpha}   = \{L^2(\cC),\Gamma_0,\Gamma_1\}$, where
\begin{equation*}
 \Gamma_0 f = \gamma_N^+f^++\gamma_N^-f^-,  \qquad f\in \dom (T),
\end{equation*}
 and
\begin{equation*}
 \Gamma_1 f=\gamma_D f-\frac{1}{\alpha}\bigl(\gamma_N^+f^++\gamma_N^-f^-\bigr),\qquad f \in\dom (T),
\end{equation*}
is a $B$-generalized boundary triple such that
\begin{equation}\label{condikk}
 A_{\rm free}=T\upharpoonright\ker(\Gamma_0)\quad\text{and}\quad A_{\delta,\alpha}=T\upharpoonright\ker(\Gamma_1),
\end{equation}
and the corresponding Weyl function
\begin{equation}\label{e}
M^{\rm free}_{\delta,\alpha}(z) :=  \Lambda_1(z)^{-1} - \frac{1}{\alpha}, \qquad z \in \dC_+,
\end{equation}
is $\sS_1$-regular.

In fact, for $f=\{f^+,f^-\}$, $g=\{g^+,g^-\}\in\dom (T)$ we compute with the help of Green's identity \eqref{green32} and the interface conditions
$\gamma_D^+ f^+=\gamma_D^-f^-$ and $\gamma_D^+ g^+=\gamma_D^-g^-$
that
\begin{equation*}
 \begin{split}
 (\Gamma_1 f,\Gamma_0 &g)-(\Gamma_0 f,\Gamma_1 g)\\
 &=\bigl(\gamma_D f- \alpha^{-1}(\gamma_N^+f^++\gamma_N^-f^-), \gamma_N^+g^++\gamma_N^-g^-\bigr) \\
 &\qquad\quad - \bigl(\gamma_N^+f^++\gamma_N^-f^-,\gamma_D g- \alpha^{-1}(\gamma_N^+g^++\gamma_N^-g^-)\bigr)\\
 &=\bigl(\gamma_D f, \gamma_N^+g^++\gamma_N^-g^-\bigr)  - \bigl(\gamma_N^+f^++\gamma_N^-f^-,\gamma_D g\bigr)\\
 &=(\gamma_D^+f^+,\gamma_N^+g^+)-(\gamma_N^+f^+,\gamma_D^+g^+)+(\gamma_D^-f^-,\gamma_N^-g^-)-(\gamma_N^-f^-,\gamma_D^-g^-)\\
 &=(-\Delta f^+,g^+)-(f^+,-\Delta g^+)+(-\Delta f^-,g^-)-(f^-,-\Delta g^-)\\
 &=(Tf,g)-(f,Tg),
 \end{split}
\end{equation*}
which shows \eqref{Green_formula_for_BGBtrip}. In order to show that $\Gamma_0$ is surjective we fix some $\lambda_0\in\dR$ such that $\lambda_0<\min\{\sigma(A_D),\sigma(A_N)\}$
and we note that the direct sum decomposition
\begin{equation*}
 \dom (T)=\dom (A_D)\,\dot +\,\ker(T-\lambda_0)
\end{equation*}
holds since $\lambda_0\in\rho(A_D)$. It follows from \eqref{tli} and \eqref{gammadcont} that $\gamma_D$ maps $\ker(T-\lambda_0)$ onto $H^1(\cC)$. As
$\Lambda_1^\pm(\lambda_0)=(\cN^\pm(\lambda_0))^{-1}$ (cf. \eqref{ndmap}) are uniformly positive self-adjoint operators in $L^2(\cC)$
it follows that also $\Lambda_1(\lambda_0)=\Lambda_1^+(\lambda_0)+\Lambda_1^-(\lambda_0)$ is a uniformly positive
self-adjoint operator in $L^2(\cC)$. Let $\psi\in L^2(\cC)$, choose
$\varphi\in H^1(\cC)$ and $f_{\lambda_0}=\{f_{\lambda_0}^+,f_{\lambda_0}^-\}\in\ker(T-\lambda_0)$ such that $\Lambda_1(\lambda_0)\varphi=\psi$ and $\gamma_D f_{\lambda_0}=\varphi$.
Then we have
\begin{equation*}
 \Gamma_0 f_{\lambda_0}=\gamma_N^+f^+_{\lambda_0}+\gamma_N^-f^-_{\lambda_0}=\Lambda_1(\lambda_0)\gamma_D f_{\lambda_0}=\Lambda_1(\lambda_0)\varphi=\psi
\end{equation*}
and this implies $\ran(\Gamma_0)=L^2(\cC)$.

It is not difficult to check that
$\dom (A_{\rm free})$ and $\dom (A_{\delta,\alpha})$ are contained in $\ker(\Gamma_0)$ and $\ker(\Gamma_1)$, respectively, and since
$A_{\rm free}$ and $A_{\delta,\alpha}$ are self-adjoint and $T\upharpoonright\ker(\Gamma_0)$ and $T\upharpoonright\ker(\Gamma_1)$
are symmetric by Green's identity \eqref{Green_formula_for_BGBtrip} it follows that \eqref{condikk} holds. From $S=A_{\rm free}\cap A_{\delta,\alpha}$ and
\begin{equation*}
 \dom(A_{\rm free}) + \dom(A_{\delta,\alpha}) \subset\dom (T) \subset\dom (S^*)
\end{equation*}
we conclude with the help of Proposition~\ref{prop:2.5}  that $\overline T=S^*$.
Hence $\Pi^{\rm free}_{\delta,\alpha}$ is a $B$-generalized boundary triple such that \eqref{condikk} is satisfied.

In order to show that the corresponding Weyl function is given by \eqref{e} let $f_z=\{f_z^+,f_z^-\}\in\ker(T-z)$ and $z\in\dC_+$. Then we have
\begin{equation*}
 \Lambda_1(z)\gamma_D f_z=\Lambda_1^+(z)\gamma_D^+ f_z^++\Lambda_1^-(z)^-\gamma_D^- f_z^-=\gamma_N^+f_z^++\gamma_N^-f_z^-=\Gamma_0 f_z
\end{equation*}
and since $\ker(\Lambda_1(z))=\{0\}$
we conclude
\begin{equation*}
 \left(\Lambda_1(z)^{-1}-\frac{1}{\alpha}\right)\Gamma_0 f_z=\gamma_D f_z-\frac{1}{\alpha}\bigl(\gamma_N^+f_z^+-\gamma_N^-f_z^-\bigr)=\Gamma_1 f_z.
\end{equation*}
This proves the representation \eqref{e}. In order to see that the Weyl function
$M^{\rm free}_{\delta,\alpha}(\cdot)$ is $\sS_1$-regular we argue in the same way as in the previous proofs.
Denote the $\gamma$-field corresponding to
the $B$-generalized boundary triple  $\Pi^{\rm free}_{\delta,\alpha}$ by $\gamma^{\rm free}_{\delta,\alpha}(\cdot)$ and use
\begin{equation}\label{lplp4422}
 M^{\rm free}_{\delta,\alpha}(z)=M^{\rm free}_{\delta,\alpha}(\xi)^*+(z-\bar\xi)\gamma^{\rm free}_{\delta,\alpha}(\xi)^* \gamma^{\rm free}_{\delta,\alpha}(z)
\end{equation}
with some $\xi\in\rho(A_{\rm free})\cap\rho(A_{\delta,\alpha})\cap\dR$ and all $z\in\rho(A_{\rm free})$.
For $h=\{h^+,h^-\}\in L^2(\dR^n)$ we have
\begin{equation*}
 \gamma_{\delta,\alpha}^{\rm free}(\xi)^*h=\Gamma_1(A_{\rm free}-\bar \xi)^{-1}h=\gamma_D(A_{\rm free}-\bar \xi)^{-1}h\in H^{3/2}(\cC)
\end{equation*}
and hence Lemma~\ref{hurra} yields
\begin{equation}\label{gammafreespqq}
\gamma^{\rm free}_{\delta,\alpha}(\xi)^*\in\cS_\frac{2(n-1)}{3}\bigl(L^2(\dR^n),L^2(\cC)\bigr).
\end{equation}
As before we conclude
\begin{equation}\label{okok}
\gamma^{\rm free}_{\delta,\alpha}(z)\in\cS_\frac{2(n-1)}{3}\bigl(L^2(\cC),L^2(\dR^n)\bigr),\qquad  z\in\rho(A_{\rm free}).
\end{equation}
It follows from \eqref{lplp4422} that 
$$K(z) := M^{\rm free}_{\delta,\alpha}(z) - M^{\rm free}_{\delta,\alpha}(\xi) \in \sS_1\bigl(L^2(\cC)\bigr),\quad z \in \dC_+,$$
where $M^{\rm free}_{\delta,\alpha}(\xi) = M^{\rm free}_{\delta,\alpha}(\xi)^*$ was used. Since the operator $\frac{1}{\alpha}$ is boundedly invertible and
$\ran(\Lambda_1(\xi)^{-1}) \subseteq H^1(\cC)$ the operator $M^{\rm free}_{\delta,\alpha}(\xi)$ is a Fredholm operator. Furthemore, 
$\ker(M^{\rm free}_{\delta,\alpha}(\xi))$ is trivial as otherwise there is an element $f_\xi \in \ker(T-\xi)$ with 
$\Gamma_1f_\xi = 0$, that is, $f_\xi \in \ker(A_{\delta,\alpha} - \xi)$. But $\xi \in \rho(A_{\delta,\alpha})$ implies $f_\xi = 0$; a contradiction. 
Hence $M^{\rm free}_{\delta,\alpha}(\xi)$ is invertible
and it follows that $M^{\rm free}_{\delta,\alpha}(\cdot)$ is $\sS_1$-regular for $n=2,3$.

The assertions in Theorem~\ref{dddthm} follow from
Theorem~\ref{th:3.1} and
\begin{equation*}
 \Im M_{\delta,\alpha}^{\rm free}(z)=\Im\Lambda_1(z),\quad M^{\rm free}_{\delta,\alpha}(z)^{-1}=-\bigl(I-\alpha \Lambda_1(z)^{-1}\bigr)^{-1}\alpha,\quad z\in\dC_+,
\end{equation*}
and
\begin{equation*}
\begin{split}
 \Im M_{\delta,\alpha}^{\rm free}(\lambda+i0)&=\Im\Lambda_1(\lambda+i0),\\
 M^{\rm free}_{\delta,\alpha}(\lambda+i0)^{-1}&=-\bigl(I-\alpha \Lambda_1(\lambda+i0)^{-1}\bigr)^{-1}
 \alpha,
 \end{split}
\end{equation*}
for a.e. $\lambda\in\dR$. 
\end{proof}

\begin{remark}\label{remarkv}
{\rm
As in previous remarks it follows from \eqref{gammafreespqq}--\eqref{okok} and Krein's formula that
\begin{equation*}
 (A_{\delta,\alpha}-z)^{-1}-(A_{\rm free}-z)^{-1}=-\gamma^{\rm free}_{\delta,\alpha}(z) M^{\rm free}_{\delta,\alpha}(z)^{-1}\gamma^{\rm free}_{\delta,\alpha}(\bar z)^*
 \in\cS_{\frac{n-1}{3}}(L^2(\dR^n))
\end{equation*}
for all $z\in\rho(A_{\rm free})\cap\rho(A_{\delta,\alpha})$; cf. \cite{BLL13}.
}
\end{remark}

\appendix

\section{Spectral representation and scattering matrix}

\subsection{Spectral representations and operator spectral integrals}\label{A.I}

Let $E(\cdot)$ be a spectral measure in the separable Hilbert space
$\sH$ defined on the Borel sets $\mathfrak B(\dR)$ of the real axis $\dR$.
Further, let $C$ be a Hilbert-Schmidt operator in $\sH$. Obviously,
$\Sigma(\delta) := C^*E(\delta)C$, $\delta \in \mathfrak B(\dR)$
defines a trace class valued measure on $\mathfrak B(\dR)$ of finite variation;
cf. \cite[Lemma 3.11]{BW83}
The measure admits a unique decomposition
\begin{equation*}
\Sigma(\cdot) =  \Sigma^{s}(\cdot) + \Sigma^{ac}(\cdot)
\end{equation*}
into a singular measure $\Sigma^{s}(\cdot) = C^*E^s(\cdot)C$ and an
absolutely continuous measure $\Sigma^{ac}(\cdot) = C^*E^{ac}(\cdot)C$.
From \cite[Proposition 3.13]{BW83} it follows that the trace class valued function $\Sigma(\lambda) :=
C^*E((-\infty,\lambda))C$ admits a derivative $K(\lambda) :=
\frac{d}{d\lambda}\Sigma(\lambda) \ge 0$ in the trace class norm
for a.e. $\lambda \in \dR$ with respect  the Lebesgue measure $d\lambda$ such that
\begin{equation*}
\Sigma^{ac}(\delta) = \int_\delta K(\lambda) d\lambda, \quad \delta \in \mathfrak B(\dR).
\end{equation*}
By $\cH_\lambda := \overline{\ran(K(\lambda))} \subseteq \sH$ we define a
measurable family of subspaces in $\sH$. The
orthogonal projection $P(\lambda)$ from $\sH$ onto $\cH_\lambda$ form a
measurable family of projections which defines by
\begin{equation*}
(Pf)(\lambda) := P(\lambda)f(\lambda), \quad f \in L^2(\dR,d\lambda,\sH),
\end{equation*}
an orthogonal projection from $L^2(\dR,d\lambda,\sH)$ onto a subspace which is denoted
by $L^2(\dR,d\lambda,\cH_\lambda)$. Let us assume that the closed linear
span of the sets $E^{ac}(\delta)\ran(C)$, $\delta \in \mathfrak B(\dR)$, coincides
with $\sH^{ac} = E^{ac}(\dR)\sH$. Let
\begin{equation*}
(\Phi E^{ac}(\delta)Cf)(\lambda) := \chi_{\delta}(\lambda)\sqrt{K(\lambda)}f, \quad \delta \in
\mathfrak B(\dR), \quad f \in \sH,
\end{equation*}
where $\chi_\delta(\cdot)$ denotes the characteristic function of $\delta \in \mathfrak B(\dR)$.
Obviously, we have
\begin{equation*}
\int \|(\Phi E^{ac}(\delta)Cf)(\lambda)\|^2_\sH d\lambda =
\int_\delta \|\sqrt{K(\lambda)}f\|^2_\sH d\lambda = \|E^{ac}(\delta)Cf\|^2_\sH.
\end{equation*}
Hence $\Phi: \sH^{ac} \longrightarrow L^2(\dR,d\lambda,\cH_\lambda)$ defines an
isometry from $\sH^{ac}$ into $L^2(\dR,d\lambda,\cH_\lambda)$. Let us show that
$\Phi$ is onto $L^2(\dR,d\lambda,\cH_\lambda)$. Let $g \in L^2(\dR,d\lambda,\cH_\lambda)$ such that
\begin{equation*}
0 = (\Phi E^{ac}(\delta)Cf,g) = \int_\delta\;(\sqrt{K(\lambda)}f,g(\lambda))_\sH d\lambda
\end{equation*}
for $f \in \sH^{ac}$, $\delta \in \mathfrak B(\dR)$. Since $\delta$ is arbitrary we find
$(\sqrt{K(\lambda)}f,g(\lambda))_\sH =0$ for a.e. $\lambda \in \dR$. Hence $g(\lambda) \perp
\cH_\lambda$ for a.e. $\lambda \in \dR$ which shows $g(\lambda) = 0$ for a.e. $\lambda \in
\dR$. Hence $\Phi$ is an isometry form $\sH^{ac}$ onto the subspace $L^2(\dR,d\lambda,\cH_\lambda)$.

Obviously, we have
\begin{equation*}
(\Phi E^{ac}(\delta) f)(\lambda)  = \chi_\delta(\lambda)(\Phi f)(\lambda), \quad \delta \in
\mathfrak B(\dR), \quad f \in \sH^{ac}.
\end{equation*}
Let $A$ be  a self-adjoint operator in $\sH$ and let $E_A(\cdot)$ be the corresponding spectral measure, i.e.
$A = \int_\dR \lambda \,dE_A(\lambda)$. Then $M\Phi = \Phi A^{ac}$ where $M$ is
the natural multiplication operator defined by
\begin{equation*}
\begin{split}
(Mf)(\lambda) & :=  \lambda f(\lambda),\\
f \in \dom(M) & :=  \{f \in
  L^2(\dR,d\lambda,\cH_\lambda: \lambda f(\lambda) \in L^2(\dR,d\lambda,\cH_\lambda\}.
\end{split}
\end{equation*}
If $\varphi(\cdot): \dR \longrightarrow \dR$ is a bounded Borel
function then $\varphi(M)\Phi = \Phi \varphi(A^{ac})$.

\begin{lemma}
 Let $A$, $E_A(\cdot)$, $C$ and $K(\lambda)$ be as above and assume that the absolutely continuous subspace $\sH^{ac}(A)$ satisfies the condition
 \begin{equation*}
 \sH^{ac}(A) =\text{\rm clsp}\,\bigl\{E_A^{ac}(\delta)\ran(C):\delta \in \mathfrak B(\dR)\bigr\}.
 \end{equation*}
Then the mapping
 \begin{equation*}
  E^{ac}(\delta)Cf\mapsto\chi_\delta(\lambda)\sqrt{K(\lambda)}f\quad\text{for a.e.}\,\,\,\lambda\in\dR,\quad f\in\sH,
 \end{equation*}
 onto the dense subspace $\text{\rm span}\,\{E_A^{ac}(\delta)\ran(C):\delta \in \mathfrak B(\dR)\}$ of $\sH^{ac}(A)$
admits a unique continuation to an isometric isomorphism from $\Phi:\sH^{ac}(A)\rightarrow L^2(\dR,d\lambda,\cH_\lambda)$ such that
\begin{equation*}
 (\Phi E_A^{ac}(\delta) g)(\lambda)=\chi_\delta(\lambda)(\Phi g)(\lambda),\quad g\in\sH^{ac}(A),
\end{equation*}
holds for any $\delta\in\mathfrak B(\dR)$.
\end{lemma}

Let us consider operator spectral integrals of the form
$\int_\dR dE^{ac}(\mu)Cf(\lambda)$, which are defined whenever $f(\cdot): \dR
\longrightarrow \sH$ is a Borel measurable function, cf. \cite[Section 5.2]{BW83}.
From \cite[Proposition 5.13]{BW83} we find that this integral exists if
and only if $\int_\dR \|\sqrt{K(\mu)}f(\mu)\|^2_\sH d\mu$ exists and is finite.
One verifies that
\begin{equation}\label{A.1}
\left(\Phi\int_\dR dE^{ac}(\mu)Cf(\mu)\right)(\lambda) =
\sqrt{K(\lambda)}f(\lambda).
\end{equation}

\subsection{Scattering}\label{A.III}

In the following let $A$ and $B$ be self-adjoint operators in $\sH$, let $J\in\cL(\sH)$ be a
bounded operator such that $J\dom A \subseteq \dom B$. If
$$V:= BJ - JA,\qquad \dom V := \dom A,$$
is closable and its closure
is a trace class operator then the wave operators
\begin{equation*}
W_\pm(A,B;J) := s-\lim_{t\to\pm\infty}e^{itB}Je^{-itA}P^{ac}(A)
\end{equation*}
exist, see \cite{BW83,1978Neidhardt,1978Pearson}. The scattering operator $S_J$ is defined by
\begin{equation*}
S_J(A,B) := W_+(A,B;J)^*W_-(A,B;J).
\end{equation*}
Usually the wave operators $W_\pm(B,A;J)$ and the scattering operator $S_J$ are not the quantities of main interest.
The objects one is more interested in are the wave operators $W_\pm(A,B) := W_\pm(A,B;I)$ and $S(A,B) := S_I(A,B)$. However, if the resolvent difference of $A$ and $B$ is compact, then the existence of
$W_\pm(B,A;J)$ yields the existence of $W_\pm(B,A)$ and both operators are related by
\begin{equation*}
W_\pm(A,B;J) = -W_\pm(A,B)(A-i)^{-2}.
\end{equation*}
In particular, this yields
\begin{equation}\label{eq:A3a}
S_J(A,B) = S(A,B)(I + A^2)^{-2}.
\end{equation}
The following theorem was announced in \cite[ Appendix A]{BMN09} but not proved there. Below the complete proof of theorem is given.
\begin{theorem}\label{att}
Let $A$ and $B$ be self-adjoint operators in the separable Hilbert space $\sH$ and suppose that the resolvent difference admits the factorization
\begin{equation}\label{eq:A3}
\sS_1(\sH) \ni (B-i)^{-1}-(A-i)^{-1}=\phi(A)CGC^*=QC^*,
\end{equation}
where $C\in\sS_2(\cH,\sH)$, $G\in\cL(\cH)$,  $\phi(\cdot): \dR \rightarrow \dR$ is a bounded continuous function and $Q = \phi(A)CG$. Assume that the condition
\begin{equation}
\sH^{ac}(A) = \text{\rm clsp}\,\bigl\{E_A^{ac}(\delta)\ran(C):\delta \in \mathfrak B(\dR)\bigr\}
\end{equation}
is satisfied and let $K(\lambda) =\frac{d}{d\lambda}C^*E_A((-\infty,\lambda))C$ and $\cH_\lambda = \overline{\ran(K(\lambda))}$ for a.e. $\lambda\in\dR$.
Then $L^2(\dR,d\lambda,\cH_\lambda)$ is a  spectral representation of $A^{ac}$ and the scattering matrix $\{S(A,B;\lambda)\}_{\lambda\in \dR}$ of the scattering system
$\{A,B\}$ has the representation
\begin{equation}\label{A.3}
S(A,B;\lambda) = I_{\cH_\lambda} + 2\pi i (1 + \lambda^2)^2\sqrt{K(\lambda)}Z(\lambda)\sqrt{K(\lambda)}
\end{equation}
for a.e. $\lambda \in \dR$, where
\begin{equation}\label{WScatA8}
Z(\lambda) = \frac{1}{\lambda+i}Q^*Q +\frac{\phi(\lambda)}{(\lambda+i)^2}G + \lim_{\varepsilon\to+0}Q^*(B-(\lambda+i\varepsilon))^{-1}Q
\end{equation}
and the limit of the last term on the right hand side exists in the Hilbert-Schmidt norm.
\end{theorem}
\begin{proof}
Consider the scattering operator $S_J(A,B) := W_+(A,B;J)^*W_-(A,B;J) : \sH^{ac}(A)
\longrightarrow \sH^{ac}(A)$,
where $J := -R_B(i)R_A(i)$ and
\begin{equation*}
R_B(\xi) := (B-\xi)^{-1}, \quad\quad R_A(\xi) :=
(A-\xi)^{-1}.
\end{equation*}
One easily checks that
\begin{equation*}
V := BJ - JA = (B-i)^{-1} - (A-i)^{-1} = \phi(A)CGC^*
\end{equation*}
where we have used the assumption \eqref{eq:A3}.
We note that the scattering operator commutes with $A$. From
\cite[Theorem 18.4]{BW83} one gets the representation
\begin{equation*}
 \begin{split}
 & S_J(A,B) - W_+(A,B;J)^*W_+(A,B;J) =\\
 &\qquad\qquad s-\lim_{\epsilon\to+0}w-\lim_{\tau\to+0}
\left\{-2\pi i\int_\dR dE^{ac}_A(\lambda)\,T(\tau;\lambda)\delta_\epsilon(A;\lambda)P^{ac}(A)\right\}
\end{split}
\end{equation*}
where
\begin{equation*}
T(\tau;\lambda) := J^*V -V^*R_B(\lambda+i\tau)V.
\end{equation*}
and
\begin{equation*}
\delta_\epsilon(A;\lambda) := \frac{1}{2\pi i}(R_A(\lambda+i\epsilon) -
R_A(\lambda-i\epsilon)) = \frac{1}{\pi}\frac{\epsilon}{(A-\lambda)^2 + \epsilon^2}.
\end{equation*}
If condition \eqref{eq:A3} is satisfied, then
\begin{equation*}
R_B(i) = R_A(i) + \phi(A)CGC^* = R_A(i) + QC^*
\end{equation*}
and we get
\begin{equation*}
\begin{split}
J^*V &= -R_A(-i)R_B(-i)V\\
& = -R_A(-i)CQ^*V - R_A(-i)^2V \\
& = -R_A(-i)CQ^*V - R_A(-i)^2\phi(A)CGC^*.
\end{split}
\end{equation*}
Hence we find
\begin{equation*}
T(\tau;\lambda) =
-\left(R_A(-i)CQ^*Q + R_A(-i)^2\phi(A)CG +
CQ^*R_B(\lambda+i\tau)Q\right)C^*.
\end{equation*}
Using \eqref{A.1} we get
\begin{equation*}
 \begin{split}
&  \left(\Phi\int_\dR dE^{ac}_A(\mu)T(\tau;\mu)\delta_\epsilon(A;\mu)P^{ac}(A)Ch\right)(\lambda) =\\
&\qquad\qquad\qquad -\sqrt{K(\lambda)}Z(\tau;\lambda)C^*\delta_\epsilon(A;\lambda)P^{ac}(A)Ch
 \end{split}
\end{equation*}
where
\begin{equation*}
Z(\tau;\lambda) := \frac{1}{\lambda+i}Q^*Q +
\frac{\phi(\lambda)}{(\lambda+i)^2}G + Q^*R_B(\lambda+i\tau)Q.
\end{equation*}
We note that the limit $Q^*R_B(\lambda+i0)Q := \lim_{\tau\to+0}Q^*R_B(\lambda+i\tau)Q$ exists in
the Hilbert-Schmidt norm. Hence the limit $Z(\lambda) :=
\lim_{\tau\to+0}Z(\tau;\lambda)$ exists in the operator norm and is given by
\begin{equation*}
Z(\lambda) = \frac{1}{\lambda+i}Q^*Q +
\frac{\phi(\lambda)}{(\lambda+i)^2}G + Q^*R_B(\lambda+i0)Q.
\end{equation*}
This gives
\begin{equation*}
\begin{split}
\Big(\Phi\Big\{\slim_{\epsilon\to+0}\wlim_{\tau\to+0}
\int_\dR dE^{ac}_A(\mu)&T(\tau;\mu)\delta_\epsilon(A;\mu)P^{ac}(A)Ch\Big\}\Big)(\lambda)\\
&= -\sqrt{K(\lambda)}Z(\lambda)K(\lambda)h\; .
\end{split}
\end{equation*}
By the compactness of $V$ we get that $W_+(B,A;J)^*W_+(B,A;J) = (I +
A^2)^{-2}$. Therefore we have
\begin{equation*}
\bigl(\Phi(W_+(A,B;J)^*W_+(A,B;J)\Phi^*f\bigr)(\lambda) = (1 + \lambda^2)^{-2}f(\lambda).
\end{equation*}
Hence $\Phi S_J(A,B) \Phi^*$ is equal to  a
multiplication operator with a measurable function $S_J(A,B;\lambda): \cH_\lambda
\longrightarrow \cH_\lambda$ given by
\begin{equation*}
S_J(A,B;\lambda) := (1 + \lambda^2)^{-2}I_{\cH_\lambda} +
2\pi i \sqrt{K(\lambda)}Z(\lambda)\sqrt{K(\lambda)}.
\end{equation*}
Using \eqref{eq:A3a} we find that
 $\Phi S(A,B) \Phi^*$ is a multiplication operator
induced by the measurable function $S(A,B;\lambda): \cH_\lambda \longrightarrow
\cH_\lambda$. Both functions $S_J(A,B;\lambda)$ and $S(A,B;\lambda)$ are related by
\begin{equation*}
S_J(A,B;\lambda) = S(A,B;\lambda)(1 + \lambda^2)^{-2}
\end{equation*}
which yields the representation \eqref{A.3}.
\end{proof}

\noindent
{\bf Acknowledgements.}
Jussi Behrndt gratefully acknowledges financial support by the Austrian
Science Fund (FWF), project P 25162-N26.
We are indebted  to  Ludvig Faddeev, Boris Pavlov, and Andrea Posilicano for useful discussions and remarks.
The preparation of the paper was supported by the European Research Council via ERC-2010-AdG no 267802 (``Analysis of Multiscale Systems Driven by Functionals'').

\end{document}